\documentclass[a4paper,11pt]{article}

\usepackage{fullpage}
\usepackage{color-edits}
\usepackage{natbib}

\addauthor{mb}{orange}
\addauthor{uf}{blue}
\addauthor{te}{red}

\usepackage{hyperref}
\hypersetup{colorlinks=true,linkcolor=blue,citecolor=blue,urlcolor=blue}

\usepackage{algorithm}

\usepackage{thm-restate}

\usepackage{amsmath}
\usepackage{amssymb}
\usepackage{xcolor}
\usepackage{algorithm,algorithmic}

\floatname{algorithm}{}

\def\proofname{proof}
\newtheorem{theorem}{Theorem}[section]
\newtheorem{corollary}[theorem]{Corollary}

\newtheorem{lemma}[theorem]{Lemma}
\newtheorem{claim}[theorem]{Claim}
\newtheorem{definition}[theorem]{Definition}
\newtheorem{proposition}[theorem]{Proposition}
\newtheorem{observation}[theorem]{Observation}
\newtheorem{remark}[theorem]{Remark}
\newtheorem{example}[theorem]{Example}

\newenvironment{proof}{\noindent\bf{Proof.}\rm}{\hfill$\blacksquare$\bigskip}

\newcommand{\mbc}[1]{\mbcomment{#1}}

\newcommand{\OLD}[1]{}

\newcommand{\bid}{r}
\newcommand{\vals}{\textbf{v}}
\newcommand{\bval}{z}
\newcommand{\budgets}{\textbf{b}}
\newcommand{\items}{{\cal{M}}}
\newcommand{\agents}{{\cal{N}}}
\newcommand{\prices}{{\cal{P}}}

\newcommand{\anyprice}[3]{AnyPrice(#1,#2,#3)}
\newcommand{\goodsetsiz}{{{\cal{G}}_i(z)}}
\newcommand{\MMSfull}[1]{MMS_{{#1}}(v_i,\items)}
\newcommand{\MMSnum}[1]{MMS_{#1}}
\newcommand{\MMSn}{\MMSnum{n}}
\newcommand{\MMSTwo}{\MMSnum{2}}
\newcommand{\disutility}{disutility}
\newcommand{\disutilities}{disutilities}

\newcommand{\anypricei}{\anyprice{b_i}{v_i}{\items}}
\newcommand{\anypriceic}{\anyprice{b_i}{c_i}{\items}}
\newcommand{\anypricevali}[1]{AnyPrice(#1)}
\newcommand{\anypricebi}{\anypricevali{b_i}}
\newcommand{\pqshare}{{\ell\mbox{-out-of-}d\mbox{-}Share}}
\newcommand{\proportionalbi}{Proportional(b_i)}
\newcommand{\pessimisticvali}[3]{Pessimistic(#1,#2,#3)}
\newcommand{\proportional}[3]{Proportional(#1,#2,#3)}
\newcommand{\pessimisticb}[1]{Pessimistic(#1)}
\newcommand{\pessimisticbi}{Pessimistic(b_i)}
\newcommand{\TPSbi}{TPS(b_i)}

\newcommand{\vect}[1]{\bold{#1}}

\newcommand{\truncated}[3]{TPS(#1,#2,#3)}
\newcommand{\truncatedi}{\truncated{b_i}{v_i}{\items}}
\newcommand{\FullVersion}[1]{#1}

\begin{document}

\title{Fair-Share Allocations for Agents with Arbitrary  Entitlements}
%\author{Moshe Babaioff, Tomer Ezra, Uriel Feige}
	\author{Moshe Babaioff\thanks{Microsoft Research ---  E-mail: \texttt{moshe@microsoft.com}.}, Tomer Ezra\thanks{Tel Aviv University ---  E-mail: \texttt{tomer.ezra@gmail.com}.}, Uriel Feige\thanks{Weizmann Institute ---  E-mail: \texttt{uriel.feige@weizmann.ac.il}. Part of the work was conducted at Microsoft Research, Herzeliya.}}

\maketitle

\begin{abstract}
	
%\mbc{first draft:} This paper considers the problem of fair division of indivisible goods among agents that have \emph{unequal entitlements} to the goods.  While for equal entitlements there are natural notions of fair shares as the maximin share, and it was shown [Kurokawa et al., 2018] %\cite{KurokawaPW18} 
%that for additive agents there is always an allocation that gives  every agent a constant fraction of her share, no such results are known for the case of unequal  entitlements. We first suggest a notion of share that we call the \emph{AnyPrice share (APS)}, a pricing-based notion which is appropriate for the unequal entitlements case.  We then show that for agents with additive valuations, there is always an allocation that gives every agent at least $3/5$ of her AnyPrice share. As far as we know, this is the first result showing that it is possible to give agents with unequal entitlements a  constant fraction of an appropriate share.  We then strengthen this result for the special case of agents with equal entitlements,  showing that there is always an allocation that gives every agent at least  $2/3$ of her APS. 

We consider the problem of fair allocation of indivisible goods to $n$ agents, with no transfers. When agents have equal entitlements, the well established notion of the maximin share (MMS) serves as {an attractive} fairness criterion, where to qualify as fair,  an allocation needs to give every agent at least a substantial fraction of her MMS. 
%For the case of additive valuations, it may not be possible to give every agent her MMS, but it is possible to give every agent at least a $\rho_{MMS}$-fraction of her MMS, for some $\rho_{MMS} > \frac{1}{2}$ (the current know bounds establish that $\rho_{MMS} \ge \frac{3}{4} + \Omega(\frac{1}{n})$).

In this paper we consider the case of arbitrary (unequal) entitlements. 
%We explain why we think that previous attempts to extend the MMS to this setting are not satisfactory. \mbc{maybe replace by: 
We explain shortcomings in previous attempts that extend the MMS to 
unequal
%arbitrary
entitlements. 
%, which calls for a new share definition.}
Our conceptual contribution is the introduction of a new notion of a share, the {\em AnyPrice share} (APS), that is appropriate for settings with arbitrary entitlements. Even for the equal entitlements case, this notion is new, and satisfies $APS \ge MMS$, where the inequality is sometimes strict. We present two equivalent definitions for the APS (one as a minimization problem, the other as a maximization problem), and provide comparisons between the APS and previous notions of fairness. 
%and present results concerning the complexity of computing the APS of an agent. 

Our main result concerns additive valuations and arbitrary  entitlements, for which we provide 
a polynomial-time algorithm 
%an algorithm 
that gives every agent at least a $\frac{3}{5}$-fraction of her APS.
This algorithm can also be viewed as providing strategies in 
a certain natural bidding game, and these strategies secure each agent at least a $\frac{3}{5}$-fraction of her APS.

%This implies that every agent also receives the same in every equilibrium of this game. 

%OLD: 
%Our main result concerns additive valuations (and unequal entitlements), for which we provide a pseudo-polynomial \ufc{will most likely be able to remove the {\em pseudo}}\mbc{is this due to the need to know the share?} time algorithm that gives every agent at least a $\frac{3}{5}$-fraction of her APS. Our algorithm is based on careful design of agent strategies for a certain allocation game that we introduce.\mbc{how about "a certain natural bidding game."}

\end{abstract}

%\tec{Open Question: can we argue that the "hardest case" is the one in which all agents have the same valuations? This is \emph{not} the case for MMS (can be ensured for equal entitlements with same valuations, but not with different valuations by \citet{KurokawaPW18}), so if it holds for APS it will be good property of APS.  }

\section{Introduction}\label{sec:intro}

%We consider fair allocation of indivisible items to agents with different entitlements to the goods. 

%\mbc{old:}
%There is a vast literature on fair division of indivisible goods when agents have equal entitlements (see, for example, {the books} \citep{moulin2004fair,brams1996fair}).
%Yet, equal entitlements is a special case of more general model in which agents have arbitrary (possibly unequal) entitlements to the goods \mbe{(sometimes referred to as weights)}. \mbc{If there are books talking about unequal entitlements it is good to cite them. We can change the approach from "there are only few papers that thought about unequal entitlements", to "many have thought and we have made significant progress on a studied problem" - in particular, it will fight the argument that unequal is not important }

There is a vast literature on the problem of fairly allocating indivisible items to agents that have preferences over the items
%, with no monitory transfers 
(see, for example, the classic books \citep{moulin2004fair,brams1996fair} and the recent survey \citep*{BouveretCM16}).
While most literature considers the important special case that agents have equal entitlements to the items, 
there is also a fast growing recent literature on the more general model in which agents may have arbitrary, possibly unequal,  entitlements to the items, see, for example, [\citealp{BabaioffNT2020}; \citealp{farhadi2019fair}; \citealp{chakraborty2019weighted}; \citealp{aziz2019weighted};  \citealp{aziz2020polynomial}]. The settings with arbitrary entitlements are sometimes referred to in the literature as having agents with ``asymmetric shares/weights" or ``weighted/unequal entitlements". % or ``unequal entitlements".  % or the agents' ``weights".

{The general model of arbitrary entitlements captures many settings that are not captured by the equal entitlement model. For example, it captures ``agents" that each  is actually a representative of a population, {and population sizes are different} 
%of different sizes 
(e.g., states in the US have different population sizes). It can also capture situations in which agents have made unequal 
%\tec{do we want to change here?}% MB: no, this is explaining why entitlements might be unequal 
prior investments in the goods (as partners that have invested different amounts in a partnership), and can also capture other asymmetries between agents (for additional motivation see examples presented, e.g., in \citep*{chakraborty2019weighted}).
The definition of a ``fair division" when agents have unequal entitlements} is more illusive than the equal entitlements case. In this paper we explain shortcomings of  recent work that aims to capture fairness in allocation of indivisible goods
when agents have unequal
% might have \tee{arbitrary} 
entitlements. We then suggest a new notion of a share and prove existence of fair allocations with respect to that notion.

%\mbc{say equal is a special case of unequal (due to seniority, population size, etc.)}
%In comparison, there are relatively few papers on the problem with unequal entitlements, \mbe{all rather recent}  \cite{BabaioffNT2020, farhadi2019fair,segal2020competitive}\mbc{see if there are more papers to CITE. weighted EF, chores}.
%None of these prior papers was able to ensure good fairness guarantees for agents with unequal entitlements, and the goal of this paper is to address this shortcoming.  
%%, the problem with unequal entitlements is much less studied. 
%Observe that unequal entitlements naturally arise in many settings: \mbc{expand}. 

A fundamental problem in this space is the allocation of a set $\items$ of $m$ indivisible {goods (items)} to $n$ agents with additive valuations over the items, with no monetary transfers. 
Every agent $i$ has an additive valuation $v_i$ and entitlement of $b_i>0$ to the items, assuming the entitlements sum up to $1$ ($\sum_i b_i=1$). We want to partition the items between the agents in a ``fair" way according to their entitlements, without using money.  
How should we interpret the meaning of ``entitlements" in this setting? 
If items were divisible then there is a natural interpretation -- each agent should get at the least as high value as if getting a $b_i$ fraction of each item, which for additive valuations is the same as her \emph{proportional} share 
$b_i \cdot v_i(\items)$, which we denote by $\proportionalbi$. 
Yet when items are indivisible, a proportional allocation (one  which gives every agent her proportional share)  might not exist, nor any finite approximation to it (consider two agents with equal entitlements over a single item -- one must get nothing). 

To address the problem of indivisibility with equal entitlements, the notion of \emph{maximin} share (MMS)  was suggested {by \citet{Budish11}}. This share is equal to the value an agent can ensure she gets if she partitions the items into $n$ bundles, and gets the worst one. 
Although for some instances it is not possible to give all agents their maximin share, % \cite{KurokawaPW18}, 
in every instance it is possible to give at least a large constant fraction of it \citep*{KurokawaPW18}. 
Yet, this notion is clearly inappropriate for agents with unequal % \tec{do we want to change here?} MB: no.
entitlements. Consider two agents with entitlements $9/10$ and $1/10$. With ten identical items, clearly the first agent should not get only five of the items (the maximum she can ensure by splitting the items into two and getting the smaller of the two bundles). 

% \mbc{I now realize that our examples that show that WMMS and WNSW are not good solutions for unequal entitlements have a common theme, I wonder if we should explicitly state it as a kind of "minimal requirement" we think should hold for any good fairness notion for arbitrary entitlements. Assume that and agent has more than one item for which she has positive value for. As her entitlement approaches 1, we want it to be the case that she should at least be deserved (in the sense of share as in WMMS) or allocated (in the sense of allocation based solution as WNSW) at least the value of her second best item. These prior solution (WMMS and WNSW) fail to give this basic property. }

To address this, \citet*{farhadi2019fair} have suggested the notion of \textit{Weighted-max-min-share} (WMMS).
{The WMMS of agent $i$ with valuation $v_i$ when the vector of entitlements is $(b_1,b_2,\ldots,b_n)$	is defined to be the maximal value $z$ such that there is an allocation that gives each agent $k$ at least $\frac{z\cdot b_k}{b_i}$ according to $v_i$ (see Appendix \ref{sec:missing-defs} for a formal definition).} 
In contrast to the maximin share, the WMMS is  influenced by the partition of the entitlements  among the other agents. This leads to scenarios in which an agent with arbitrary large entitlement might have a WMMS of $0$ even when there are many items. 
%For example, when $m=n-1<n$ the WMMS of every agent is 0, even if an agent has entitlement that is almost 1. 
For example, an agent that has entitlement of 99\% ($b_i=0.99$) for 100 items that she values the same, has a WMMS of 0, if the total number of agents is~$101$.
This suggests that the WMMS is sometimes too small.
%OLD: Thus the WMMS is sometimes too small. 
Moreover, that paper presents an impossibility result suggesting the WMMS is sometimes too large: there are instances where in every allocation some agent receives at most a $\frac{1}{n}$-fraction of her WMMS. Due to the above examples, we do not view the WMMS as a concept that is aligned well with our intended intuitive understanding of fair allocations when entitlements are unequal.
%\mbc{unclear:} \tec{I removed the rest of the sentence. The rest of the sentence tried to explain the scenarios in which it happens, but we don't need to be too specific}\mbc{maybe. but first I want to understand this. can you explain? can the WMMS be 0 even if there are many items and large budget, say of half? As this paper is our main "competition" for share definition, maybe we do need to give example when this WMMS is bad. } \tec{If there are at least $n$ items with positive values, then the WMMS is positive. The problem is that it might be arbitrarily small. for example if $n=m$ and the values are $1,1,1,\ldots,1,\epsilon$, and the entitlements are $1/2,1/2n,1/2n\ldots,1/2n$, then the WMMS for the agent with the half entitlement is $n\epsilon$ although "clearly" that he deserves at least 1.}\mbc{about you say "This leads to scenarios in which an agent with a large entitlement (larger than 1 over the number of items), might have a WMMS of $0$" but here the share is not 0. Is there an example with 0? }
%\tee{if the number of agents is larger than the number of items.}
%\tec{}

%Given the incompatibility of MMS to unequal entitlements, 
\citet*{BabaioffNT2020} %(henceforth BNT) 
have suggested a different generalization of MMS to the case of arbitrary entitlements: the \emph{$l$-out-of-$d$ share} of an agent is the value she can ensure herself by splitting the set of goods into $d$ bundles and getting the worse $l$ out of them. The \emph{Pessimistic} share of an agent with entitlement of $b_i$, which we denote\footnote{This share is implicit in \citet{BabaioffNT2020}, which did not give it a name.} by $\pessimisticbi$,  is defined to be the highest  $l$-out-of-$d$ share of all integers $l$ and $d$ such that $l/d\leq b_i$. 
For the equal entitlement case with $n$ agents, $b_i=1/n$ and the $\pessimisticbi$ share is the same as the MMS.  
For the example above, an agent with  entitlement of $9/10$ will split the items to ten singletons, and will get nine of them.  
% The share we later define will be larger than the $\pessimisticbi$ share and have better properties, all discussed later.  
%\mbc{list problems with this definition here? problems beyond additive, smaller than AnyPrice}.
%While CITE BNT have proven that in  
%Yet, this share has several drawbacks - for example, 
\citet{BabaioffNT2020} did not prove that there always exists an allocation  that  concurrently gives each agent some fraction of the pessimistic share. In this paper we present a new share that is always (weakly) larger, {and 
%although harder to approximate, we are able to 
prove} existence of a share-approximating allocation\footnote{{For some positive constant fraction, a \emph{share approximating allocation} gives every agent at least that fraction of her share.} }
{even} for our larger share.
%\footnote{An allocation %$A=(A_1,\ldots,A_n)$ of $\items$ \emph{approximates the MMS} if for some constant $\alpha>0$, 
%%it holds that $v_i(A_i)\geq \alpha \cdot MMS_n(\items)$ for every agent $i$. 
%%That is, 
%the value every agent gets under the allocation is at least $\alpha$ times her share.}
\citet{BabaioffNT2020} have related the pessimistic share to Competitive Equilibrium (CE), which we discuss next.  

A different approach for fair allocation is to try to come up with a fair allocation for each particular instance, instead of defining the share of each agent. 
For equal entitlements, {\citet{varian1973equity}} has suggested to use \emph{Competitive Equilibrium from Equal Income (CEEI)} in which each agent has  a budget of $1/n$ and the allocation is according to the Competitive Equilibrium allocation with these budgets. {For equal entitlements such a CE allocation is attractive as it ensures several nice fairness properties -- it is envy free and thus proportional (and gives each agent at least her MMS).}

More generally, the notion of \emph{Competitive Equilibrium (CE)} is defined for any entitlements, by treating  the entitlement of an agent as her budget, and looking for item prices that clear the market.  
CE allocations are intuitively fair as all agents have access to the same (anonymous) item prices, and in a CE every agent gets her favorite bundle within her budget. 
% Moreover, any CEEI has several nice formal fairness properties - it is envy free and thus proportional (and MMS). 
% Moreover, the notion of CE generalizes nicely to unequal entitlements, by giving each agent a budget that equals to her entitlement.  
For the case of arbitrary entitlements, \citet{BabaioffNT2020} have proven that the allocation of a CE ensures that an agent with budget $b_i$ will get her pessimistic share $\pessimisticbi$. {Interestingly, unlike the case of equal entitlements, with unequal entitlements a CE allocation need not give every agent her proportional share (see Proposition \ref{prop:TPS-CE}).}

While CE allocations have some nice fairness properties, unfortunately, even for equal entitlements, CE allocations might fail to exist, 
as one can immediately see from the simple case of two agents with equal budgets and a single item.\footnote{\label{footnote:generic}
	%In Section \ref{sec:related} we discuss attempts presented in prior works to overcome the CE existence problem through budget perturbations.
%	\ufc{This footnote is problematic, as it might be interpreted as if perturbed CE is the desired definition for fairness, and APS is a temporary replacement until generic CE can be proved to exist.} \mbc{I see your point. I thought we will mitigate it by saying that a CE  after the perturbation might be much smaller than the proportional share (which CE guarantees). } 
	\citet{BabaioffNT2020} have tried to overcome the CE existence problem through budget perturbations (see also \citep{segal2020competitive}). 
	We note that a CE for such perturbed budgets, even if it exists, no longer provides guarantees with respect to the original pessimistic shares. 
%	They proved generic existence of  CE in some two-agent additive settings, leaving open the problem of generic existence for all additive settings.  \citet{segal2020competitive} showed that there exists a non-additive setting for which even making the budgets generic does not ensure existence of a CE. 
}
Thus, as there are instances with no CE, we cannot rely on CE allocations as a method for fair division of indivisible items.  

An alternative instance-based approach is to allocate items using the allocation that maximizes the Nash social welfare (NSW), for the equal entitlements case, or the weighted Nash social welfare, for the arbitrary  entitlements case. Recall that for valuations $\vals=(v_1,v_2,\ldots,v_n)$ and entitlements $\budgets=(b_1,b_2,\ldots,b_n)$,  the \emph{weighted Nash social welfare} of an allocation $A=(A_1,A_2,\ldots, A_n)$ is $\prod_{i \in \agents} v_i(A_i)^{b_i}$. 
%\mbc{check -- I suspect we need to take some root} \tec{it is ok this way since the $b_i$ sums up to 1}
\citet*{CaragiannisKMPS19} showed that an allocation that maximizes NSW is always Pareto optimal, envy-free up to one good (EF1) and always gives each agent at least $O(\frac{1}{\sqrt{n}})$ fraction of her maximin share, showing that for equal entitlements maximizing NSW ensures a combination of some efficiency and fairness properties.

Unfortunately, for some simple unequal entitlements cases, maximizing the weighted NSW seems to produce unfair allocations. 
As an example consider an instance with three agents and three items, two of the items are good (value $1$ each) and the last item is unattractive. 
Assume the first two agents have a value 0 for it, while the last agent has a value $\varepsilon>0$. No matter how high the entitlement of the last agent is, in every allocation that maximizes the weighted NSW she will receive the unattractive item. This seems very unfair to the last agent, as although she has the largest entitlement (possibly by far), she {only gets her least} preferred item.
  
% \mbc{a good property of a share is that it does not decrease when the valuation increases. APS satisfies this. I suspect WNSW does not (but it is not a share, so it is not clear what we can say).  }

% \mbc{Uri,  envy-cycles not good. On second thought, lets not say what does not work, but just say what worked for us. Hence we can remove this comment.}

\subsection{Our Contribution}
Our first, conceptual contribution, is the definition of a new share, the \emph{AnyPrice share}, a share definition that is suitable not only for equal entitlements but also for arbitrary entitlements, and makes sense beyond additive valuations. 
We present two alternative definitions for this share, one based on prices (inspired by CE) and  another on item bundling ({inspired by MMS}).
%in the spirit of $l$-out-of-$d$ share).

The first definition, based on prices, defines the share of an agent $i$ to be the value she can guarantee herself whenever her budget is set to her entitlement  $b_i$ (when $\sum_i b_i=1$) and she buys her highest value affordable set when items are {adversarially} priced with a total price of $1$.  
Let $\prices=\{(p_1,p_2,\ldots,p_m) | p_j\geq 0\ \forall j\in\items,\ \  \sum_{j\in \items} p_j=1\}$ be the set of item-price vectors that total to $1$.   
The price-based definition of AnyPrice share (APS) is the following:  %\mbc{I have slightly revised this definition (Tomer, the prelim definition should be a copy of this one)}
\begin{restatable}{definition}{defanyprice}[AnyPrice share]
%\begin{definition}[AnyPrice share] 
	 \label{def:anyprice-prices}
	Consider a setting in which agent $i$ with valuation $v_i$ has entitlement $b_i$ to a set of indivisible items $\items$.
	The \emph{AnyPrice share (APS)} of agent $i$, denoted by $\anypricei$, is the value she can guarantee herself whenever the items in $\items$ are adversarially priced with a total price of $1$, and she picks her favorite affordable bundle: 
	$$\anypricei = \min_{(p_1,p_2,\ldots,p_m)\in \prices}\ \ \max_{S\subseteq \items} \left\{v_i(S) \Big | \sum_{j\in S} p_j\leq b_i\right\}$$
	
	When $\items$ and $v_i$ are clear from context we denote 	 the APS share of an agent $i$ with entitlement $b_i$ 
	by $\anypricebi$, instead of $\anypricei$.	
%\end{definition}
\end{restatable}

Observe that the AnyPrice share is well-defined for any instance, and is a definition of a share that 
depends on the agent valuation and entitlement, but not on the valuations and entitlements of the other agents. 
This can be contrasted with the value the agent gets under a CE allocation, a value that depends on the entire instance, and additionally, is not well defined as a CE might not exist. The AnyPrice share is defined using prices, but unlike CE, it does \emph{not} require finding prices that clear the market  (which might not exist) but rather use (worst case) prices that define a share for an agent based on her own  valuation and entitlement. %, independently of the valuation and entitlements of the others.     
From the definition it is clear that when a Competitive Equilibrium exists (with entitlements serving as budgets),
every agent gets at least her APS (as the CE price vector is one possible price vector in the minimization). 
% When a Competitive Equilibrium exists (with entitlements serving as budgets), each agents clearly gets her AnyPrice share. 
%Yet, a CE might not exist. 
%The definition of APS overcomes this non-existence problem as it does \emph{not} require finding prices that clear the market as in a CE (which might not exist) but rather use (worst case) prices to define a share for an agent based on her  valuation and entitlement, independently of the valuation and entitlements of the others.    

We further present an alternative definition that turns out useful in proving claims regarding APS.  
Using LP duality we show that the AnyPrice share (APS) of agent $i$ also equals to the maximum value $z$ such that there exists a distribution over bundles for which $i$ has a value of at least $z$, such that no item is allocated more than $b_i$ times in expectation (see Definition~\ref{def:anyprice-bundles} for a formal definition).  
% While the price-based definition is very useful in proving that the APS is not larger than some value, by simply presenting prices for which the agent cannot afford \mbc{this is problematic as it require checking all sets of high enough value.} any bundle of that value, the distribution-based definition  is very useful in proving the APS is at least as large as some value $z$ - by simply presenting weights that satisfy the definition for some value $z$. \mbc{Is there a formal way to argue that MMS does not have this nice property? }\mbc{remove this?: }
We show (Proposition \ref{prop:small-support}) that there is always a solution to this optimization problem with small support -- at most $m$ sets in the support.
While some of our proofs use the priced-based definition, others use the alternative definition, demonstrating the usefulness of both. 

The next example illustrates the notion of APS.  %\mbc{in this example the APS is the same as the proportional share. Maybe go back to Example 2.1} 
\begin{example}\label{example:APS-twice-Pessimistic}
	Consider a %\tec{grammerly says the "a" is not needed}\mbc{with "a" sounds better to me. lets keep.} 
	setting with five items and an agent $i$ with entitlement  $b_i=\frac{2}{5}$ and an additive valuation with values for the items being $2,1,1,1,0$. Her proportional share is $\frac{2}{5} (2+1+1+1) =2$. Her pessimistic share $\pessimisticbi$  is $1$, as partitioning the items to five bundles and getting the worst two bundles, or partitioning to four or three bundles and getting the worst bundle, gives a value of only~1.  
	Her APS is at least $2$ as for any pricing of the items, either the item of value $2$ has price at most $\frac{2}{5}$, or there are two items of value $1$ whose total price is at most $\frac{2}{5}$. 
	Her APS is not larger than 2, as by pricing each item at fifth of its value, the agent cannot afford a bundle of value more than $2$. 	
\end{example}
%In Example~\ref{example:APS-twice-Pessimistic} the APS is larger than the pessimistic share. 
The example illustrates that when valuations are additive, the APS can be larger than the pessimistic share (even 
twice as large).
% \footnote{The pessimistic share is the same as the MMS for agents with equal entitlements.\mbc{Uri, please move to an appropriate place in the body}}. 
This is no coincidence, as we show that the APS is always at least the pessimistic share, for any valuations (even non-additive). 

% OLD

The definition of APS does not rely on the valuations being additive, and is well defined for any valuations. 
Yet, in this paper we mainly focus on proving results for the additive case, leaving the study of more general valuations for future research. %\footnote{As discussed in Footnote \ref{footnote:super-add} \tec{we deleted this footnote}, ex-post shares are less appropriate for super-additive valuations, so the interesting case is that of sub-additive valuations. Along with additive valuations, unit-demand valuations are basic sub-additive valuations. We briefly discuss APS for unit-demand valuations in Section \ref{sec:unit-demand}, showing its applicability beyond additive valuations.}   
We start by making several observations regarding properties of the APS for additive valuations. First, we observe that it is at most the proportional share. Secondly, we show that while the APS is NP-hard to compute (like the MMS), it has a pseudo-polynomial time algorithm, while the MMS does not.
%With the definition of the AnyPrice share at hand, 
We then turn to prove results about the approximation to the APS that can be given to all agents concurrently. % , when valuations are additive.   
We observe that for two agents, it is always possible to give both of them their APS, and turn to study the more involved  case of more than two agents. 
Our main result is that for any number of agents with additive valuations and arbitrary entitlements, it is always possible to give each agent at least a $\frac{3}{5}$ fraction of her AnyPrice share, and moreover, 
at least  $\frac{1}{2-b_i}$ fraction of the APS (which is more than $\frac{3}{5}$ if her entitlement is larger than $1/3$). Note that the fraction $\frac{1}{2-b_i}$  goes to $1$ as $b_i$ grows to $1$. 
\begin{theorem}\label{thm:intro-unequal}
	For any instance with $n$ agents with additive valuations \vals\ and entitlements \budgets,  there exists an allocation in which every agent gets at least 
	$\max \{\frac{3}{5}, \frac{1}{2-b_i} \}$ fraction of her AnyPrice share.
%	$60\%$ %$3/5$ fraction 
%	of her AnyPrice share, and at least $\frac{1}{2-b_i}$ fraction of her AnyPrice share (which is more than $60\%$ for $b_i>1/3$).  
	Moreover, such an allocation can be computed in polynomial time.% \ufc{Not happy about saying pseudo-polynomial time, when I believe that it should be polynomial time.}
%	\footnote{\mbe{The algorithm is actually polynomial when the value of the APS is given. Yet, computing the value of the APS is hard, and we only have a pseudo-polynomial time algorithm to compute it.}}
	
%	\mbc{an alternative version that is stronger - which do we want? we can replace the above by the below or alternatively, we can state  the second version in the text after the theorem.}
	
%	There exist an allocation mechanism for agents with additive valuations and unequal entitlements,  such that in that mechanism % for any additive valuations \vals\ and entitlements \budgets\ 
%	there exists a strategy for every agent that guarantee she is allocated a set of value  at least $60\%$ %$3/5$ fraction 
%	of her AnyPrice share, and also at least $\frac{1}{2-b_i}$ fraction of her APS (which is more than $60\%$ for $b_i>1/3$), for any strategies the others play.  	Moreover, that mechanism runs in polynomial time. \mbc{move to a remark after the theorem.}
	
\end{theorem}
As far as we know, this %result 
is the first positive result that guarantees a constant fraction of \emph{any} share to agents with arbitrary entitlements. (See Section~\ref{sec:discussion} for clarifications regarding this statement.) 
Recall that for equal entitlements, with at least three agents it is not possible to give all agents their MMS \citep*{KurokawaPW18}.
As the APS is at least as large as the MMS,
%\mbc{say APS cannot be given due to KPW18 instead of:}
%To complement this positive result we show that for any\footnote{For the case of two agents we show (REF PROP) that there is always an allocation that gives both agents their APS.} $n\geq 3$ there are instances with $n$ agents with \mbc{add ''identical''? This works against us as MMS can be given} additive valuations and equal entitlements for which there is no allocation that gives every agent at least $99\%$ of her AnyPrice share (REF PROP). 
%Thus, 
we thus conclude that when there are more than two agents  one cannot hope to give all agents their APS share exactly, even in the equal entitlements case, and constant approximation is the best we can hope for (we leave open the problem of finding the exact constant).  
%MB: we have already said this:
%In contrast to the need for approximation for case of more than two agents, we show that in the case of exactly two agent it is actually possible to give both their APS. 

The proof of the theorem is based on 
analyzing a natural ``bidding game", %\ufc{make sure we have a consistent name for this game} 
in which entitlements serve as budgets, and at each round the highest bidder wins, taking any items she wants and paying her bid for each item taken. We show that in this game, no matter how other players bid, an agent has a bidding strategy that gives her at least	$\max \{\frac{3}{5}, \frac{1}{2-b_i} \}$ fraction of her APS. % AnyPrice share.  
If the bidding game is used as the mechanism for allocating the items, then in every equilibrium (e.g., of the full-information game), every agent receives at least $\max \{\frac{3}{5}, \frac{1}{2-b_i} \}$ fraction of her AnyPrice share.

%We further present a procedure that ensures that every agent will get at least half her APS, yet gives better guarantees for agents with large entitlements \mbc{$b_i>1/3$ the guarantee is better}(fraction that might be larger than $60\%$ - it goes to $100\%$ as the entitlement goes to $1$):
%\begin{theorem}
%	For any additive valuations \vals\ and entitlements \budgets\  there exists an allocation in which every agent $i$ gets at least $\frac{1}{2-b_i}$ fraction 
%	of her AnyPrice share. Moreover, such an allocation can be computed in polynomial time. \mbc{unify the two theorem and say you get the maximum of $3/5-APS$ and the fraction of TPS. }
%\end{theorem}
%%The proof of the theorem actually prove the approximation with respect to another share 
%\mbc{Uri suggested to mention the Truncated proportional share (TPS) and the tightness of this result for TPS for equal entitlements. I think it will only be confusing here and suggest not doing that in the introduction, only in the body of the paper. }

Finally, we consider the APS in the special case of equal entitlements. 
{While for equal entitlements the pessimistic share coincides with the MMS, the APS does not.} 
% OLD: this was hard to parse:
% Unlike the pessimistic share, in this special case the APS does not coincide with the MMS. 
Rather, the APS is 
at least as large as the MMS, and sometimes strictly larger. Hence existing algorithms that guarantee every agent a positive fraction of the MMS need not guarantee the same fraction with respect to the APS. Nevertheless, we prove that in the equal entitlements case an algorithm of~\citep{BK20} gives every agent at least a $\frac{2}{3}$ fraction of her APS (which is better than the  $\frac{3}{5}$ fraction that we get for the arbitrary entitlements case). % , \mbe{which in particular implies that every agent at least a $\frac{2}{3}$ of her MMS.}  
%\vspace{-3mm}
\begin{theorem}
		For any instance with $n$ agents with additive valuations \vals\  and  equal entitlements ($b_i=1/n$ for every $i$), there exists an allocation in which every agent $i$ gets at least a
% $(2/3)\anypricei$. 
	$\frac{2n}{3n-1}$ 
	fraction of her AnyPrice share. Moreover, such an allocation can be computed in polynomial time. 
\end{theorem}
As mentioned above, giving every agent her APS is not possible, even for equal entitlements. Finding the best approximation is open even for the case of equal entitlements. 

We remark that our results are actually somewhat stronger than the theorems stated above. 
Some bounds are proven with respect to a notion of a share called the \emph{Truncated Proportional Share} (TPS) (see Section~\ref{sec:tps}) which for additive valuation is at least as large as the APS, and sometimes larger. 
\medskip 

The % rest of the 
paper is organized as follows. We first present some additional related work in Section \ref{sec:related}. 
We then present the model and some prior fairness definitions in Section \ref{sec:model}. 
We introduce the notion of AnyPrice share and discuss its properties in Section \ref{sec:aps}. 
We then present our main result for APS with arbitrary entitlements in Section \ref{sec:allocation-game}, and our result for the special case of equal entitlements in Section \ref{sec:greedy-efx}.   
All proofs missing from the body of the paper appear in the appendix. The main part of this paper considers only allocation of goods. Extensions of the AnyPrice share to chores and to mixed manna are presented in Appendix~\ref{app:chores}.

\subsection{Additional Related Work}
\label{sec:related}

% \mbc{I have removed the CEEI discussion and the generic issue and have changed Footnote \ref{footnote:generic}.}

%As CE might fail to exists, \cite{Budish11} has suggested to use ``CE from almost equal income" to get fair allocation for agents with equal entitlements,  by first slightly perturbing the budget (=entitlement) of $1/n$ of each agent, and then looking for a CE after the perturbations. Budish has  showed that an approximate CE (which does not have to exactly clear the market, only approximately so) always exists, and every such approximate CE guarantees an existence of a maximin allocation with respect to $n+1$ agents \mbc{Tomer check if this is known. maybe still open?}. 
%Note that such an allocation might not approximate the maximin share with respect to $n$ agents at all (as in the case that there are only $n$ items).

%Similarly, \citet{BabaioffNT2020} have tried to circumvent the competitive equilibrium existence problem, but for unequal entitlements, by making the  budgets generic (i.e., perturbing them slightly). The authors proves generic existence of exact CE in some two-agent additive settings, but left open the problem of generic existence for all additive settings. When moving to non-additive settings,  \citet{segal2020competitive} showed that there exists a non-additive setting for which  even making the budgets generic is not enough to ensure existence of a CE. 
%%Thus, as CE might not exist, even generically, we can not rely on it as a method for fair division of indivisible items.  

One approach taken for fairness when agents have equal entitlement is to require the allocation to be envy free (EF). With indivisible items, such envy-free allocation might not exist, and relaxations as envy-free up to one good (EF1) were considered~\citep{Budish11}. {EF1 allocations were shown to exist in~\citep*{LiptonMMS04}.  Clearly, when agents have different entitlements, naively using this notion is inappropriate.  \citet*{chakraborty2019weighted} have extended envy-freeness to settings with agents having arbitrary entitlements (introducing notions of strong and weak weighted EF1, WEF1), and proved existence result for Pareto optimal allocations satisfying these notions. We remark that even in the equal entitlements settings, 
{there are instances with an EF1 allocation (that is Pareto optimal) that has an agent that only gets an $O(\frac{1}{n})$ fraction of her MMS. }
%OLD: We remark that even in the equal entitlements settings, EF1 allocations (even if they are Pareto optimal)  may give  agents only a $O(\frac{1}{n})$ fraction of their MMS. 
Hence WEF1 allocations do not guarantee agents a constant fraction of their APS. }

While we consider fair allocation of indivisible goods to agents with arbitrary entitlement, \citet*{aziz2019weighted} considered the related problem of fair assignment of indivisible chores (negative value items).
\citet*{aziz2020polynomial} considered the case that every item value might have mixed signs (a good for some agents, a bad for others) when agents have arbitrary entitlements,   {and show that there always exists an allocation that is not Pareto dominated by any fractional allocation and also gives each agent her proportional share up to one item. Moreover, they show that such an allocation can be found in polynomial time.} 
% MB: I think we do not need these details: 
% for the case of two agents, or multiple agents with binary valuations.}

%\tee{\cite{chakraborty2019weighted} extended the definition of EF1 to WEF1\mbc{what is WEF1?}\tec{there are two definitions, strong WEF1 (when we removing one item from each of the bundles of the others will resolve the envy, or weak (which is either remove one or add one item to me)} for the case of unequal entitlements. They showed an existence result for PO WEF1 allocation.}  \tec{they also showed that fMNW is weak WEF1.}

%\tec{\cite{cseh2020complexity}? cake cutting do we want it? } \mbc{how are these related? Answer Unequal entitlements, but it is not indivisible}

Our result for agents with equal entitlements shows that we can give each agent $2/3$ of her APS. As the APS is at least the MMS, this result relates to the literature on MMS and approximate MMS allocations. 
{The maximin share (MMS) was first introduced by \citet{Budish11} as a relaxation of the proportional share. \citet*{KurokawaPW18} showed that for agents with additive valuations, an allocation that gives each agent her MMS may not exists.	A series of papers %\ufc{Reorder references in order of first appearance of conference version -- seems like compilation automatically rearranges the order. Check also in BoBW paper that this reordering does not happen}
[\citealp{KurokawaPW18}; \citealp{amanatidis2017approximation}; \citealp{BK20}; \citealp{GhodsiHSSY18}; \citealp{garg2019approximating}; \citealp{GT20}] considered the best fraction of the MMS that can be concurrently guaranteed to all agents, and the current state of the art (for additive valuations) is $\frac{3}{4} +\Omega(\frac{1}{n}) $ fraction of the MMS.}

%\mbc{Tomer please add references for  bidding games }

{In our proofs we consider a natural bidding game in which agents have budgets, and bid for the right to select items (as many items as their budget allows, given the bid). We design and analyse strategies for the agents that guarantee that they obtain a high value in the game. Our bidding game is a variant of the \textit{Poorman game} introduced by \citet*{lazarus1995richman} where the agents start with budgets, and in each turn, the agents perform a first price auction in order to determine who selects the next action of the game.
Variants of these games were later studied in the context of  allocating items in multiple papers, for example, \citep{huang2012sequential,meir2018bidding}. 
}

\section{Model} % and Preliminaries}
\label{sec:model}
%\subsection{Model}
%\begin{itemize}
%	\item model
%	\item additive/subadditive valuations
%\end{itemize}

We consider fair allocation of a set  $\items$ of $m$ indivisible goods (items)  to a set $\agents$ of $n$ agents, with  entitlements to the items that are potentially different. 
Each agent $i$ is also associated with an \emph{entitlement} %{(or budget)} 
$b_i > 0$, and the sum of entitlements of the agents is $1$ (i.e., $\sum_{i\in \agents} b_i = 1$).
We use $\budgets= (b_1,\ldots,b_n)$ to denote the vector of agents' entitlements, % (budgets).} 
{and say that agents have \emph{arbitrary entitlements}. 
If $b_i=1/n$ for every $i$ then we say that agents have  \emph{equal entitlements}. }
In this work we consider deterministic allocations. An \emph{allocation} $A$ is a partition of the items to % collection of 
$n$ disjoint bundles  $A_1, \ldots, A_n$ (some of which might be empty), where $A_i \subseteq {\items}$ for every $i\in \agents$.\footnote{One can extend the terminology and notation so that an allocation may leave some of the items not allocated. This extension is not needed in this paper, as we only consider goods (and not bads), and allocating additional items to an agent cannot hurt her.} %\mbc{do we want to allow items to be left unallocated?} 
%A \emph{randomized allocation} is a distribution over deterministic allocations. 
%We consider setting with $n$ agents and a set of items $M$ of size $m$. Each agent $i$ has a \emph{valuation function} $v_i$ over sets of items. 
%An \emph{allocation} $A=(A_1,\ldots,A_n)$ is a list of sets of items, where $A_i$ is the subset of items agent $i$ is allocated. It satisfies that $A_i \cap A_j = \emptyset$ for every $i\neq j$. 
%We use the term \emph{randomized allocation} to refer to a distribution over allocations.   
 %\te{(old delete) and postpone the treatment of randomized allocations to later (as they require also to model the attitude of agents towards risk).}	
{We denote the set of all allocations of $\items$ to the $n$ agents by $\mathcal{A}$.} 
A valuation $v$ {over the set of goods $\items$} is called \textit{additive} if each item $j\in \items$ is associated with a value $v(j)\geq 0$, and the value of a set of items $S \subseteq \items$ is $v(S)=\sum_{j\in S}v(j)$.
A valuation $v$ is called \textit{subadditive} if for every two sets of items $S,T\subseteq \items$, it holds that $v(S)+v(T)\geq v(S\cup T)$. 
In this work we assume % for simplicity 
that all valuations are monotone (i.e., $v(S) \geq v(T)$ for every $T\subseteq S$), and normalized (i.e., $v(\emptyset)=0$.)
{Let $v_i(\cdot)$ denote the valuation of agent $i$, and let $\vals =(v_1,\ldots,v_n)$ denote the vector of agents' valuations.}
%Our main focus is the case of additive valuations, but some definitions and results extend to the broader family of subadditive valuations.
We assume that the valuation functions of the agents as well as the entitlements are known to the social planner, and that there are no transfers (no money involved). 

We formally define some shares from the literature. 
{The \emph{proportional share} of agent $i$ with valuation $v_i$ and entitlement $b_i$ for items $\items$ is 
$\proportional{b_i}{v_i}{\items}= b_i\cdot v_i(\items)$. 
When $v_i$ and $\items$ are clear from context we omit them from the notation and denote this share by $\proportionalbi$.
}
We next define the maximin share (MMS), which is defined for the equal entitlements case. 

\begin{definition}[maximin Share (MMS)]
	\label{def:pes}
	The \emph{maximin share (MMS)} of agent $i$ with valuation $v_i$ over $\items$, when there are $n$ agents, which we denote  by $\MMSfull{n}$, 
	is defined to be the highest  value she can ensure herself by splitting the goods to $n$ bundles and getting the worst one. Formally: 
	$$\MMSfull{n} =   \max_{(A_1,\ldots,A_n)\in \mathcal{A}} \min_{j \in [n] } \left\{v_i(A_j)	%\Big | (A_1,\ldots,A_n) \mbox{ is a partition of } \items 
	\right\},$$

\end{definition}
When $v_i$ and $\items$ are clear from context, we denote this share by $\MMSn$, and when $n$ is also clear from context, we simply use the term MMS.

%{The \emph{maximin share (MMS)} of agent $i$ with valuation $v_i$ over the set of items $\items$ when there are $n$ agents, which we denote  by $\MMSfull{n}$, is the maximum over all partitions of $\items$ to $n$ bundles, of the value of the least valued bundle with respect to $v_i$:
%	$$\MMSfull{n} =   \max_{(A_1,\ldots,A_n)} \min_{j \in [n] } \left\{v_i(A_j) \Big | (A_1,\ldots,A_n) \mbox{ is a partition of } \items \right\},$$
%When $v_i$ and $\items$ are clear from context, we denote it by $\MMSn$, and when $n$ is also clear from context, we simply use the term MMS.}

In order to formalize the notion of pessimistic share, we first present the definition of  the $\ell$-out-of-$d$ share, a generalization of the MMS.
\begin{definition}[$\ell$-out-of-$d$ share]
	\label{def:lds}
%	\tec{add definition}\mbc{This is called $l$-out-of-$d$ share in introduction, following \cite{BabaioffNT2020}: "the \emph{$l$-out-of-$d$ share} of an agent is the value she can ensure herself by spiting the goods to $d$ bundles and getting the worse $l$ out of them. }
	The \emph{$\ell$-out-of-$d$ share} of agent $i$ with valuation $v_i$ over $\items$ is the value she can ensure herself by splitting the goods to $d$ bundles and getting the worse $\ell$ out of them. Formally: 
	
		$$\pqshare(v_i,\items) =   \max_{(P_1,\ldots,P_d)\in\mathcal{Z}  } \min_{J \subseteq [d]:~ |J| = \ell } \left\{v_i(\cup_{j\in J}P_j) % \Big | (P_1,\ldots,P_d) \mbox{ is a partition of } \items 
		\right\},$$
	where $\mathcal{Z}$ is the set of all partitions of $\items$ into $d$ disjoint sets. % \ufc{This explanation makes the conditioning in the math notation redundant}
	%\tec{Do we need the definitions or the text in the intro is enough? or do we want to change it in the intro to a definition environment?}  \mbc{repeating this definition in the appendix is fine. I would also not change the intro. I would change it here to be more formal (with explicit sets, partition and value notations), and have the current text as the informal description before or after the definition.  }
\end{definition}

\begin{definition}[Pessimistic Share]
	\label{def:pes}
	The \emph{Pessimistic share} of agent $i$ with valuation $v_i$ over $\items$, and entitlement $b_i$, denoted by $\pessimisticvali{b_i}{v_i}{\items}$, is defined to be the highest  $\ell$-out-of-$d$ share for all integers $\ell$ and $d$ such that $\frac{\ell}{d} \leq b_i$. Formally: 
	
	{
		$$\pessimisticvali{b_i}{v_i}{\items}
		 =   \max_{\ell,d\in \mathbb{N}: ~\frac{\ell}{d}\leq b_i} \Big\{\pqshare(v_i,\items)\Big\}.$$}
\end{definition}
When $v_i$ and $\items$ are clear from context we omit them from the notation and denote this share by $\pessimisticbi$.

\section{The AnyPrice Share (APS)}
\label{sec:aps}
In this section we introduce a new share definition that is well defined for arbitrary entitlements (even for non-additive valuations), presents some of its properties and compare it to previous suggested shares.

\subsection{The Definition of the AnyPrice Share}

We present two definitions of $\anypricei$, the AnyPrice share (APS) of an agent $i$ that has valuation $v_i$ and entitlement $b_i$ (when $\sum_i b_i=1$) to items $\items$. % , when the goods are $\items$. 
The definitions are general and well defined for every valuations, not requiring that the valuations are necessarily additive.
{Our first definition is based on prices and was presented as Definition \ref{def:anyprice-prices} in Section \ref{sec:intro}.
}
\FullVersion{
%\mbc{we are half page short. Might consider removing this repeat of the  first definition after all (not in the full version, only the EC one): } 
The first definition, based on prices, defines the share of an agent $i$ to be the value she can guarantee herself whenever her budget is $b_i$ and she buys her highest value affordable set when items are adversarially priced with total price of $1$.  
Let $\prices=\{(p_1,p_2,\ldots,p_m) | p_j\geq 0\ \forall j\in\items,\ \  \sum_{j\in \items} p_j=1\}$ be the set of item-price vectors that total to $1$.   
We restate the price-based definition of APS:  %\mbc{Tomer, please copy from the intro.}
\defanyprice*
%\begin{definition}[AnyPrice share]  \label{def:anyprice-prices}
%	Consider a setting in which agent $i$ with valuation $v_i$ has entitlement $b_i$ to a set of indivisible items $\items$.
%	The \emph{AnyPrice share (APS)} of $i$, denoted by $\anypricei$, is 
	% For agent $i$ with valuation $v_i$ and entitlement $b_i$, the \emph{AnyPrice share (APS)} $\anypricei$ of $i$ is 
%	the value she can guarantee herself whenever items are (adversarially) priced with total price of $1$ and she picks her favorite affordable bundle: 
	%$$\anypricei = \min_{(p_1,p_2,\ldots,p_m)\in \prices}\ \ \max_{S\subseteq \items} \left\{v_i(S) \Big | \sum_{j\in S} p_j\leq b_i\right\}$$	
%\end{definition}
%\tee{When $v_i,\items$ are clear from context, we use $\anypricebi = \anypricei$, and when $v_i,\items,b_i$ are clear from context, we use $APS=\anypricei$.} 
}

We next present the second definition of the AnyPrice share, using item packing (in Proposition \ref{prop:anyprice-eqe} we 
%use LP duality to 
show that the two definitions are equivalent).

\begin{definition}[AnyPrice share, dual definition] 	\label{def:anyprice-bundles} 	
	Consider a setting in which agent $i$ with valuation $v_i$ has entitlement $b_i$ to a set of indivisible items $\items$.
	The \emph{AnyPrice share (APS)} of $i$, denoted by $\anypricei$, is 
	%For agent $i$ with valuation $v_i$ and entitlement $b_i$, the \emph{AnyPrice share (APS)} $\anypricei$ of $i$ is 
	the maximum value $z$ she can get 
	by coming up with non-negative
	%a family of $t$  (possibly overlapping) sets $\{S_k\}_{k=1}^t$ with 
	weights $\{\lambda_T\}_{T\subseteq \items}$ that total to $1$ (a distribution over sets), 
	such that any set $T$ of value below $z$ has a weight of zero, and
	any item appears in sets of total weight at most $b_i$:
	$$ \anypricei = \max z $$ subject to the following set of constraints being feasible for $z$: 
	\begin{itemize}
		%\item $t\in \mathbb{N}$ \mbc{having $t$ is annoying. Maybe just split to m bundles, allowing empty ones, and claiming later than allowing more bundles does not help (os this indeed the case)? }
		\item $\sum_{T\subseteq \items} \lambda_T=1$
		\item  $\lambda_T\geq 0\ \forall {T\subseteq \items}$
		\item  $\lambda_T= 0\ \forall {T\subseteq \items}$ s.t. $v_i(T)<z$ % \tec{This is not a linear constraint}
		\item $\sum_{T: j\in T} \lambda_T \leq b_i \ \forall j\in \items	$
	\end{itemize}
% 	\mbc{should we write the above with less space?}

	% REVISE:	$$\anypricei =  \sup_{t\in \mathbb{N}} \min_{i=1}^t \max_{(S_1,\ldots,S_t): S_k\subseteq \items \forall k} \sup_{(\lambda_1,\ldots,\lambda_t)} \{v_i(S) | \sum_{1}^{t} \lambda_k=1 \} $$
	
	%	For agent $i$ with valuation $v_i$ and entitlement $b_i$, the \emph{AnyPrice share (APS)} $\anypricei$ of $i$ is the maximum over all pairs of integers $p,q$ such that $\frac{p}{q}\leq b_i$, and maximum over all $q$ (not necessarily disjoint) sets $\vect{S} = (S_1,\ldots,S_q)$, such that every item appears in at most $p$ of the bundles among $\vect{S}$ of $\min_{i} v(S_i)$.\mbc{This is WRONG. revise.}
\end{definition}
An equivalent set of constraints that will sometimes be convenient to use is the following. 
Let $\goodsetsiz$ be the family of sets $T\subseteq \items$ such that $v_i(T)\geq z$. 
The set of constraints now becomes $\sum_{T\in \goodsetsiz} \lambda_T=1$,
$\lambda_T\geq 0\ \forall {T\in \goodsetsiz}$,  and $\sum_{T\in \goodsetsiz: j\in T} \lambda_T \leq b_i \ \forall j\in \items	$.
{For the case of equal entitlements this definition can be viewed as a fractional relaxation of the definition of MMS (and hence implying that the APS is at least as large as the MMS). Specifically, one can get the definition of the MMS (with $n$ agents, hence $b_i=\frac{1}{n}$) by 
%adding to Definition~\ref{def:anyprice-bundles} the requirement that bundles $T$ of positive weight $\lambda_T > 0$ are disjoint. 
%This is equivalent to 
restricting every weight $\lambda_T$ to be either $b_i$ or~0 for any non-empty set.
As the total weight is $1$, there must be at most $n$ bundles with non-zero weight. 
The constraint that every item belongs to bundles with total weight of at most $b_i$ implies that all %the $n$ 
non-empty positive-weight bundles are disjoint. 
%\mbc{There is a small problem with the empty set. If $n=3$ but there is only one item, the empty set has weight $2/3$ (or there should be two ''different'' empty sets, each with weight of $1/3$.) This case is not formally captured, I think (the definition allows for only one empty set - each set is different).}
The APS on the other hand, relaxes the constraint $\lambda_T \in \{0,b_i\}$ for non-empty sets (which can be thought of as an integer constraint of the form $\lambda_T \in \{0,1\} \cdot b_i$) to the fractional constraint $0 \le \lambda_T \le b_i$, and maintains the constraint that every item belongs to bundles with total weight of at most $b_i$.}

While the price-based definition (Definition \ref{def:anyprice-prices}) is very useful in proving that the APS is not larger than some value, by simply presenting prices for which the agent cannot afford any bundle of that value, the weight-based definition (Definition \ref{def:anyprice-bundles}) is very useful in proving the APS is at least as large as some value $z$ -- by simply presenting weights that satisfy the definition for some value $z$. 
Indeed, let us go back to  Example \ref{example:APS-twice-Pessimistic} to illustrate this point. 
%To illustrate the usage of this alternative definition, let us go back to Example \ref{example:APS-twice-Pessimistic}. 
Recall that the example considers a setting with five items and an agent $i$ with entitlement  $b_i=2/5$ and an additive valuation with values for the items being $2,1,1,1,0$.
Using the price-based definition one could easily see that the APS of the agent is at most $2$, by pricing each item at fifth of its value. 
To see that the APS is at least $2$ using the alternative definition, 
%the same with the other definition, 
consider a distribution over four sets of value $2$, one set contains the item of value $2$ and weight  $2/5$, and three sets each contains a different pair  of items of value $1$, and each of the sets has a weight of $1/5$. Each of the four sets has a value of $2$, and the total weight of the sets that contain any item is at most $2/5$, thus showing the APS is at least $2$.

We start with several simple observations. We first {show, using linear programming duality,} that Definitions~\ref{def:anyprice-bundles} and \ref{def:anyprice-prices} are indeed equivalent. 
%\begin{proposition}
\begin{restatable}{proposition}{propDefEq}		
	\label{prop:anyprice-eqe}
	The two shares defined in Definition~\ref{def:anyprice-bundles} and Definition~\ref{def:anyprice-prices} are equivalent.
\end{restatable}
%\end{proposition}

We further observe that there is always a solution with small support to the feasibility problem of Definition~\ref{def:anyprice-bundles} -- at most $m$ sets in the support. 
\begin{restatable}{observation}{propSmallSupport}		
	%\begin{proposition}
	\label{prop:small-support}
	For any valuation $v_i$ and entitlement $b_i$ there is a solution to the optimization problem presented in 
	Definition~\ref{def:anyprice-bundles} in which there are at most $m$ non-zero weights.
	%\end{proposition}
\end{restatable}

We observe  that in any Competitive Equilibrium\footnote{For a formal definition of CE see Appendix \ref{sec:missing-defs}.} in which the entitlement of an agent is her budget,
every agent gets her AnyPrice share. This follows immediately from Definition \ref{def:anyprice-prices} as the APS of an agent is the value the agent  can secure for \emph{any} possible prices, and in particular, at least that value  can be secured for the CE prices.  
\begin{observation}
	\label{obs:anyPrice-CE}
	Assume that the pair $(A,p)$ is an allocation-prices pair that forms a Competitive Equilibrium  when agents have valuations $\vals$ and budgets (=entitlements) $\budgets$.
	Then the allocation $A$ is an APS allocation, that is, for every agent $i$ it holds that $v_i(A_i)\geq \anypricei$.
\end{observation}

We next show the APS of an agent is at least as large as the pessimistic share. 
\begin{restatable}{proposition}{propAPSPES}		
%\begin{proposition}
\label{prop:APS-pes}
	For any valuation $v_i$ and entitlement $b_i$ it holds that %the AnyPrice share (APS) satisfies 
	$$\anypricebi\geq \pessimisticbi.$$
	%$ b_i$-AnyPrice $\geq b_i$-pessimistic$\geq \frac{1}{2}\cdot b_i$-AnyPrice. 	
%	\mbc{Is the right inequality true for sub-additive? }
%\end{proposition}
\end{restatable}

%For an agent with entitlement of $1/2$ (as in the case of two agents with equal entitlements) the APS is the same at the MMS (see Observation \ref{obs:aps-mms-half}). 
Computationally, the problem of computing the APS is NP-hard, much like the problem of computing the MMS.  
\begin{proposition}
	\label{prop:computation}
	The problem of computing the AnyPrice share of an agent %with an additive valuation 
	is NP-hard.  This is true even for $b_i=\frac{1}{2}$ (the case of two equally entitled agents) and even when the valuations are additive. 
\end{proposition}

Proposition~\ref{prop:computation} follows from the fact that for $b_i=\frac{1}{2}$ the APS and the MMS are identical (Observation \ref{obs:APS-two-equal}), and as MMS is hard to compute in the case of two additive agents with equal entitlements (need to solve \emph{PARTITION}), so is the APS.

The APS definition is general and applies to all valuations, not only additive. 
Yet, ex-post fairness is problematic for some super-additive valuations {(even for equal entitlements)}. 
To illustrate the  problem,
%that arises with super-additive valuations,
consider the simple setting with agents that each has zero value for her set, unless she gets all items. Splitting the set of items makes little sense, and any reasonable solution will allocate all items to one agent (possibly at random, to get some ex-ante fairness). Thus we view ex-post fairness notions (APS, as well as all other ex-post notions like MMS, CE, etc.) as unsuitable for some super-additive valuations. 
In contrast, for sub-additive valuations,
splitting the set of items does make sense, and so do ex-post fairness notions, and the APS in particular. 
%While our main focus is on additive valuations, in Appendix~\ref{sec:unit-demand} we illustrate that the APS is a useful fairness notion also for unit-demand valuations.}
Before moving to our main interest in this paper, the case of  additive  valuations, we briefly illustrate the AnyPrice share on unit-demand valuations (which, like additive valuations, is a simple family of sub-additive valuations), illustrating that the APS is indeed a reasonable notion of a share for some sub-additive valuations. 

\subsubsection{Unit-Demand valuations}\label{sec:unit-demand}
% Before moving to our main interest in this paper, the case of  additive  valuations, we briefly illustrate the AnyPrice share on unit-demand valuations (which are another simple family of sub-additive valuations). 
A valuation is \emph{unit demand} if it assigns a value for each item, and the value of any set is the maximal value of any item in the set. 
We first note that unit-demand valuations illustrate that the proportional share is not attractive as a share definition  for non-additive valuations, even for equal entitlements. 
Consider for example the simple case of an agent that has a unit-demand valuation over $n$ identical items (same value of $1$ for each, and no additional value for more than one item).
While $n$ such agents can get a value of $1$ each, the proportional share of such an agent is only $1/n$,\footnote{This example shows that for non-additive valuations the proportional share might be smaller than the APS. Yet, Proposition \ref{prop:shares-comp} shows that is never the case for additive valuations. } so the proportional share is not even as large as the worst of the $n$ items.
In contrast, the APS share in this example is $1$. 

More generally, we observe that for unit-demand valuations with any entitlements, the value of the APS of an agent with entitlement $b_i$ is the value of her $\lceil 1/b_i \rceil$ ranked item.\footnote{It is easy to see that for unit-demand valuations the pessimistic share is the same as the APS.}
%Note that for an agent with small entitlement (smaller than $1/n$) the APS might be smaller than the proportional share, yet in such a case it might be impossible to give each agent her proportional share (when all player have the same valuations, with the worst item having a very tiny value).
%In contrast, 
There is always an allocation that gives each player her APS: simply let the agents pick items sequentially in the order of their entitlements, breaking ties arbitrarily.

%We note that for non-additive valuations, even for equal entitlements,  the proportional share might be smaller than the APS (yet, Proposition \ref{prop:shares-comp} show that is not the case for additive valuations). Consider for example and agent that has unit demand valuations over $n$ identical items (same value of $1$ for each, and no additional value for more than one item). While the APS of such an agent is $1$, his proportional share is only $1/n$.\footnote{This example indicates that unlike the APS, proportional share is not very attractive as a share definition for non-additive valuations, even for unit demand (which is sub-additive). } 

% \mbc{maybe expand on unit-demand, showing that the APS makes sense - an agent can get her $\lceil 1/b_i \rceil$ ranked item. incomparable to proportional (which we can not give to all). We can give the APS for unit demand by picking by budget order. Equal to pessimistic. }

\subsection{Basic Properties of APS with Additive Valuations}
We next focus on additive valuations. 
We start by relating the different notions of shares to each other. 
\begin{restatable}{proposition}{propShareComp}		
%\begin{proposition}
\label{prop:shares-comp}
	For any additive valuation $v_i$ over a set of items $\items$, and any entitlement $b_i$ it holds that 
	$$\proportionalbi \geq \anypricebi \geq \pessimisticbi \geq \frac{1}{2}\cdot \anypricebi$$
	
	%$b_i$-proportional$\geq b_i$-AnyPrice $\geq b_i$-pessimistic$\geq \frac{1}{2}\cdot b_i$-AnyPrice. 

	Moreover, for each of the above inequalities there is a setting in which it holds as equality. 
	Finally, even for the case of equal  entitlements ($b_i=\frac{1}{k}$ for an integer $k$), in which the pessimistic share is equal to the MMS, for each inequality above there is a setting in which it holds as a strict inequality.
%\end{proposition}
\end{restatable}
In particular, this claim shows that the APS is at least as large as the pessimistic share of \citep*{BabaioffNT2020}, and thus {obtaining share approximating allocations for the APS} is harder. Note that for equal entitlements ($b_i=\frac{1}{k}$ for an integer $k$) the MMS and the pessimistic shares are the same, so the above says that even for equal entitlements, the APS and the MMS are different. Yet, they are the same in the special case of two agents with equal entitlements ($b_i=\frac{1}{2}$).  	

%\begin{observation}
%\end{observation}
\begin{restatable}{observation}{obsApsTwo}
\label{obs:APS-two-equal}
For any additive valuation $v_i$ over a set of items $\items$, it holds that \\ $\anyprice{\frac{1}{2}}{v_i}{\items}= \MMSfull{2}$. %, that is, the APS with entitlement of half (corresponding to two agents with equal entitlements) is the same as the MMS of that agetn when there are two agents. 
%\mbc{proof should move to the appendix.}
\end{restatable}

With two agents, the cut-and-choose procedure ensures that each of the two agents gets her MMS (when $b_1=b_2=1/2$). We next observe that for the case of two agents, even with arbitrary entitlements, there is always an allocation that gives each agent her AnyPrice share. 
\begin{restatable}{observation}{obsTwo}

%\begin{observation}
\label{obs:APS-two-agents}
    For any setting with two agents with additive valuations $(v_1,v_2)$ and arbitrary entitlements $(b_1,b_2)$ there is an allocation $A=(A_1,A_2)$ that is an APS allocation. 
    That is, for every agent $i\in \{1,2\}$ it holds that $v_i(A_i)\geq \anypricei$. %\mbc{proof should move to the appendix.}
%    \mbc{old:}
%	For two agents with additive valuations and {arbitrary} entitlements, there is an allocation that gives \mbe{each of the two agents her AnyPrice share.} %both agents their AnyPrice share.   	
%\end{observation}
\end{restatable}

Recall that we have observed that the problem of computing the AnyPrice share of an agent with an additive valuation is NP-hard, like the problem of computing the MMS.
 Yet, we next show that it has a pseudo-polynomial time algorithm, while the MMS does not.\footnote{This implies the APS computation problem has a fully polynomial time approximation scheme (FPTAS), while the MMS does not. Note that the MMS problem does have a polynomial time approximation scheme (PTAS) (\citep{WOEGINGER1997}).}  
That is, for any integer additive valuation $v_i$ and entitlement $b_i$, if the value of the APS of an agent is $z$ 
% all item values are integers of value at most $T$, 
then the APS can be computed exactly in time polynomial in the input size and $z$ (note that $z$ is at most $v_i(\items)$), which is pseudo-polynomial\footnote{Note that a polynomial algorithm would run in time polynomial in $\log z$, not polynomial in $z$.}. 
This is in contrast to the MMS which is computationally harder -- %while it has a polynomial time approximation scheme (PTAS) (\cite{WOEGINGER1997}), 
it is known to be strongly NP-hard when the number of agents is large (\citep{WOEGINGER1997}, can be proven by a reduction from $3$-PARTITION) and hence it does not have a pseudo-polynomial time algorithm, unless P=NP.

% We next show that for additive valuations it can be well approximated, by showing that there is a  fully polynomial time approximation scheme (FPTAS) for the problem of computing the APS for an agent with an additive valuation. 
% This is in contrast to the MMS which is computationally harder - while it has a polynomial time approximation scheme (PTAS) (\cite{WOEGINGER1997}), it is known to be strongly NP-hard when the number of agents is large (\cite{WOEGINGER1997}, can be proven by reduction from $3$-PARTITION) and hence it cannot have an FPTAS, unless P=NP.
%\begin{proposition}
\begin{restatable}{proposition}{proppseudo}

	\label{prop:computation-psedu-poly}
	Consider the problem of computing, for any integer additive valuation $v_i$ and entitlement $b_i$ (where $b_i \leq 1$  is a rational number), the AnyPrice share $z=\anypricei$.
	There is an algorithm that computes the APS in time polynomial in the input size (representation of the valuations and entitlement) and $z$. %\mbc{proof should move to the appendix.} % \tec{rephrase, and it is supposed to be polynomial in the support of $v_i$} \mbc{is this better? I think we do not need the size of the support, only in the representation and in $z$ (that might be as large as the value of the entire set of items.) }%the value of  $\anypricei$. 
%	There exists an algorithm that for any constant $\varepsilon>0$ computes an $(1-\varepsilon)$-approximation to the APS, in polynomial time in the input size and $1/\varepsilon$. \mbc{Uri proved it has an FPTAS. Is the statement correct? Is it possible to run in time poly in $\log(1/\varepsilon)$? } 
%\end{proposition}
\end{restatable}

% \mbc{Proof from Uri's email. \\ APS has an FPTAS: For the LP that uses the price formulation, there is an FPTAS for the separation oracle (given prices, the separation problem is to solve a knapsack problem, and knapsack has an FPTAS).}

%\section{Relations between shares}
%\input{observations}
\section{APS Approximation for Additive Agents: Arbitrary Entitlements}\label{sec:allocation-game}

In this section we present our main result regarding fair-share approximation for additive agents with arbitrary entitlements.  
We start with extending the definition of Truncated Proportional Share (TPS) of  \citet*{babaioff2021bestofbothworlds} to arbitrary entitlements, and show it is at least as high as the APS, {and thus can serve as a tool in proving results for APS}.
%\ufc{Perhaps say that we first present an easier $\frac{1}{2}$ fraction.}\mbc{see added sentences below}
Then, we present our main result -- showing that for agents with  additive valuations and arbitrary entitlements 
there is a bidding game in which each agent can always secure herself a $\frac{3}{5}$ fraction of her AnyPrice share, and can also secure at least $\frac{1}{2-b_i}$ fraction of her Truncated Proportional Share (which is at least her APS).
% We repeat all this paragraph at the beginning of section 4.2, so we can at least save some space here.  
%We also show that any agent $i$ can secure herself a value of at least the value of her $\lfloor 1/b_i\rfloor$ 
%%$\Bigl\lfloor \frac{1}{b_i}\Bigr\rfloor$ 
%ranked item. 
% Our proof   
{As the analysis of the strategies for these % the first two  
bounds is rather involved, most details of the proofs are deferred to the appendix. To give the reader a sense of some ideas that we use, we  present some weaker bounds that are easier to analyze, along with their analysis. In particular, we present a natural strategy that secures half the TPS, and a strategy that obtains some constant fraction, that is strictly larger than half, of the APS.    }

\subsection{The Truncated Proportional Share (TPS)}
\label{sec:tps}
The \emph{Truncated Proportional Share (TPS)} was defined in \citep*{babaioff2021bestofbothworlds}
%\tec{We need to think of how to present the story if we use a definition that is not based on removing items} \mbc{we just need to show that this general definition for arbitrary entitlements is giving the same share as the TPS for equal entitlements, in the special case of equal entitlements.   }
for additive agents with equal entitlements. In this section we  extend their definition to arbitrary entitlements and show that for additive agents the TPS is always at least as large as the APS.
We then present a polynomial time algorithm that gives every agent a constant fraction of her TPS, and thus also that fraction of her APS share.
 
%\tec{we need to rephrase from here depending on the other paper} \mbc{please draft an alternative}
%\tec{Old:CITE \cite{} have defined the TPS for agents with equal entitlements by an iterative procedure: if there is no item that has a normalized value\footnote{Scaled so that the set of all items has value of $1$.} larger then the entitlement ($1/n$), then the TPS is simply the proportional share of the agent. Otherwise the TPS is defined recursively to be the share after removing an item of highest value, and one agent (so now the share is computed for a smaller set of items and with each agent having an entitlement of $1/(n-1)$).   We adopt this procedure to the unequal entitlements case by setting the removal threshold to be the proportional share, and carefully doing the entitlement adjustment.}
\citet{babaioff2021bestofbothworlds}  have defined the TPS for equal entitlements. For equal entitlements setting with $n$ agents and a set of items $\items$, the \emph{truncated proportional share} of agent $i$ with additive valuation function $v_i$ is the largest value $t$ such that $\frac{1}{n} \sum_{j\in \items} \min[v_j, t] = t$. 
We extend this definition to the case of arbitrary entitlements in a natural way, by thinking of $\frac{1}{n}$ as the entitlement of the agent in the equal entitlement setting, and replacing it by the agent's entitlement $b_i$ in the arbitrary entitlement case: 

% \mbe{ for agents with equal entitlements by $\truncated{n}{\items}{v_i}$ to be the maximal value $t$ such that $t= \frac{1}{n} \sum_{j\in \items} \min(t,v_i(j)) $. We extend this definition to the case of unequal entitlements in the following way: } 
\begin{definition}\label{def:TPS}
	The \emph{Truncated Proportional Share (TPS)} of agent $i$ with an additive valuation $v_i$ over the set of items $\items$ and entitlement $b_i$, denoted by $\truncatedi$, is:
	\begin{equation}
	\truncatedi = \max\{~z \mid b_i \cdot \sum_{j\in \items} \min(v_i(j),z) = z ~\}. \label{eq:tps}    
	\end{equation}
	 
%	 the result of the following procedure:
%	\begin{itemize}
%		\item If all items relative value is less than $b_i$ (i.e., for all $j\in \items$, $\frac{v_i(j)}{v_i(\items)}\leq b_i$), then $\truncatedi=b_i \cdot v_i(\items)$.
%		\item Else, $\truncatedi = \truncated{v_i}{\min(\frac{b_i}{1-b_i},1)}{\items \setminus \{j^*\}}$, where $j^*$ is agent $i$ most favorable item in $\items$.
%	\end{itemize} 
%\label{def:truncated}
When $\items$ and $v_i$ are clear from the context we denote this share by $\TPSbi$. %, instead of $\truncatedi$.
\end{definition}

{Observe that it is immediate from Definition~\ref{def:TPS} that the TPS is at most the proportional share, and that  the TPS is exactly the proportional share of the instance after capping the value of each item at the TPS.  
% Definition~\ref{def:TPS} basically says that the TPS is identical to the proportional share, unless there are items that are worth more than the proportional share.
%Moreover, it is easy to observe that the TPS is the PS of the instance after each items is reduced to the TPS.  
%In such a case their values are decreased until all items have values that are equal to the truncated proportional share of the valuation after the decrease.
Alternative, one can think about the following continuous process to reach the TPS by capping the valuation: a cap on item values is reduced 
%An alternative way to think about the TPS is to consider a continuous process of reducing  a cap on item values 
until the point that no item has a capped value which is strictly greater than the proportional share of the current capped valuation.
%As there is no reason to reduce the value of an item below the value of the proportional share  after the reduction, the TPS can be viewed as the result of a continuous reduction of a cap on item values, until that point that under the cap no item has a value which is strictly greater than the current reduced proportional share.

}
%\tec{Old:}
%In such a case their values are decreased until all items have values that are at most the proportional share of the valuation after the decrease.
%\tec{new:}
%In such cases, the values of all most valuable items are decreased in a  continuous way until no item has a value which is strictly greater than the current reduced proportional share. Then, the TPS of the original valuation is the proportional share of the reduced values.  
%%the sum of entitlements by $b_i$ and then normalize the entitlements, and continue computing the share with respect to the new instance. 
%{Observe that by Definition~\ref{def:TPS} the TPS is at most the PS.}  

{We next show that the TPS is always at least as large as the APS. Thus, by obtaining a guarantee with respect to the TPS, we will also obtain the same guarantee with respect to the APS.} 
%OLD: {The TPS can be used in order to obtain an approximate APS allocation, because the TPS is always at least as large as the APS.} 
%as we will show there is always an allocation that gives each agent at least  $\frac{1}{2-b_i}$-fraction of her TPS, and 
 \begin{proposition}\label{prop:TPS-APS}
 	For agent $i$ with an additive  valuation $v_i$ over $\items$ and entitlement $b_i$, it holds that 
 	$\truncatedi \geq \anypricei$.
% 	$b_i$-TPS $\geq b_i$-AnyPrice.
 \end{proposition}
\begin{proof}
Let $w=\truncatedi$.
We denote the set of items that have values strictly more than $w$ by $K$, and let $k=|K|$. If $k=0$ then the $\truncatedi = \proportional{b_i}{v_i}{\items}$, and thus by Proposition~\ref{prop:shares-comp} the proposition holds.
Assume $k>0$. 
By Definition~\ref{def:TPS}, it holds that $w=b_i \cdot \left(k\cdot w +\sum_{j\in \items \setminus K} v_i(j)\right) $.
{It also holds that $b_i \cdot k <1$.
To see this, first observe that if $b_i \cdot k > 1$ then $w = b_i \cdot \left(k\cdot w +\sum_{j\in \items \setminus K} v_i(j)\right)  \geq b_i \cdot k\cdot w >w $, a contradiction.
Finally, consider the case that $b_i \cdot k =1$. As $w = b_i \cdot \left(k\cdot w +\sum_{j\in \items \setminus K} v_i(j)\right) $ we have $\sum_{j\in \items \setminus K} v_i(j)=0$. 
Let $\ell=\min_{j\in K} v_i(j) $. 
Observe that 
$$ b_i \cdot \sum_{j \in \items} \min(v_i(j),\ell)   =  b_i \cdot \sum_{j \in K} \min(v_i(j),\ell) = b_i \cdot k \cdot \ell =\ell.$$
%where the first equality is since $b_i \cdot \sum_{j \in \items \setminus K} \min(v_i(j),\ell) =b_i \cdot \sum_{j \in \items \setminus K} \min(v_i(j),w)=0 $.
%The second equality is since by the definition of $K$, it holds that $v_i(\j) \geq \ell$ for every item $j \in K$. The third equality is since $b_i=\frac{1}{k}$.
Thus  $\ell$ satisfies Equation~\eqref{eq:tps} and hence the TPS is at least $\ell > w$, which is a contradiction. We thus conclude that $b_i \cdot k <1$.
%the TPS must be at least $\min_{j\in K} v_i(j) $ which is strictly greater than $z$ since $b_i \cdot k \cdot \min_{j\in K} v_i(j) \geq  \min_{j\in K} v_i(j)$. 
} 
%This expression is always well defined since its denominator is always strictly more than 0, since $b_i< \frac{1}{k}$.

If $v_i(\items \setminus K)=0$ then if we price all items in $K$  by $\frac{1}{k}$, the agent cannot afford any 
{item she has positive value for %valuable item 
($v_i(\items \setminus K)=0$  and since $b_i<\frac{1}{k}$ she cannot get any item in $K$),}
and thus $\anypricei =0 \leq \truncatedi$.

Else, we have $v_i(\items \setminus K)>0$. Assume towards contradiction that the value  $z^*=\anypricei$ is strictly more than $w$.
Thus, there exists an $\epsilon>0$ such that  
$$ w < \frac{b_i \cdot v_i(\items \setminus K)}{1-k(b_i+\epsilon)} < z^*.$$ 
{Such an $\epsilon>0$ must exist since when $\epsilon \to 0 $ then this expression goes to $w$, while when $\epsilon \to \frac{1}{k}-b_i$, the expression goes to infinity, and it is continuous in $\epsilon$ for $\epsilon\in(0,\frac{1}{k}-b_i)$.}
Consider the pricing $b_i+\epsilon$ for every item in $K$, and for every item $j \in \items \setminus K$, we give a price $p_j=\frac{v_i(j)\cdot (1-k(b_i+\epsilon))}{v_i(\items \setminus K)}$.
The agent cannot afford any item among $K$. Among the other items, her value is proportional to the budget spent and thus is at most $\frac{b_i}{1-k(b_i+\epsilon)} \cdot v_i(\items \setminus K)$, which is strictly smaller than $z^*$, contradicting the definition of $z^*$ as her APS.
\end{proof}

In Appendix \ref{app:TPS} we discuss some other properties of the TPS. We show that the TPS might be much larger than the APS (Observation \ref{obs-TPS-mush-larger-APS}).
{We observe that while the APS is (weakly) NP-hard to compute, the TPS can be computed in polynomial time (Observation~\ref{obs:tps-computation}).} Yet, we only use TPS as a tool, and do not consider it as the most appropriate share for agents with arbitrary entitlements, as it has some significant drawbacks.
The main drawback of TPS is that, similarly to the proportional share, it seems unattractive beyond additive valuation:  the unit-demand valuations (with identical items)  example in Section \ref{sec:unit-demand}  is one in which the TPS is identical to PS, and the proportional share
(and TPS) seem too small (only $1/n$ instead of $1$).

We also show (Proposition \ref{prop:TPS-CE}) that a competitive equilibrium (CE)  does not necessarily guarantee that every agent gets her TPS, unlike the case for APS (in which in every CE, every agent gets her APS).  
{Finally, we note (see Observation~\ref{obs:TPS-app-UB}) that even for equal entitlements, it is not possible to give every agent more than a $\frac{n}{2n-1}$ fraction of her TPS. 
%More precisely, for any constant $\rho$ larger than half, there is an instance (with equal entitlements) in which some agent does not get $\rho$-fraction of her TPS. 
This is in contrast to the APS for which our main result shows that it is possible to give every agent a $\frac{3}{5}$ fraction of her APS.}

\subsection{Main Result: Approximate APS Allocations}

%Now that we have defined the AnyPrice share,  
We next turn to present and prove our main result: %, about the approximation to APS that can be given to all agents concurrently.   
% Our first main result is  that 
for agents with additive valuations and arbitrary entitlements, it is always possible to give each agent at least a $\frac{3}{5}$ fraction of her AnyPrice share, and at least $\frac{1}{2-b_i}$ fraction of her Truncated Proportional Share (which is at least her APS).
This $\frac{1}{2-b_i}$ fraction  is more than $\frac{3}{5}$ if the agent entitlement is larger than $\frac{1}{3}$, and it %. Additionally, note that the fraction $\frac{1}{2-b_i}$  
goes to $1$ as $b_i$ grows to $1$. Moreover, we show that agent $i$ also gets a value that is % can also secure herself value of 
at least the value of her $\lfloor 1/b_i\rfloor$  %$\Bigl\lfloor \frac{1}{b_i}\Bigr\rfloor$
ranked item. 
Note that this result implies Theorem \ref{thm:intro-unequal} (it is a stronger version of it). 
{Our result follows from presenting a natural ``bidding game", and showing that each agent has a strategy to secure each of the  guarantees mentioned above. }
   
Consider the following  ``bidding game" among the agents: %allocation mechanism: % \mbc{procedure? algorithm? "Mechanism" might be mistaken for truthful/NE }:
Every agent $i$ starts with a budget of $b_i$.
As long as there are still items left, in each round $t$, every agent bids an amount $\bid_i^{(t)}$ between $0$ and $b_i^{(t)}$. 
An agent with the highest bid (breaking ties arbitrarily) is the winner of the round, we denote that winning agent by $w^{(t)}$.
% Let $w^{(t)}$ be the winning agent that gave the maximal bid.
{ % shortening this and making it less formal (the formal version is in the game description):
Within her remaining budget $b_{w^{(t)}}^{(t)}$, 
agent $w^{(t)}$ selects available items she wants to take, paying  $\bid_{w^{(t)}}^{(t)}$ for each item she takes.}
%We also denote $s_i^{(1)}$ by $s_i$.
%OLD: 
%Agent $w^{(t)}$ can select any number of items between $1$ and  $\Bigl\lfloor \frac{b_{w^{(t)}}^{(t)}}{\bid_{w^{(t)}}^{(t)}} \Bigr\rfloor$ and her budget is decreased by $\bid_{w^{(t)}}^{(t)}$ for every chosen item. 

%\ufc{It should not be called an algorithm}\mbc{in many places we say "Algorithm 1", we need to make sure we change them if we change the name here. }
\begin{algorithm}
	\caption{The Bidding Game }
	%\label{alg:bidding}
%	\addcontentsline{loa}{algorithm}{Bidding Game}
	\begin{algorithmic}[1]
	\STATE Input: Set of items $\items$, entitlements $\vect{b}=(b_1,\ldots,b_n)$.
	\\ We have the following notations for the beginning of round $t$:  
	\\ \quad\quad\quad\quad\quad
	$\items^{(t)}$ - the available items at the beginning of round $t$.
	\\ \quad\quad\quad\quad\quad $s_i^{(t)}=v_i(\items^{(t)})$ - the value of agent $i$ for $\items^{(t)}$, the items available. %  at round  $t$. % at the beginning of round $t$.
	\\ \quad\quad\quad\quad\quad $x_i^{(t)}$ -- the highest value $i$ assigns to any item in $\items^{(t)}$.
	\\ \quad\quad\quad\quad\quad $y_i^{(t)}$ -- the second highest value $i$ assigns to any item in $\items^{(t)}$.
	\\ \quad\quad\quad\quad\quad $b_i^{(t)}$ -- the  budget available to  $i$ at the beginning of round $t$.
	\\ \quad\quad\quad\quad\quad $B^{(t)}=\sum_k b_k^{(t)}$ -- the total budget remaining at the beginning of round $t$.
	% \\ \quad\quad\quad\quad\quad\quad\quad\quad items respectively   at the beginning of round $t$ for agent $i$.
%	\\ \quad\quad\quad\quad\quad $x_i^{(t)},y_i^{(t)}$ - the value of the highest  and second highest  value available  \\ \quad\quad\quad\quad\quad\quad\quad\quad items respectively   at the beginning of round $t$ for agent $i$.
%	\\ We also use $w^{(t)}$to denote the winning agent of round $t$.
	\STATE Initialize: $t=1$, for every agent $i$, $b_i^{(1)}=b_i$, and $\items^{(1)} =\items$ 
	\WHILE{$\items^{(t)} \neq \emptyset $}
	\STATE Every agent $i$ bids an amount $\bid_i^{(t)} \in [0,b_i^{(t)}]$
	\STATE Let $w^{(t)} \in \arg\max \bid_i^{(t)}$ (breaking ties arbitrarily) be the winning agent of round $t$.
	\STATE Agent $w^{(t)}$ selects a {non-empty set 
	%of items 
	$W^{(t)} \subseteq \items^{(t)}$ 
	that she can afford: $|W^{(t)}|\cdot \bid_{w^{(t)}}^{(t)}\leq b_{w^{(t)}}^{(t)} $.}
%	of size between $1$ and 	$ \Bigl\lfloor \frac{b_{w^{(t)}}^{(t)}}{\bid_{w^{(t)}}^{(t)}} \Bigr\rfloor$.
	% OLD: $ \Bigl\lfloor \frac{b_i^{(t)}}{\bid_i^{(t)}} \Bigr\rfloor$.
    \STATE{Update budgets:} ~\quad\quad\quad\quad\quad  For every agent $i \neq w^{(t)}$ set $b_i^{(t+1)} = b_i^{(t)}$.
	\STATE \quad \quad \quad\quad\quad\quad \quad \quad\quad\quad\quad\quad Set $b_{w^{(t)}}^{(t+1)} = b_{w^{(t)}}^{(t)} - |W^{(t)}| \cdot \bid_{w^{(t)}}^{(t)}$
	\STATE{Remove allocated items: } ~\quad $\items^{(t+1)} = \items^{(t)}\setminus W^{(t)}$
	\STATE $t\leftarrow t+1$
	\ENDWHILE
	\end{algorithmic}
\end{algorithm}

\begin{restatable}{theorem}{thmappAPS}		   
%\begin{theorem}
	\label{thm:app-APS}

	There exists an allocation mechanism  
	for settings with $n$ agents that have additive valuations and arbitrary entitlements, which has the following properties. 
	In that mechanism {(the Bidding Game)} % for any additive valuations \vals\ and entitlements \budgets\ 
	every agent $i$ has a strategy that regardless of the %for any 
	strategies of the other agents, ensures she gets a bundle of value at least the largest of the following three values:
	\begin{itemize}
		\item$60\%$ %$3/5$ fraction 
		of her AnyPrice share.
		\item $\frac{1}{2-b_i}$ fraction of her Truncated Proportional Share.
		\item the value of her  $\lfloor 1/b_i\rfloor$  % $\Bigl\lfloor \frac{1}{b_i}\Bigr\rfloor$ 
		ranked item.
	\end{itemize}  
%	there exists a strategy for every agent that guarantees she is allocated a set of value  at least $60\%$ %$3/5$ fraction 
%	of her AnyPrice share, and also at least $\frac{1}{2-b_i}$ fraction of her Truncated Proportional Share,  for any strategies the others play. 
	Moreover, that mechanism runs in polynomial time. 
	% For any additive valuations \vals\ and entitlements \budgets\  there exists an allocation in which every agent gets at least $60\%$ %$3/5$ fraction of her AnyPrice share.  Moreover, such an allocation can be computed in polynomial time. 
%\end{theorem}
\end{restatable}

% \mbc{as we have a claim about each agent being able to secure some value, it seems  that in any pure, and maybe even every mixed, NE of this game, each agent will also secure that value (although it is not clear that a pure NE exists I think this hold for mixed NE as well, right?). I think we should remark about this. } \ufc{This property holds for every share-based allocation algorithm, not just for our mechanism. So maybe we should not emphasize it too much here, though it is worth mentioning this aspect in the introduction.}\mbc{we did. Should we remove this comment?}

Before moving to the proof, we remark that for the case of equal entitlements ($b_i=1/n$ for every $i$), the approximation of $\frac{1}{2-b_i}$ for the TPS in the theorem is tight and cannot be improved.  
Indeed, consider $n$ agents with equal entitlements and  $2n-1$ identical items, each of value $\frac{1}{2n-1}$. 
One of the $n$ agents must get at most one item, getting a value of only $\frac{1}{2n-1}$ while her TPS is $\frac{1}{n}$ (identical to her proportional share as no item has value larger than her {proportional} share). Thus, that agent gets only $\frac{n}{2n-1}$ fraction of her TPS, which equals a $\frac{1}{2-b_i}$ fraction of her TPS, as $b_i=1/n$. Note that $\frac{n}{2n-1}$ approaches  $\frac{1}{2}$ when $n$ goes to infinity, and thus any constant approximation to the TPS that is larger than $50\%$ is impossible, while we are able to obtain $60\%$ of the weaker benchmark of APS. 

\vspace{0.2in}
The proof of the theorem is based on presenting three strategies, each guaranteeing one of the bounds promised by the theorem. By picking the one with the highest guarantee, the agent gets all three guarantees.  
While the last bound (with respect to the ranked items) is simple, the other two bounds {are substantially more difficult}.  %require involved strategies, and their analysis is also involved. 
{As a warm-up, we will start by presenting a simple strategy that ensures that agent $i$ gets at least half of her TPS.}
We leave the proof  that  agent $i$ can secure the stronger bound  $\frac{1}{2-b_i}$ fraction of her TPS to the appendix. 
{The claim that there is a strategy for agent $i$ that guarantees $60\%$ of her APS appears in Lemma \ref{lem:35}.} As presenting the strategy and its analysis in full is {rather complex}, %still too involved, 
in the body of the paper we show how the most involved part can be substituted with a simpler step (presented in Lemma \ref{lem:safeStrategy}) that is enough to guarantee a smaller fraction of $8/15$ (while leaving the more {difficult} %involved 
proof of the stronger bound of $3/5$ to the appendix).
We note that this fraction of $8/15$ still illustrates an important qualitative point - that it is possible for an agent to guarantee herself some constant fraction that is strictly more than 50\% of her APS. 
        
% The result of the theorem guarantees each agent $60\%$ of her APS. 

%We next move to present another approximation result that gives better guarantees for agents with large entitlements (guarantee that goes to $100\%$ as the entitlement goes to $1$). To prove this result we first discuss another share that will be useful in our proof. 
{To illustrate the use of the bidding game, and likewise, the usefulness of the TPS as an upper bound on the APS, we first present a relatively easy proof that there is an allocation that gives every agent at least half of her TPS (and hence of her APS).} 

The \emph{bid-your-max-value strategy} is the strategy in which at each round agent $i$ bids her normalized value of her highest-value remaining  item, unless it is higher than her remaining budget -- in this case she bids her entire remaining budget. If she wins she picks that item. 
Formally: 
If at the beginning of round $t$, the remaining budget of $i$ is $b_i^{(t)}$, and the highest value remaining  item has value $x_i^{(t)}$, then at round $t$ agent $i$ bids $\min\left(\frac{x_i^{(t)} }{v_i(\items)},b_i^{(t)}\right)$ and chooses the maximal value remaining item when winning.

{It is immediate from the definition of the strategy that if an agent $i$ is able to spend her entire budget, she will get her proportional share. The following lemma plays a key role in reasoning about those instances in which $i$ fails to spend her full budget. %\ufc{I am fine with ending this paragraph here.}
It considers only the special case that no single item has a value larger than the proportional share, but other cases can be reduced to this special case, as will be seen in Corollary~\ref{cor:safeStrategy}. For this special case,
any time that agent $i$ bids her max value but does not win, the total budget of the other agents decreases by at least her max value. It follows that the total available budget of all agents decreases at a rate that is at least as fast as the rate of decrease of the total value of the remaining items. Hence either agent $i$ wins some items, or the other agents run out of budget before they manage to buy all the items that agent $i$ values.}

\begin{lemma}
	\label{lem:safeStrategy}
	The bid-your-max-value strategy  of agent $i$ with additive valuation $v_i$ and entitlement $b_i$ provides the following guarantee, regardless of the  strategies of the other agents. 
	If for some positive integer $k$ 
    no subset of $k$ items has a value larger than the proportional share $b_i \cdot v_i(\items)$,
%	no item has value more than $\frac{b_i \cdot v_i(\items)}{k}$, 
	then the value agent $i$ gets is at least a {$\frac{k}{k+1}$-fraction of her proportional share (i.e., at least $\frac{k}{k+1}\cdot b_i \cdot v_i(\items)$).}
	
%	\mbc{new:}	If for agent $i$ with valuation $v_i$ and  entitlement $b_i$ it holds that  for positive integer $k$ no item has value more than $\frac{b_i \cdot v_i(\items)}{k}$, then agent $i$	has a strategy in the bidding game that guarantees herself a value of at least $\left(1-\frac{1}{k+1}\right) b_i \cdot v_i(\items)$, for any strategies of the other agents. 
	
%	\mbc{Is it possible to replace the assumption that "no item has value more than $\frac{b_i \cdot v_i(\items)}{k}$" with setting this as the benchmark, that is, setting the benchmark to be $	 \max\{~z \mid b_i \cdot \sum_{j\in \items} \min(\frac{v_i(j)}{k},z) = z ~\}$ and getting the fraction of the benchmark ? If so it will be clear we get the TPS for $k=1$. }

%\mbc{OLD:} 
	
%	For an agent $i$ with an additive valuation function $v_i$ and entitlement $b_i$, if for some positive integer $k$ no item has value more than $\frac{b_i \cdot s_i}{k}$, \mbe{then the strategy in which at round $t$ agent $i$ bids $\min\left(\frac{x_i^{(t)} }{s_i},b_i^{(t)}\right)$ and chooses the maximal value item when winning, ensures her a value of at least $\left(1-\frac{1}{k+1}\right) b_i \cdot s_i$}. 
	
%	\mbc{this does not look like the required bound, explain how they are related - seems non-trivial. Why do we have this lemma here?}
\end{lemma}
\begin{proof}
      {We denote $v_i(\items)$ by $s_i$.} Recall that we consider the strategy in which at round $t$ agent $i$ bids $\min\Big\{\frac{x_i^{(t)} }{s_i},b_i^{(t)}\Big\}$, and selects the item of highest value if she wins. We shall say that agent $i$ bids her {\em max value} at round $t$ if $\frac{x_i^{(t)} }{s_i} \le b_i^{(t)}$, and that $i$ bids her budget if $\frac{x_i^{(t)} }{s_i} > b_i^{(t)}$.

	%	Consider first the case that $k = 1$.  
	Let $t^*$ be the latest round such that for every round up to and including $t^*$, agent $i$ bids her max value. That is, $\frac{x_i^{(t)} }{s_i } \leq b_i^{(t)}$ for every $t \le  t^*$. %\tee{(in all rounds until $t^*$, the agent bids proportionally to its value)}. 
	%If agent $i$ has no budget left after round $t^*$, then the total value $y$ that she accumulated satisfies $\frac{y}{s_i} = \b_i$, and hence $y \ge b_i \cdot v_i(\items)$, proving the lemma.
	Let $u_i^{(t^*+1)}$ denote the total value accumulated 
	by agent $i$ up to the beginning of round $t^* + 1$. %(If $t^*$ happens to be the last round of the bidding game, then 	$(t^*+1)$ refers to the values after the last round.)  
% 	\mbe{$x_i^{(t^*+1)}$ refers to the highest value, according to $v_i$, of any item no sold}.)
%	\mbe{That is, we use the notations $s_i^{(t^*+1)},x_i^{(t^*+1)}$ \tec{we don't use in this proof  $s_i^{(t^*)}$} as if there was another round. So they denote the corresponding values at the end of round $t^*$}). % (and then $s_i^{(t^*+1)}=0$.) 
%\mbc{as we later also refer to $s_i^{(t^*+1)}, x_i^{(t^*+1)}$ I think we need a more general statement of the form "With slight abuse of notation, if $t^*$ is the last round we use the notations  $s_i^{(t^*+1)},x_i^{(t^*+1)}$ as if there was another round. So they denote the corresponding values at the end of round $t^*$"}
	
	Observe that $\frac{u_i^{(t^*+1)}}{s_i} = b_i - b_i^{(t^*+1)}$, {as up to (and including) round $t^*$  agent $i$ always bids exactly her max value}, and hence  $u_i^{(t^*+1)} = s_i\cdot (b_i - b_i^{(t^*+1)})$. Thus the value accumulated by agent $i$ by the end of the bidding game is at least $\left(b_i - b_i^{(t^*+1)}\right)\cdot v_i(\items)$. Hence if $b_i^{(t^*+1)} \le \frac{b_i}{k+1}$, the Lemma is proved.
	%This follows by that agent $i$ value up to time $t^*+1$ is proportional to the budget used by him.
	
	%We show that $\frac{s_i^{(t^*+1)}}{s_i} \geq B^{(t^*+1)}$.
	%Observe that  $\frac{s_i^{(t^*+1)}}{s_i} \geq B^{(t^*+1)}$ holds, because for every $t \in [t^*]$ the purchased items  cost at least their proportional budget according to $v_i$ since in all those round agent $i$ bids $\frac{x_i^{(t)} }{s_i}$.
	
	%We then observe that the value accumulated of agent $i$  at time $t^*+1$ is exactly $s_i\cdot (b_i^{(1)}-b_i^{(t^*+1)})$ \tee{(and thus her value in the end of the game is at least this much)}. 
	%This follows by that agent $i$ value up to time $t^*+1$ is proportional to the budget used by him.
	
It remains to consider the case that $b_i^{(t^*+1)} > \frac{b_i}{k+1}> 0$. {We claim that in this case $s_i^{(t^* + 1)}   >  0$. The claim follows from the following argument. In every round $t \le t^*$ agent $i$ bids her max value $\frac{x_i^{(t)}}{s_i}$. Hence in every round $t \le t^*$, regardless of which agent won the round, the payment per item consumed in round $t$ was at least $\frac{x_i^{(t)}}{s_i}$, whereas every item consumed in round $t$ had value at most $x_i^{(t)}$ (according to $v_i$). As the initial total budget is~1, and $B^{(t^*+1)} \ge b_i^{(t^*+1)} > 0$, the total budget consumed up to and including round $t^*$, namely, $1 - B^{(t^*+1)}$, satisfies $1 - B^{(t^*+1)} < 1$. Hence the total value (according to $v_i$) consumed is at most $(1 - B^{(t^*+1)})s_i < s_i$, implying that $s_i^{(t^* + 1)}   >  0$, as claimed.}
%In this case $B^{(t^*+1)} \ge b_i^{(t^*+1)} > 0$. 
%Consequently, \tee{ 	$s_i - s_i^{(t^* + 1)}  \leq\frac{s_i}{B} \cdot \left( B-B^{(t^* + 1)}\right) < s_i$}, as up to and including round $t^*$, the rate of consuming $v_i$ value from $s_i$, %\mbe{from the perspective  of the valuation of $i$,}
%was not higher than the rate of consuming the total budget $B = B^{(1)} = 1$.  {As some positive budget remains, so does some positive value}. \tee{Therefore, 	$s_i^{(t^* + 1)}   >  0$.}

By the maximality of $t^*$, in round $t^{*} + 1$ agent $i$ bids her budget. This implies that $b_i^{(t^*+1)} < \frac{x_i^{(t^*+1)}}{s_i}$. 
	
	Let $F_i$ denote the set of items  that agent $i$ received up to round $t^*+1$. 
% \tee{new proof:}
	If $|F_i|< k$ then by the assumption of the lemma, it holds that %the value of $F_i$ plus the value of the most valuable item in $\items^{(t^*+1)}$ is at most $b_i \cdot s_i$ that is: 
	$v_i(F_i) + x_i^{(t^*+1)} \leq b_i\cdot v_i(\items) = b_i \cdot s_i$. 
	%This is since $F_i$
	%Since $x_i^{(t^*+1)}$ is at most the value of the $k$ highest value item for $i$ it holds that 
	Thus, it holds that $\frac{x_i^{(t^*+1)}}{s_i}\leq b_i^{(t^*+1)} $, and agent $i$ can afford to bid her max value, a contradiction to the definition of $t^*$.
	%\mbc{need to finish}
	
	%\mbc{old (proof under the  old assumption) }
	%It must be that $|F_i|\geq k$, since otherwise we get a contradiction as follows \mbc{must revise this.}
	%$$b_i^{(t^*+1)}  \geq b_i -(k-1)\cdot \frac{x_i^{(1)}}{s_i} \geq b_i-\frac{(k-1)b_i\cdot s_i}{k\cdot s_i}  = \frac{b_i}{k}\geq   \frac{x_i^{(t^*+1)}}{s_i},$$
	%\tec{remove second expression and explain better the inequality}
	%where the first inequality holds since $\frac{x_i^{(1)}}{s_i}$ is the maximal amount the agent may pay for an item. The second and third inequalities follow by the assumption of the lemma, that  $x_i^{(1)}\leq \frac{b_i\cdot s_i}{k}$.
    We are left to handle the case of $|F_i| \geq k$.
    Since $b_i^{(t^*+1)} < \frac{x_i^{(t^*+1)}}{s_i} $ and since $b_i^{(t^*+1)} =b_i - \frac{v_i(F_i)}{s_i} $,  
    it holds that  $v_i(F_i)+x_i^{(t^*+1)} > b_i \cdot s_i$.
    %Since \mbe{at round $t^*+1$} agent $i$ cannot afford to bid her max value, and all \tee{items in $F_i$} were purchased \mbc{unclear. do you mean $F_i$?} when she bid her max value, then $v_i(F_i)+x_i^{(t^*+1)} > b_i \cdot s_i$.
    %In addition, 
	Hence, 
	%> s_i \cdot b_i^{(t^*+1)}$}. 
	%$$v_i(F_i) \ge k \cdot x_i^{(t^*+1)} \ge k \cdot s_i \cdot b_i^{(t^* + 1)} \ge k\cdot v_i({\items}) \cdot \frac{b_i}{k+1},$$   
	$$v_i(F_i) =  \frac{|F_i| \cdot v_i(F_i) + v_i(F_i)}{|F_i|+1} \geq 
	|F_i|\cdot \frac{v_i(F_i) + x_i^{(t^*+1)}}{|F_i|+1} > \frac{k}{k+1} \cdot  b_i  \cdot s_i ,$$   
	where the first inequality holds since 
	every item in $F_i$ has a value of at least $x_i^{(t^*+1)}$ and thus $ v_i(F_i) \geq |F_i| \cdot x_i^{(t^*+1)}$. The second inequality holds since $|F_i| \geq k$ and $v_i(F_i) + x_i^{(t^*+1)} > b_i \cdot s_i$.
	This concludes the proof. % looks good.
	%This means that together with $x_i^{(t^*+1)}$ agent $i$ receives at least $b_i\cdot s_i$, but agent $i$ received at least $\frac{f}{f+1} \geq \frac{k}{k+1}$ of $b_i\cdot s_i$, \tee{(since she spent at least $\frac{f}{f+1}$ of her budget)} which concludes the proof.
%	\tec{change to spent $\frac{k}{k+1}$ of his budget.}
\end{proof}

% \mbc{I suspect that for the proof of the $8/15$ to work (as a variant of Lemma \ref{lem:35}) we need a stronger version of this corollary - the one that holds also for $k=2$ (or more generally, for every $k$). We then get the TPS result will be the special case, for $k=1$.} \tec{in the proof of Lemma \ref{lem:35}) we can use the Lemma~\ref{lem:safeStrategy}  since the assumption already holds.}

An immediate corollary of Lemma~\ref{lem:safeStrategy} (for $k=1$) is that the bid-your-max-value strategy gives the agent half her TPS.

\begin{corollary}
	\label{cor:safeStrategy}
	Agent $i$ with additive valuation $v_i$ and a budget of $b_i$ has a strategy  in the bidding game %(Algorithm~\ref{alg:bidding}) 
	that guarantees herself a value of at least $\frac{1}{2} \cdot TPS(b_i )$, regardless of the strategies of the other agents. Hence for any additive valuations and arbitrary entitlements 
	there is an allocation that gives every agent at least half of her TPS.
\end{corollary}
\begin{proof}
Let $z=\truncated{b_i}{v_i}{\items}$.
Let $v_i'$ be the additive valuation where $v_i'(j) = \min( z,v_i(j)) $ for every item $j\in \items$.
By Definition~\ref{def:TPS} it holds that the proportional share with respect to $v_i'$ is $b_i \cdot  v_i'(\items) = z$. % and thus $v_i'(\items) = \frac{z}{b_i}$.
Consider the bid-your-max-value strategy with respect to valuation $v_i'$. Since every item $j\in \items$ has a value of at most $z = b_i \cdot v_i'(\items)$,  Lemma~\ref{lem:safeStrategy} shows that agent $i$ receives a bundle $A_i$ of value $v_i'(A_i) \geq \frac{1}{2} \cdot b_i \cdot v_i'(\items) = \frac{z}{2}$.
Since for every item $j \in \items$ it holds that $v_i(j) \geq v_i'(j) $, then it holds that $v_i(A_i) \geq v_i'(A_i) \geq \frac{z}{2} $, and thus agent $i$ gets at least half of $\truncated{b_i}{v_i}{\items}$.
\end{proof}
 
While the analysis of the strategy that gives a $\frac{1}{2-b_i}$ fraction of the  TPS is rather involved and deferred to the appendix, the strategy itself is not very different from the  bid-your-max-value strategy presented above. 
%The following lemma presents a strategy for agent $i$ with a budget of $b_i$ that she can use in the bidding game (Algorithm~\ref{alg:bidding}) to guarantee herself a value of at least $\frac{1}{2-b_i}$ fraction of her TPS. 
At each round, {bid as follows}. If there is one item that suffices by itself to obtain the goal, bid your whole budget. Else, if two items suffice, bid half your budget (and take both items if you win). 
Otherwise, use the bid-your-max-value strategy. 
% If at least three items are needed, make a bid proportional to the value of the most valuable remaining item, and pick that item if you win. 

We next turn to present the proof of Theorem \ref{thm:app-APS}. 
%\mbc{we should consider moving this proof to the appendix. I think we make all the main points in the other text}

%In the rest of the paper we use the following notations:
\begin{proof}(\emph{of Theorem \ref{thm:app-APS}})
{We present three strategies for agent $i$  that plays in the bidding game, %(Algorithm~\ref{alg:bidding})
each ensures one of the required guarantees. 
Among the suggested strategies, depending on her valuation and entitlement, the agent can choose the strategy for which the corresponding guarantee has the highest value.   % I guess the point that computing which of the 3 to pick does not require understanding which will give the highest value (which might be hard to compute), only which will give the best guarantee. 
%out of the three, ensuring she gets all three  guarantees (note that this is easy to compute)}.
	% MB: this is unlike picking the strategy that gives the  highest value, which might be hard to compute.
	In Lemma~\ref{lem:1bi} we present a strategy that guarantees agent $i$ at least  $\frac{1}{2-b_i}$-fraction of her TPS. 
	%and it runs in polynomial-time.
	In Lemma~\ref{lem:35} we present a strategy that guarantees agent $i$ at least  $60\%$ of her APS.
	%\mbe{The strategy can be computed in polynomial time, if given the value of the APS (which can be computed in pseudo-polynomial time by Proposition \ref{prop:computation-psedu-poly}.)} 
{	In order to guarantee the $\lfloor 1/b_i\rfloor$ %$\Bigl\lfloor \frac{1}{b_i}\Bigr\rfloor$ 
ranked item the agent can use the strategy of bidding $b_i$, and selecting the maximal available item if she wins. She is guaranteed to select one of the top $\lfloor 1/b_i\rfloor$  ranked items, since otherwise the other agents spent a budget of {$\lfloor 1/b_i\rfloor \cdot b_i$} % $b_i\cdot \Bigl\lfloor \frac{1}{b_i}\Bigr\rfloor$ 
which is strictly more than their budget $1-b_i$. %{This strategy can be computed in polynomial time.}
}

One can readily see that all aspects of the strategies above can be implemented in polynomial time, except for one issue.
{While Observation \ref{obs:tps-computation} shows that the TPS can be computed efficiently,}
computing the APS is NP-hard, and hence potentially the agent cannot compute beforehand which of the three guarantees is highest (as she cannot compare $\frac{3}{5}\cdot APS$ with the other two guarantees). Moreover, to run the strategy of Lemma~\ref{lem:35}, it appears that the agent needs to know her own APS. 

There are two ways of addressing this problem. The simple way is to settle for a pseudo-polynomial time algorithm, and use Proposition~\ref{prop:computation-psedu-poly}. However, there is a {better} solution that does give a truly polynomial time algorithm. The basic idea is to show that the strategy of Lemma~\ref{lem:35} does not just guarantee a value of $\frac{3}{5}\cdot APS$, but in fact a value of $\frac{3}{5}\cdot V$, where $V \ge APS$, and moreover, $V$ can be computed exactly in polynomial time. See more details in Appendix~\ref{app:V}.} 
\end{proof}

Lemma~\ref{lem:35} shows that an agent can secure a $\frac{3}{5}$ fraction of her APS. Let us provide here a brief overview of the proof.
In Step~1 of the algorithm, if there is a single item of value at least $\frac{3}{5}\cdot APS$, the agent bids her whole budget, attempting to win the item. In step~2, if there are two items whose sum of values is at least $\frac{3}{5}\cdot APS$, she bids half her budget, attempting to win both items. A difficulty arises if there are items of value strictly less than $\frac{3}{5}\cdot APS$ but strictly more than $\frac{1}{2}\cdot APS$,  the agent bids $\frac{b_i}{2}$ (attempting to win two items), but fails to win these items. In this case the budget of the remaining agents is consumed at a rate that is slower than the rate at which the total value of items decreases. This difficulty is handled in Step~3 of the algorithm. The difficulty is counter-balanced by a positive effect: each of the remaining items has relatively low value, and we have already seen in Lemma~\ref{lem:safeStrategy} that in such cases {the agent can secure herself strictly more than half her TPS (and hence strictly more than half her APS)}. We show that the positive effect is more significant than the negative effect. For this purpose, we introduce Lemma~\ref{lem:3/4}, that for the APS provides better bounds than  Lemma~\ref{lem:safeStrategy} does, though requires stronger assumptions, and its proof is significantly more involved. 
%That part of the strategy, which complete the strategy of getting $\frac{3}{5}$ fraction of her APS, is described in the appendix (Lemma~\ref{lem:3/4}). 

For those readers who prefer not to follow the more involved proof of Lemma~\ref{lem:3/4}, we remark that one can use Lemma~\ref{lem:safeStrategy} instead, without changing anything in the proof of Lemma~\ref{lem:35}. This gives a bound of $\frac{8}{15}\cdot APS$ {(which is already better than half the APS) } instead of $\frac{3}{5}\cdot APS$. 
%\mbc{FINISH. make sure it}
%the $\truncated{\frac{b_i}{B'}}{v_i}{\items'} \geq 2\cdot \truncated{\frac{b_i}{v_i}{2\cdot B'}}{v_i}{\items'}  \geq \frac{4\bval}{5}$, and thus can guarantee a value of at least $\frac{2}{3} \cdot \frac{4\bval}{5}=\frac{8\bval}{15}$.
%\mbc{I have moved this to here. We need to merge it above.}
%\tec{rephrase this paragraph. say that 4.8. is part of c.4}
%if in step~$3$ of the strategy of Lemma~\ref{lem:35},
%if agent $i$ plays the strategy described in Lemma~\ref{lem:safeStrategy} instead, 

%\mbc{new by Tomer, Uri, please read as well.}

Specifically, 
let $\items'$ be the set of remaining items at the beginning of step~3 of the strategy described in Lemma~\ref{lem:35}, and let $B'$ be the total remaining budget of all agents at this point.
By the definition of step~2 of the strategy described in Lemma~\ref{lem:35} % (\mbe{bidding for a pair of enough value}) 
it holds that there are no two items worth together more than $\frac{3}{5} \cdot \anyprice{b_i}{v_i}{\items}$.
Moreover, it holds that the proportional share of agent $i$ with a normalized  budget of $\frac{b_i}{B'}$ over a set of items $\items'$ satisfies:
\begin{eqnarray*}
\proportional{\frac{b_i}{B'}}{v_i}{\items'}  & = & 2\cdot \proportional{\frac{b_i}{2\cdot B'}}{v_i}{\items'} \geq %\\ & \geq & 
2\cdot \anyprice{\frac{b_i}{2\cdot B'}}{v_i}{\items'} \\ &  \stackrel{\eqref{eq:half-anyprice}}{\geq} & \frac{4}{5} \cdot \anyprice{b_i}{v_i}{\items}. 
\end{eqnarray*}
The equality is by definition of proportionality. The first inequality is by Proposition~\ref{prop:APS-pes}.
The second inequality follows directly from  Equation~\eqref{eq:half-anyprice},
{that states  that $\anyprice{\frac{b_i}{2B'}}{v_i}{\items'} \geq\frac{2}{5}\cdot \anyprice{b_i}{v_i}{\items}$
whenever reaching step 3.} %, and proven in Lemma \ref{lem:35}. }
Thus, we can apply at this point the strategy of Lemma~\ref{lem:safeStrategy} for $k=2$ and guarantee a value of at least $ \frac{2}{3}\cdot \proportional{\frac{b_i}{B'}}{v_i}{\items'}  \geq \frac{8}{15} \cdot \anyprice{b_i}{v_i}{\items}.$

%\mbc{this uses notations that I do not think we have defined, like $z,\items',B'$}

    \begin{lemma}
    	\label{lem:35}
    	Agent $i$ with additive valuation $v_i$ over $\items$ and entitlement $b_i$ has a strategy that guarantees a value of at least $60\%$ of $\anypricei$, regardless of the strategies of the other agents. 
    \end{lemma}
    \begin{proof}
    	Let $\bval=\anypricei$. By Definition~\ref{def:anyprice-bundles} there exists a  list of bundles $\mathcal{S}=\{S_j\}_j$ %\ufc{we used $B$ for budget, so perhaps use $S$} 
    	and associated nonnegative weights $\{\lambda_j\}_j$ such that:
    	
    	\begin{itemize}
    		\item $\sum_j \lambda_j = 1$.
    		\item For every item $e$ we have $\sum_{j | e\in S_j} \lambda_j \leq b_i$.
    		\item $v_i(S_j) \ge \bval$ for every {$S_j \in \mathcal{S}$}. %\mbc{we have a notation for these sets, please use it.}
    	\end{itemize}
    	
    	%Let $f_i$ ($s_i$, respectively) denote the value of the first (second, respectively) most valuable item in the beginning of round $i$. We assume that items are consumed in order of decreasing value according to $v$, and the price paid for each consumed item is the bid made by our agent in the corresponding round. These assumptions can be made without loss of generality, because the strategy that we propose for our agent has these properties, and any violation of these properties by the adversary only works in favor of our agent (it either causes the remaining items to have higher value, or causes the adversary to have smaller remaining budget).
    	
    	We propose the following strategy for the agent.
    	
    	\begin{enumerate}
    		\item If $x_i^{(t)} \ge \frac{3\bval}{5}$, bid $b_i^{(t)}$. If the agent wins, she selects {an item with the highest value, among the available items.}  
    		%OLD: the item of highest value. 
    		\item Else, if $x_i^{(t)} +y_i^{(t)} \ge \frac{3\bval}{5}$, bid $\frac{b_i}{2}$. If the agent wins any such round, she selects
    		{a pair of items with the highest total value, among the available items.} % her top two available items.
    		\item Let $t'$ be the first round such that $x_i^{(t')} +y_i^{(t')} < \frac{3\bval}{5}$ (so the pre-conditions for both steps above do not hold). 
    		Let the total remaining budget be $B' =\sum_j b_j^{(t')}$.
    		Consider a new instance in which the budget of every agent $j$ is $b_j=\frac{b_j^{(t')}}{B'}$.
    		Agent $i$ runs the strategy of Lemma~\ref{lem:3/4} on the instance defined by these budgets and the set of items $\items' = \items^{(t')}$. 
%     		\mbc{the next is unclear. Where do we define $B_i^{(t)}$? } Normalize the budgets \tee{use a budget of $b_i'=\frac{b_i}{B^{(t)}}$ where $t$ is the first round where none of the conditions of steps 1 or 2 apply,} and run the strategy of Lemma~\ref{lem:3/4} for the remaining items.
    	\end{enumerate}
    	
    	If the agent wins a bid in either step~1 or step~2 then we are done. Hence we may assume that the agent does not win any such bid.
    	%Now we enter step~3. The purpose of steps~1 and~2 was to guarantee that in step~3 every item has relatively low value, and then Lemma~\ref{lem:3/5} can be applied in order to.
    	%The difficulty in analysing step~3 stems from the possibility that in step~2 the adversary wins items whose value is higher than $\frac{b}{2}$, but pays for every such items only $\frac{b}{2}$. If this happens, the budget of the adversary is consumed at a rate that is lower than the rate at which item values are consumed.
    	%Let $B \ge b$ denote the total budget of all agents remaining at the beginning of step~3. The agent has a budget of $b$ left. %
    %	Consider the items that remained in the beginning of step~3 denoted by $\items'$ and the total remained budget by $B' =\sum_j b_j^{(t')}$.
  %  	
    	%	\begin{proposition}
    	%		\label{pro:chargingthm}
    	We show that if the agent does not win any item in steps~1 and~2, then %$\frac{5}{4}s \ge B$. Moreover,
    	\begin{equation}
    	\gamma\stackrel{\Delta}{=} \anyprice{\frac{b_i}{2B'}}{v_i}{\items'} \geq\frac{2}{5} \cdot \anyprice{b_i}{v_i}{\items}. \label{eq:half-anyprice}
    	\end{equation} 
    	%	\end{proposition}
    	
    	If the agent plays according to the strategy defined in Lemma~\ref{lem:3/4},  the agent guarantees a value of $\frac{3}{2}\cdot \gamma \geq \frac{3\bval}{5}$, as desired.
    	
    	%	\begin{proof}
    	To complete the proof it remains to prove Equation~(\ref{eq:half-anyprice}). 
    	To prove that it holds we 
    	use the  bundles $\mathcal{S}=\{S_j\}_j$ and associated nonnegative weights $\{\lambda_j\}_j$,
    	and
    	construct bundles $\mathcal{S'}=\{S'_j\}_j$ and weights $\{\beta_j\}_j$, as required by Definition~\ref{def:anyprice-bundles}, for entitlement $\frac{b_i}{2B'}$ and share that is at least $\frac{2\bval}{5}$. 
    	%Assume for the sake of contradiction that the agent does not win any item in steps~1 and~2, but no items remain for step~3.  We show that the fact that no items remain implies that the total budget of the adversary and the agent is exhausted, contradicting the fact that the agent's budget is not exhausted.
    	%Recall the bundles $\{B_j\}$ implied by the anyprice share. 
    	Let $Q = Q_1 \cup Q_2$ denote the set of items consumed by the adversary in the first two steps ($Q_1$ for rounds of step~1, $Q_2$ for rounds of step~2). %For every item $q\in Q$, let $p(q)$ denote the price paid for the item. Recall that $p(q) = b_i$ if $q \in Q_1$, and $p(q) = \frac{b_i}{2}$ if $q \in Q_2$. %Let $R$ denote the set of items that remain in the beginning of step~3.
    	For each $S_j \in \mathcal{S}$ 
    	\begin{itemize}
    		\item If $S_j \cap Q_1 \neq \emptyset$ or $|S_j \cap Q_2 |>1 $, then $S_j$ is discarded.
    		\item If $S_j \cap Q_1 = \emptyset$ and $|S_j \cap Q_2 |= 1 $, then $S_j \setminus Q_2$ is added to $\mathcal{S'}$ with weight $\frac{\lambda_j}{2B'}$.
    		\item If $S_j \cap Q_1 = \emptyset$ and $|S_j \cap Q_2 | = 0 $, then $S_j$ is split to $C_j,D_j$ as defined below, and both $C_j$ and $D_j$ are added to $\mathcal{S'}$, each with weight $\frac{\lambda_j}{2B'}$.
    	\end{itemize}
    	The way we split $S_j$ to $C_j,D_j$ is as follows.
    	We start by setting both $C_j,D_j=\emptyset$, and go over  the items in $S_j$ in order of decreasing values.
    	As long as $v_i(C_j)<\frac{2\bval}{5}$ we add an item to $C_j$. When	$v_i(C_j)\geq \frac{2\bval}{5}$, we add all remaining items to $D_j$. 
    	In order to prove that $\gamma \geq \frac{2\bval}{5}$, we need to show that all sets in $\mathcal{S'}$ have a value of at least $\frac{2\bval}{5}$.
    	By definition $v_i(C_j) \geq \frac{2\bval}{5}$.
    	It holds that $v_i(C_j) \leq \frac{3\bval}{5}$ since the first two items in $S_j$ are worth together less than $\frac{3\bval}{5}$ (otherwise we are not finished with step~2). If their value is at least $\frac{2\bval}{5}$, then no more items will be added to $C_j$. Else, the second item has a value of at most $\frac{\bval}{5}$, and therefore any other item in $S_j$ will increase the value of $C_j$ by  at most $\frac{\bval}{5}$.
    	Thus, $v_i(D_j) = v_i(S_j) - v_i(C_j) \geq \bval -\frac{3\bval}{5}=\frac{2\bval}{5}$.
    	The second type of sets that are added to $\mathcal{S'}$, are sets $S_j \setminus Q_2$ where $|S_j \cap Q_2| = 1$ (and $|S_j \cap Q_1| = 0$).
    	For these sets, notice that since items in $Q_2$ have a value of at most $\frac{3\bval}{5}$, then it holds that $v_i(S_j\setminus Q_2) \geq v_i(S_j) -\frac{3\bval}{5} \geq \bval-\frac{3\bval}{5}=\frac{2\bval}{5}$.
    	Observe that  for every item $e \in \items'$, if it belongs to $S_j$, then the weight for the corresponding sets added to $\mathcal{S'}$ is at most $\frac{\lambda_j}{2B'}$. This is since we either discard this set, or $e$ is added to exactly one set (notice that $e$ cannot be in both $C_j$ and $D_j$).
    	Thus, the total weight of every element is at most $\frac{b_i}{2B'}$ as needed.
    	It remains to show that the sum of weights of the sets added to $\mathcal{S'}$ is at least 1.
    	
    	Let $\alpha_1= \sum _{j\; | \; S_j \cap Q_1 \neq \emptyset}\lambda_j$ denote 
    		the sum of weights of all sets that have items from $Q_1$. Let $\alpha_2= \sum _{j \; | \; 
    		|S_j \cap Q_2| > 1 \wedge  
    		|S_j \cap Q_1| =0} \lambda_j$ denote the sum of weights of all sets that have at least two items from $Q_2$ and no item from $Q_1$. Let $\alpha_3= \sum_{j \; | \; 
    		|S_j \cap Q_2| =1 \wedge 
    		|S_j \cap Q_1| =0}\lambda_j$ denote the sum of weights of all sets that have one item from $Q_2$ and no item from $Q_1$.
    		
    	Observe that $\alpha_1 \leq b_i \cdot |Q_1|$, since every item is in at most $b_i$ weight of the sets.
    	Likewise, $2\cdot \alpha_2 +\alpha_3 \leq b_i \cdot |Q_2|$. 
    	The sum of weights of all sets that have no items from $Q_1 \cup Q_2$ is at least $\alpha_4 \geq 1-\alpha_1-\alpha_2-\alpha_3 $.
    	
    	In addition, it holds that $B'=1-|Q_1| \cdot b_i- |Q_2| \cdot \frac{b_i}{2}$ since every item in $Q_1$ was paid $b_i$, and every item in $Q_2$ was paid $\frac{b_i}{2}$.
    	The sum of weights of $\mathcal{S'}$ is $$\frac{\alpha_3}{2B'} +2 \frac{\alpha_4}{2B'} \geq \frac{\alpha_3 + 2-2\alpha_1-2\alpha_2-2\alpha_3}{2B'} \geq \frac{ 2-2b_i\cdot |Q_1|-b_i\cdot |Q_2|}{2(1-|Q_1| \cdot b_i- |Q_2| \cdot \frac{b_i}{2})}  = 1, $$
    	which proves inequality~(\ref{eq:half-anyprice}) and concludes the proof.
    \end{proof}
    
    The next lemma (whose proof is deferred to Appendix~\ref{app:allocation-game}) % MB: we do not list missing proofs for each claim.
    states that there is a strategy for agent $i$ that guarantees herself a value of $\frac{3}{2}\cdot \anyprice{\frac{b_i}{2}}{v_i}{\items}$. This Lemma (with a suitable assignment to the %modified 
    parameters $b_i$ and $\items$) is used in step~3 of Lemma~\ref{lem:35}.% to guarantee a value of at least $\frac{3}{5}\cdot\anyprice{b_i}{v_i}{\items}$. 
     
    \begin{restatable}{lemma}{lemThreeFour}
	\label{lem:3/4}
	Agent $i$ with additive  valuation $v_i$ and  entitlement   $b_i$   %such that $\anyprice{v_i}{b_i/2}{\items}\geq \bval$ 
	has a strategy that guarantees herself a value of at least $\frac{3}{2}\cdot \anyprice{\frac{b_i}{2}}{v_i}{\items}$ in the bidding game %(Algorithm~\ref{alg:bidding})
	that allocates $\items$, regardless of the strategies of the other agents.
	%\ufc{The previous set of words is not a sentence in English.}
%\end{lemma}
\end{restatable}
\section{APS Approximation for Additive Agents: Equal Entitlements}
\label{sec:greedy-efx}
In this section we consider approximating the APS in the special case of equal entitlements. 
For only two agents with equal entitlements,  there is always an allocation giving both their APS, which equals to their MMS (Observation \ref{obs:APS-two-equal}).
We saw (Proposition \ref{prop:shares-comp}) that beyond two agents, even for equal entitlements, the APS is different than the shares defined in prior work (like MMS).
As the APS is at least as large as the MMS (Proposition \ref{prop:shares-comp}), and as \citet{KurokawaPW18} showed that for three agents with equal entitlements, simultaneously  giving every agent her MMS is not possible, we cannot hope to give every agent her APS exactly, and thus look for an approximation.
While for arbitrary entitlements we saw that each agent can guarantee herself at least a $\frac{3}{5}$ fraction of the APS (Theorem \ref{thm:app-APS}),  in the equal entitlements case we 
show that the \textit{greedy-EFX} algorithm proposed by \citet{BK20} gives a stronger guarantee.

In this section we prove the following theorem.
\begin{theorem}
	\label{thm:BK17}
	When $n$ agents have additive valuations and equal entitlements, there is an allocation  that gives every agent at least the minimum of the following two values: a $\frac{3}{4}$ fraction of her APS, and a $\frac{2n}{3n-1}$ fraction of her TPS. 
	Moreover, such an allocation can be found in polynomial time.
\end{theorem}
% (compared to $3/5$ for the unequal case).    

%We do so by analyzing the outcome of the \textit{greedy-EFX} Algorithm proposed by \tec{CITE \cite{ Barman and Krishna}}.

Recall that for $n \le 2$ there is an allocation that gives every agent at least her AnyPrice share (see Observation~\ref{obs:APS-two-agents}). Hence there is interest in Theorem~\ref{thm:BK17} only when $n \ge 3$. 
Note that Theorem \ref{thm:BK17} ensures that each agent always gets at least $\frac{2}{3}$ of  her APS (as the TPS is always at least the APS).

In proving Theorem~\ref{thm:BK17}, we shall follow an approach of~\citet{BK20} who proved the existence of allocations that give every agent at least a $\frac{2n}{3n-1}$ fraction of her MMS. The algorithm is identical to that of~\citet{BK20}, but the proof of the key lemma (Lemma~\ref{lem:BK17}) is different, as we need to establish stronger guarantees than those established in~\citep{BK20}.

The first step of the proof of Theorem~\ref{thm:BK17} is a reduction to {\em ordered} instances. Given an instance of an allocation problem with additive valuations and a fixed order $\sigma$ among the items (from~1 to $m$), its {\em ordered version} is obtained by each agent permuting the values of items so that values are non-increasing in $\sigma$. The following theorem is due to~\citep{bouveret2016characterizing}.

%\begin{theorem}
\begin{restatable}{theorem}{thmordered}	(\citep{bouveret2016characterizing})	   
	\label{thm:BL14}
	For every instance with additive valuations, every allocation for its ordered version can be transformed in polynomial time to an allocation for the original instance. Using this transformation, every agent derives at least as high value in the original instance as derived by the allocation for the ordered instance.
\end{restatable}

For completeness, we sketch the proof of Theorem~\ref{thm:BL14} in Appendix~\ref{app:greedy-efx}.

Given Theorem~\ref{thm:BL14}, it suffices to prove Theorem~\ref{thm:BK17} in the special case in which the input instance is ordered (and $n \ge 3$). For such instances
we use an algorithm proposed by~\citep{LiptonMMS04,BK20}, that we refer to as {\em Greedy-EFX}.  Greedy-EFX assumes that the instance is ordered, with item $e_1$ of highest value and item $e_m$ of lowest value. 

% \mbc{As we have an informal definition of Greedy-EFX we can consider moving the formal algorithm to the appendix and only point to it here (it is not our contribution).}

The Greedy-EFX works as follows (see formal definition  in Appendix \ref{app:greedy-efx}).
% \mbc{for the arXiv I think it is better we move it back to here.}
Each bundle is associated with one agent. Greedy-EFX proceeds in rounds.  In rounds~1 to $n$ the items $e_1$ up to $e_n$ are placed in bundles $B_1$ up to $B_n$ respectively. At a round $r > n$, if there is an agent whose bundle no one envies, then $e_r$ is added to her bundle. Else, if there is no such agent, then there must be {envy cycles (in such an envy cycle, each agent envies the next in the cycle).
An envy cycle can be resolved by rotating the sets along the cycle, each agent getting the set of the agent she envies.
The cycles are iteratively resolved until there are no more envy cycles. At this point, there must be an agent whose bundle no one envies, and then $e_r$ is added to her bundle.}

Clearly, the greedy-EFX algorithm runs in polynomial time. Hence in order to prove Theorem~\ref{thm:BK17}, it suffices to prove the following lemma.

\begin{restatable}{lemma}{thmgreedyEFX}		   
	\label{lem:BK17}
	When $n$ agents have additive valuations and equal entitlements, the greedy-EFX allocation gives every agent at least the minimum of the following two values: a $\frac{3}{4}$ fraction of her APS, or a $\frac{2n}{3n-1}$ fraction of her TPS.
\end{restatable}

The proof of Lemma~\ref{lem:BK17} is based on the following principles. Consider an arbitrary agent $i$. We may assume that {under Greedy-EFX agent} $i$  receives at least one item of positive value, as otherwise her APS is~0 {(as there are less than $n$ items she values).} Without loss of generality, assume that the final value received by $i$ is~1. The main part of the proof focuses on one particular item, $e_k$, which is the item of smallest value among those that serve as a first item to enter a bundle that eventually has more than two items. If $v_i(e_k) \le \frac{3}{4}$, a fairly simple argument shows that the TPS is at most $\frac{3n-1}{2n}$, proving that $i$ received at least a $\frac{2n}{3n-1}$ fraction of her TPS. If $v_i(e_k) > \frac{3}{4}$ then we show that the APS is less than $\frac{4}{3}$, and hence that $i$ gets more than a $\frac{3}{4}$ fraction of the APS. To upper bound the APS, we exhibit prices that certify this upper bound. In fact, the prices that we exhibit depend on the exact value of $e_k$, with one set of prices if $\frac{3}{4} < v_i(e_k) \le \frac{5}{6}$, and a different set of prices when $v_i(e_k) > \frac{5}{6}$. The full proof of Lemma~\ref{lem:BK17} appears in Appendix~\ref{app:greedy-efx}.

%Algorithm~\ref{alg:greedy-efx} is guaranteed to return an EFX allocation (note that after applying Theorem~\ref{thm:ordered}, the allocation may not remain EFX).

%We next prove that for $n\geq 3$ agents the allocation returned by Algorithm~\ref{alg:greedy-efx} gives every agent $\frac{2}{3}$ of his APS (and thus also gives every agent $\frac{2}{3}$ of his APS after applying Theorem~\ref{thm:ordered}). Note that for the case that $n=2$, Observation \ref{obs:APS-two-agents} shows that there is always an allocation that gives both agents at least their APS.
%\begin{theorem}
%\begin{restatable}{theorem}{thmgreedyEFX}		   
%	For any instance with $n\geq 3$ agents with additive valuations \vals\  and  equal entitlements ($b_i=1/n$ for every $i$), there exists an allocation in which every agent $i$ gets at least 
	% $(2/3)\anypricei$. 
%	$\frac{2n}{3n-1}$   % \frac{b_i}{2-b_i}$ 
%	fraction of her AnyPrice share. Moreover, such an allocation can be computed in polynomial time. 
%\end{restatable}

  {An immediate corollary of Lemma~\ref{lem:BK17} is that for the case of equal entitlements (or for an entitlement which is an inverse of an integer), the MMS is at least $\frac{3}{4}$ of the APS. 
  % the gap between the MMS and the APS is at most $\frac{3}{4}$.
  \begin{corollary}
  For every additive valuation $v_i$ and an entitlement $b_i =\frac{1}{n}$ for some integer $n$, it holds that $$\MMSfull{n} \geq \frac{3}{4} \cdot \anypricei  $$ 
  \end{corollary}
  \begin{proof}
  Let $S$ be the allocation returned by greedy-EFX when all agents have the same valuation $v_i$.
  By Lemma~\ref{lem:BK17} every agent receives a value of at least $\frac{3}{4}$ of the APS. % Thus the least satisfied agent receives a value of $\frac{3}{4} \cdot APS$. 
  By definition of the MMS, the least satisfied agent received at most the MMS.
  This concludes the proof.
  \end{proof}
  }

We next show that the analysis of the greedy-EFX Algorithm %~\ref{alg:greedy-efx} 
is essentially tight. As the upper bound claimed in the following proposition applies also to the MMS, the proposition may have been previously known. However, we are not aware of an explicit reference to this result, and hence we state it and provide its proof in the appendix. % \ufc{I commented out the rest of the paragraph, as I did not want to clarify the distinction between more than 2/3, and a constant larger than 2/3.}
%, and a better constant than $\frac{2}{3}$ cannot be proven for this algorithm. This shows that improving over the $2/3$ approximation will require a different algorithm.

%\begin{proposition}
\begin{restatable}{proposition}{EFXupper}
\label{pro:EFXupper}
	For every $\epsilon> 0$, there is an instance {with additive valuations and equal entitlements} in which the greedy-EFX algorithm %~\ref{alg:greedy-efx} 
	does not provide
	some agent more than a $\frac{2}{3}+\epsilon$ of her MMS (and thus also does not provide more than a $\frac{2}{3}+\epsilon$ of her APS). 
%\end{proposition}
\end{restatable}

In the appendix we also show that if all agents have 
identical additive
valuation functions (and equal entitlements), then greedy-EFX gives every agent at least a $\frac{3}{4}$ of her APS. This ratio is the best 
possible guarantee that holds for greedy-EFX,
even for MMS, as we show that for every $\epsilon > 0$ there are such instances in which greedy-EFX does not give an agent more than a $\frac{3}{4} + \epsilon$ fraction of her MMS.

% (In our example, $n$ is exponential in $\frac{1}{\epsilon}$.)\mbc{why is this important to note? }

%\mbc{to be removed once we are  done with the section: }
%\begin{itemize}
	%\item entitlement =1/2 implies maximin = anyprice (3.7, also 3.17) \te{do we want to keep both proofs?}. So cut-and-choose give both equally entitled  agents their anyPrice share. 
%	\item greedy-EFX description \tec{done}
 %	\item greedy-EFX gives 2/3 of  anyprice (3.11) \tec{done}
 %	\item tightness of analysis: greedy-EFX upper bound of 2/3 for maximin share (3.12) and thus for anyPrice. 
 %	\item greedy-EFX for agents with \emph{identical} valuations gives 3/4 of  anyprice (3.13) \mbc{not sure we want to write this one. IF the valuations are identical and entitlments are equal, MMS seems good and can be given. So this  might work against us.}
 %	\item tightness of analysis: greedy-EFX upper bound for identical agents of 3/4  for maximin share (3.14) and thus for anyPrice. 
 
%\end{itemize}

\section{Discussion}\label{sec:discussion}
% \mbc{read and make sure the issue of chores is clear (if we keep chores as a section) } 

In this paper we have studied fair allocation for agents with arbitrary (not necessarily equal) entitlements. We have introduced the concept of the AnyPrice share (APS) and have presented a positive constant fair-share  approximation result for agents with arbitrary entitlements,  showing that for agents with additive valuation functions there always is an allocation that gives each agent a $\frac{3}{5}$ fraction of her APS. Following the statement of Theorem~\ref{thm:intro-unequal}, we claimed that (as far as we know), this is the first positive result that guarantees a constant fraction of {\em any} share to agents with arbitrary entitlements. We clarify here what we mean by the term ``share".

\subsection{The Notion of a {\em Share}}
Intuitively,  a share of an agent with some valuation and relative entitlement, represents what value she can reasonably expect to receive in any setting in which the items are partitioned, without any further knowledge (as the valuations of others, or the way entitlements are distributed among others).

Consider an allocation setting with a set $\items$ of items, in which every valuation function $v_i$ is normalized ($v_i(\emptyset)=0$) and monotone ($v_i(S) \le v_i(T)$ for every $S \subseteq T \subseteq \items$). 
We think of a share as a function $f$ that maps the valuation function $v_i$ of an agent and her normalized entitlement $b_i$ (where $b_i\geq 0$ and %$0 \le b_i \le 1$ and 
$\sum_i b_i=1$) to a value $f(v_i,b_i)$, where this function $f$ needs to satisfy some natural properties, such as the following: 

\begin{itemize}
    \item {\em Normalization:} $f(v_i,b_i)\ge 0$, with equality if (though not necessarily only if) either $v_i$ is identically~0, or $b_i = 0$.
    \item {\em Boundedness:} for every $v_i$ and $b_i$ it holds that $f(v_i,b_i) \le v_i(\items)$, with equality if (though not necessarily only if) $b_i = 1$.
    \item {\em Entitlement monotonicity:} for every $v_i$ and $b'_i > b_i$ it holds that $f(v_i,b'_i) \ge f(v_i,b_i)$.
    \item {\em Valuation monotonicity:} for every $v'_i \ge v_i$ (meaning that $v'_i(S) \ge v_i(S)$ for every $S \subseteq \items$) and every $b_i$ it holds that $f(v'_i,b_i) \ge f(v_i,b_i)$.
    \item {\em Value Scaling:}  if the valuation function is scaled by a multiplicative factor of $c$, then the share is scaled by the same multiplicative factor. (Note that we do not require scaling with respect to the entitlement.)
\end{itemize}

Of course, for a notion of a share to be useful, we want $f$ to have additional properties beyond those listed above. It should strike a good balance between being as large as possible, yet still allowing for allocations that give every agent some substantial fraction of her share. We shall not attempt to rigorously define what makes a notion of a share useful.

When one defines a share as a function satisfying all the above properties, then 
%Using the above definition of a share, 
the proportional share, the pessimistic share and our AnyPrice share (APS) are indeed shares. Likewise, the maximin share (MMS) is a share, if one restricts the entitlements $b_i$ to be of the form $b_i = \frac{1}{k_i}$ for an integer $k_i > 1$. The weighted maximin share (WMMS) {of \citet{farhadi2019fair}} is not a share under our definition, as its value depends not only of $v_i$ and $b_i$, but also on how the remaining $1 - b_i$ entitlement is distributed among the remaining agents. However, we may relax our definition of a share so that it also includes WMMS (this requires a certain adaptation of the entitlement monotonicity property). Among all the above notions of shares, indeed Theorem~\ref{thm:intro-unequal} is the first positive result that guarantees a constant fraction of a share to agents with arbitrary entitlements. This guarantee is given with respect to the AnyPrice  share, and hence it holds also with respect to the pessimistic share (which is never larger than the AnyPrice  share).

To further clarify our use of the term share, let us briefly discuss Prop1, which is a ``close relative" to the proportional share. An allocation is Prop1 if every agent gets at least her proportional share, up to one item. Prop1 is a property of a solution rather than a mapping of $v_i$ and $b_i$ to a value, and hence it is not a notion of a share. However, one can try to define a share based on Prop1. A natural definition would be a value equal to the proportional share, minus the value of the most valuable item (namely, $b_i v_i(\items) - \max_{e\in \items} v_i(e)$), as this is a lower bound on the value that the agent gets in any Prop1 allocation.\footnote{An alternative definition of a value based on Prop1 would be as the smallest value of a bundle, among all bundles for which adding one item not in the bundle makes the value of the bundle larger than the proportional share. The distinction between these two definitions is not important for our discussion.}
However, this definition does not satisfy valuation monotonicity (by increasing the value of the most valuable item, the ``share" value decreases rather than increases). Hence we do not view neither Prop1 nor its derivatives as a notion of a share. 

\subsection{Interpretation of the AnyPrice  Share as an Analog of {\em Cut and Choose}}

%Response: we view APS as fundamentally different from CEEI. The APS is a notion of a share, and depends only on the valuation function of an agent and her entitlement. CEEI is a solution concept that depends on valuations of all agents, the entitlements of all agents, and furthermore, does not always exist (unlike cut and choose). We do not see a simple way of presenting APS as being motivated by CEEI. 

The MMS can be thought of as being motivated by the {\em cut-and-choose} paradigm, where the agent plays the role of the ``cutter". The value of the MMS for an agent is the maximum value she can secure to herself if she were to partition (``cut") the grand bundle into $n$ bundles, and allow every other agent to choose a bundle before she does. The APS can also be thought of as being motivated by the cut-and-choose paradigm, as we explain below. 

Recall that we provided two equivalent definitions for the {APS}.  Definition~\ref{def:anyprice-bundles} (that we referred to as the dual definition) can be thought of as a natural variation on the definition of MMS. Instead of insisting on an integral partition into bundles, one allows a fractional partition, with the advantage that fractional partitions can be easily extended to situations in which agents do not have equal entitlements. In this respect, the APS is motivated by MMS, and hence, indirectly by the cut-and-choose paradigm. However, even in the case of equal entitlements, the APS is not the same as the MMS, and so some of the cut and choose motivation (with the agent as the cutter) may be lost in this adaptation. 

Consider now the price based definition of APS, Definition~\ref{def:anyprice-prices}. This definition can be viewed as being motivated by the cut-and-choose paradigm, but with the agent playing the role of the chooser, not the cutter. 
Recall that the APS depends only on the valuation function of the agent and her entitlement, but not on valuation functions and entitlements of other agents. Hence pretend that there are no other agents. Instead, there is just a single current owner of all items, and our agent that has a $b_i$ entitlement to the items. Our agent comes to claim  her fair share of items from this owner, whereas the owner is interested in keeping the items for himself. Which items should our agent be allowed to take? Here we may use a price-and-choose paradigm, in the spirit of the cut-and-choose paradigm. As the situation is not symmetric (there is a current owner of the item and our agents that has claims on the items), the current owner is the one who is allowed to price the items in any way he wishes, provided that the prices are non-negative and sum up to~1. Now our agent with a budget of $b_i$ plays the role of the chooser, and is allowed to choose any bundle that she can afford. Her APS is the maximum value that she can guarantee herself as a chooser in such price-and-choose games.

\subsection{Directions for Future Work and Open Questions}

Several approximability results appearing in this paper provide constant factor approximation ratios that are not known to be best possible. It would be interesting to improve over any of them. We list the known lower bounds and upper bounds for three such results (all for additive valuation functions).

\begin{enumerate}
    \item What is the largest possible gap between the APS and the MMS when agents have equal entitlements (or equivalently, when the entitlement $b_i$ is the inverse of a positive integer)? Theorem~\ref{thm:BK17a} shows that in this case the value of the MMS is at least a $\frac{3}{4}$ fraction of the value of the APS, and Lemma~\ref{lem:APSlargerPESS} shows that sometimes the ratio is no better than $\frac{96}{97}$.
    \item How much of the APS can one guarantee to agents with equal entitlements? Theorem~\ref{thm:BK17} shows that there always is an allocation that gives every agent at least a $\frac{2}{3}$ fraction of her APS. In~\citep*{FST21} it is shown that there are allocation instances for which there is no allocation that gives every agent more than a $\frac{39}{40}$ of her MMS. This negative result applies also to the APS (as the APS is at least as large as the MMS).
    \item How much of the APS can one guarantee to agents with arbitrary entitlements?  Theorem~\ref{thm:intro-unequal} shows that there always is an allocation that gives every agent at least a $\frac{3}{5}$ fraction of her APS. The negative result of $\frac{39}{40}$ still applies (as equal entitlements is a special case of arbitrary entitlement).
\end{enumerate}

A natural direction for future work for additive valuations is to obtain best-of-both-worlds results for agents with arbitrary entitlements, for example, simultaneously achieving ex-ante proportionality as well as APS approximation ex-post. Such results are known for agents with equal entitlements \citep*{babaioff2021bestofbothworlds}.

More generally, one would like to study the APS for additional classes of valuation functions, such as submodular valuations. % Additionally 

\bibliographystyle{abbrvnat}

\bibliography{bib}

\appendix
\section{Some Formal Definitions}
\label{sec:missing-defs}

In this section we present some formal definitions that are omitted from the body of the paper. 

\begin{definition}[Competitive Equilibrium (CE)]\label{def:CE}
	Given the valuations of the agents $\vals=(v_1,\ldots,v_n)$, and a vector of budgets $\budgets=(b_1,\ldots,b_n)$, a \emph{Competitive Equilibrium (CE)} is a pair $(A,p)$ of an allocation $A$ {of all items} and a vector of item prices $p$, where the price of a bundle received by each agent is at most her budget (i.e., $p(A_i) = \sum_{j\in A_i} p_j \leq b_i$), and every agent's bundle is utility maximizing under her budget (i.e., for all $S\subseteq \items$, $p(S) \leq b_i \implies v_i(S) \leq v_i(A_i)$).   
\end{definition}

\begin{definition}[Weighted maximin share (WMMS)]
The \emph{weighted maximin share (WMMS)} of agent $i$ with valuation $v_i$ over $\items$,  when the vector of entitlements is $(b_1,b_2,\ldots,b_n)$, which is denoted by $WMMS_i(b,v_i,\items)$,	
	is defined to be the maximal value $z$ such that there is an allocation that gives each agent $k$ at least $\frac{z\cdot b_k}{b_i}$ according to $v_i$.
Formally:
 $$WMMS_i(b,v_i,\items) = \max_{(A_1,\ldots,A_n)\in\mathcal{A}  }\  \min_{k \in[n] } \frac{b_i \cdot v_i(A_k)}{b_k}$$
%		where $\mathcal{A}$ is the set of all allocations of $\items$ to the $n$ agents.
		%\mbc{we have used $\mathcal{P}$ above to denote the simplex of prices. change this notation. I think it can simply go over all allocations. }
\end{definition}

\section{Properties of The Truncated Proportional Share (TPS)}\label{app:TPS}
Observation \ref{obs-TPS-mush-larger-APS} shows that the TPS might be much larger than the APS.

\begin{observation}\label{obs-TPS-mush-larger-APS}
	For any $\varepsilon>0$ 
	there exist an additive valuation and an entitlement $b_i$ such that $\varepsilon\cdot \TPSbi\geq \anypricebi$.
\end{observation}
\begin{proof}
	Consider the case of entitlement $b_i=\frac{1}{k+\epsilon}$ % \mbc{I think this entire argument should work if we change to  $b_i=\frac{1}{k}$, right? If so it will be good, as it says the claim is true even for equal entitlements  }
	and $k+1$ items, $k$ of them with value $b_i$, and one item with value $r=\frac{\epsilon}{k+\epsilon}$.
	The APS is $r=\frac{\epsilon}{k+\epsilon}$, since if we set the prices of the high value items at price $\frac{1}{k}$, while pricing the low value item at $0$, the agent can only afford the low value item. The APS is not lower since the agent can always afford at least one item.
	In contrast, the TPS is $ \frac{1}{k+\epsilon}$, since no item is worth more than the entitlement, so the TPS is the same as the proportional share.
\end{proof}

The next observation shows that the TPS can be computed efficiently.

\begin{observation}
\label{obs:tps-computation}
	Consider the problem of computing, for any integer additive valuation $v_i$ and entitlement $b_i$ (were $b_i \leq 1$  is a rational number), the truncated proportional $\truncatedi$.
	There is an algorithm that computes the TPS in time polynomial in the input size.
\end{observation}
\begin{proof}
%By Definition~\ref{def:TPS}, the TPS can be computed efficiently %(in contrast to the APS, see Proposition \ref{prop:computation}).
{Compute the proportional share $PS_i = b_i \cdot v_i(\items)$. Let $j$ denote an item of highest value under $v_i$. If $v_i(j) \le PS_i$ then the TPS is $PS_i$. If $v_i(j) > PS_i$ then there are two cases to consider.
In one case, $b_i \ge \frac{1}{2}$. In this case the TPS is $z = \frac{b_i}{1-b_i} \cdot v_i(\items \setminus \{j\})$. This can be verified, as after reducing the value of $v_i(j)$ to $z$, no item other than $j$ has value larger than $z$ (because $\frac{b}{1 - b} \ge 1$), and it indeed holds that $z = b_i\cdot (\min[z, v_i(j)] + v_i(\items \setminus \{j\})) = b_i(z + \frac{1-b_i}{b_i}z) = z$. In the other case $b_i < \frac{1}{2}$. In this case
Lemma~\ref{lem:ktps} implies that the TPS equals  $\frac{b_i}{1-b_i} \cdot \truncated{\frac{b_i}{1-b_i}}{v_i}{\items\setminus \{j\}}$. As in this case the number of items decreases, it takes no more than $m$ iterations to compute the TPS.} 
\end{proof}

%\tec{rephrase this paragraph}  Additionally, Proposition \ref{prop:TPS-APS} shows that TPS is at least as large as APS. These properties seem to suggest that TPS is a better share definition than APS. Yet, the main drawback of TPS is that, similarly to the proportional share, it seems unattractive beyond additive valuation: recall the unit-demand with identical items example in Section \ref{sec:unit-demand}, it holds for TPS as well as it does for the proportional share.
%We next show that a competitive equilibrium (CE)  does not necessarily guarantees that every agent get her TPS, unlike the case for APS (in which in every CE every agent gets her APS).  \mbc{ are there more drawbacks?} \tec{we had this discussion in section 4} \mbc{only partially. I think the computation part is not there. Tomer, can you draft a shorter paragraph with only the part missing from section 4?}
 
%The next proposition shows that in contrast to the APS share, a competitive equilibrium (CE)  does not necessarily guarantees that every agent get her truncated proportional share

\begin{proposition}\label{prop:TPS-CE}
	There exists an instance with additive valuations $\vals$ and entitlements $\budgets$ such that no allocation gives every agent her  truncated proportional share (or her proportional share), yet a competitive equilibrium with entitlements as budgets does exist.
%	There exist additive valuations $\vals$ and entitlements $\budgets$ such that for some  competitive equilibrium with the entitlements as budgets,  there exists a player  that does not get his  truncated proportional share.	
\end{proposition}  
\begin{proof}%\mbc{slightly revised}
	Consider an instance with two agents $1,2$, and three identical items $\items =\{a,b,c\}$. Every agent has value of $\frac{1}{3}$ for each item.
	The agents entitlements are $b_1=0.4, b_2=0.6$.
	%Both agents' valuations are $v_i(a)=v_i(b)=v_i(c)=\frac{1}{3}$. 
	The allocation where agent 1 receives item  $a$ while agent 2 receives both $b$ and $c$, with prices $p_a =0.4,p_b=p_c=0.3$, is a CE, and thus competitive equilibria exist.
	Moreover, in every CE,  agent~2 receives at least as many items as agent~1 (since she has a larger budget), and thus, she must receive two items, as all items must be allocated in a CE.
	Thus, in every CE of this instance, agent~1 receives a value of $\frac{1}{3}$. % which is strictly less than his TPS.
	However, in this instance the TPS of every agent is the same as her proportional share (and the same as the entitlement). So
	$\truncated{b_1}{v_1}{\items} =0.4>\frac{1}{3}$, and agent 1 gets less than her TPS in every CE.
\end{proof}

\citet{babaioff2021bestofbothworlds} observe that even for equal entitlements, it is not possible to give every agent more than half her TPS. More precisely, for any constant $\rho$ larger than half, 
there is an instance (with equal entitlements) in which some agent does not get $\rho$-fraction of her TPS.

\begin{observation} \label{obs:TPS-app-UB}
	For every $n $, there exists an instance with $n$ agents with equal entitlements and identical valuations, such that in every allocation at least one agent gets at most $\frac{n}{2n-1}$-fraction of her TPS. 
\end{observation}
\begin{proof}
	%Fix any $n$ large enough so that $\frac{n}{2n-1}<\rho$.
	Consider an instance with $n$ agents and $2n-1$ identical items, each of value $\frac{1}{2n-1}$ for all agents.
	The TPS of every agent is $\frac{1}{n}$, while in every allocation there is at least one agent that receives only one item, so her value is at most $\frac{1}{2n-1} \leq  \frac{n}{2n-1}\cdot \frac{1}{n}$
	, as desired.
\end{proof}

\section{Missing Proofs from Section \ref{sec:aps}}
%\mbc{add section number}
\propDefEq*
\begin{proof}
Let $z_D$ (resp., $z_P$) be the value according to Definition~\ref{def:anyprice-prices} (resp., Definition~\ref{def:anyprice-bundles}).
We first prove {that} $z_D \geq z_P$. 
Let $\vect{p}$ be the corresponding minimal prices with respect to $z_D$.
Assume towards contradiction $z_{P}>z_D$. By Definition~\ref{def:anyprice-bundles}, there {exists} a distribution $F$ over sets $T_1,\ldots,T_k$, with weights $\lambda_1,\ldots,\lambda_k$ (i.e., the probability of $T_r$ according to $F$ is $\lambda_r$) such that $\sum_r \lambda_r=1$ and for all $r$, {$v_i(T_r) \geq z_P > z_D$} , such that for every item $j\in \items$, $\sum_{r: j \in T_r} \lambda_r \leq b_i$.
It holds that $$E_{T \sim F}[\vect{p}(T)] = \sum_{j\in\items} p_j \cdot \Pr_{T \sim F}[j \in T] \leq \sum_{j\in\items} p_j \cdot b_i =b_i,$$
where $\vect{p}(T)$ is the price of bundle $T$ according to $\vect{p}$. 
Thus, there must be a set $T_r$ in the support of $F$ that is priced at most $b_i$ and has a value strictly greater than $z_D$, which leads to a contradiction.

To see the other direction, assume that  $z_D>z_P$ and fix $z$ satisfying $z_D>z>z_P$. Consider  the following  LP over the variables $\{\lambda_T \mid T \in \goodsetsiz\}$:

{\bf Maximize} $\sum_{T \in \goodsetsiz} \lambda_T$ {\bf subject to}:
\begin{itemize}
		\item $\sum_{T: j\in T} \lambda_T \leq b_i$ for all $j\in \items$.
	\item  $\lambda_T \geq 0$ for all $T \in \goodsetsiz$.	
\end{itemize}
Note that since $z>z_P$, the solution for this LP is strictly less than $1$.
The dual of this LP is over the variables $\{q_j \}_{j\in\items}$:

{\bf Minimize} $b_i \cdot \sum_{j \in \items} q_j$ {\bf subject to}:
\begin{itemize}
	\item $\sum_{j \in T} q_j \geq 1$ for all $T \in \goodsetsiz $.
	\item  $q_j \geq 0$ for all $j \in \items$.
\end{itemize}

By LP duality we get that the minimum is also strictly less than 1. By replacing $p_j=q_j \cdot b_i$ we get:

{\bf Minimize} $\min  \sum_{j \in \items} p_j$ {\bf subject to}:
\begin{itemize}
	\item $\sum_{j \in T} p_j \geq b_i$ for all $T\in \goodsetsiz$.
	\item  $p_j \geq 0$ for all $j \in \items $.
\end{itemize}
This means that the minimum of the last LP is also strictly less than 1.  Thus, if we consider the prices $p_j+\epsilon$ for small enough $\epsilon$ for $p_j$ that minimizes this LP, it holds that $\sum_j (p_j+\epsilon ) \leq 1$, and every (non empty) bundle of value at least $z $ (which is strictly more than $0$), is priced strictly more than $b_i$. Thus we get that   $ z_D$ must be smaller than $z$, a contradiction.
\end{proof}

\propSmallSupport*
\begin{proof}
	As shown in the proof of Proposition~\ref{prop:anyprice-eqe}, for the AnyPrice share $z$, it holds that the LP over the variables $\{\lambda_T \mid T\in\goodsetsiz \}$ %\mbc{ isn't $Var^{P}$  the same as $\goodsetsiz$? } such that 
	
	{\bf Maximize} $\max  \sum_{T \in \goodsetsiz} \lambda_T$  {\bf subject to}:
	\begin{itemize}
		\item $\sum_{T: j\in T} \lambda_T \leq b_i$ for all $j\in \items$.
		\item $\lambda_T \geq 0$ for all $T \in \goodsetsiz$.	
	\end{itemize}
	has solution of at least $1$ (since $z$ is feasible).
	There are only $m$ constraints (other than positivity constraints), thus there is a solution where at most $m$ variables are not zero.
\end{proof}

\propAPSPES*
\begin{proof} %\mb{change p and q to l and d in this proof and any other place (as p denote  both the price vector and an integer here).}
	By Definition~\ref{def:pes}, we need to prove that the $\anypricebi \geq \pqshare$ holds for every integers $\ell,d$ satisfying $\frac{\ell}{d}\leq b_i$.
	Let $A=(A_1,\ldots,A_q)$ be the partition of $\items$  to {$d$} bundles with respect to Definition~\ref{def:lds}.
	For every vector of prices $\vect{p}$ that sums up to 1, the expected price of a union of $\ell$ different bundles among $A$, picked uniformly over all { $ {d}\choose{\ell}$ unions},
	%random union of $p$ different bundles among $A$ 
	is exactly $\frac{\ell}{d}$. This holds as every item is picked with probability $\frac{\ell}{d}$ (as it belongs to one set, and $\ell$ sets are picked out of $d$).
%	\ufc{commented out the rest of the sentence} \tec{ok}%and by normalization the value of $\items$ is 1. 
	Recall that $\frac{\ell}{d}\leq b_i$.
	Thus, there are $\ell$ different bundles that their sum of prices is at most $b_i$ and thus can be purchased. By Definition~\ref{def:lds}, the value of their union is at least  the value of the $\pqshare$.
%	\tec{Do we want  to add also a proof using the second definition?}
\end{proof}

\propShareComp*
\begin{proof}
%		$$\proportionalbi \geq \anypricebi \geq \pessimisticbi \geq \frac{1}{2}\cdot (\anypricebi)$$
%\mbc{proof slightly revised and restructured.}
	
	The inequality $\proportionalbi \geq \anypricebi$ follows by setting the price vector to be $p_j = \frac{v_i(j)}{v_i(\items)}$. Then, the price of a bundle is proportional to the value of it, and every affordable bundle has value of at most $b_i \cdot v_i(\items)$.
	The middle inequality follows by Proposition~\ref{prop:APS-pes}. 
	Each of these two inequalities might hold as equality: when there are $2$ items each of value $1/2$, and $b_1=b_2=1/2$ then all three shares are $1/2$.
%	We now prove the third inequality: 
The rightmost inequality is proven in Lemma \ref{lem:pes-APS-2}, 
and it can hold as equality by Example \ref{example:APS-twice-Pessimistic}.

We next argue that each of the  inequalities might be strict {even for equal entitlements ($b_i=\frac{1}{k}$ for some integer $k$).  The leftmost inequality is strict in the case of a single item and $b_i=1/2$, where the proportional share is $b_i=1/2$, while the APS is $0$.  The rightmost inequality is strict in the example for which $\anypricebi = \pessimisticbi$ presented above ($2$ identical items, 2 agents with equal entitlements).  Finally, Lemma \ref{lem:APSlargerPESS} shows  that the middle inequality is sometimes strict, even for agents with equal entitlements (it is much simpler to prove strict inequality for unequal entitlements, e.g., Example \ref{example:APS-twice-Pessimistic}).}
%OLD:
%We next argue that each of the  inequalities might be strict.  The leftmost inequality is strict in the case of a single item and $b_i<1$, where the proportional share is $b_i$, while the APS is $0$.  The leftmost inequality is strict in the example for which $\anypricebi = \pessimisticbi$ presented above. Finally, we show in the next lemma that the middle inequality is sometimes strict.  
% Lemma \ref{lem:APSlargerPESS} shows  that the middle inequality is sometimes strict 
\end{proof}

\begin{lemma}\label{lem:APSlargerPESS}
	There exists an additive valuation (over $15$ items) such that for $b_i=\frac{1}{3}$ (corresponding to three agents with equal entitlements) it holds that  $\anypricevali{\frac{1}{3}} > \pessimisticb{\frac{1}{3}} = \MMSnum{3}$.
	%There exists an instance with equal entitlements (specifically, $b_i=\frac{1}{3}$) and additive valuations for which $\anypricevali{\frac{1}{3}} > \pessimisticb{\frac{1}{3}} = \MMSnum{3}$.
	% \mbc{OLD:} There exists an instance for which $\anypricebi > \pessimisticbi$. Moreover, there is such an instance with equal entitlements ($b_i=\frac{1}{3}$).  
\end{lemma}
\begin{proof}
{
Consider the additive valuation $v_i$ over a set $\items$ of $15$ items with the multi-set of items value being $$5,5,5,7,7,7,11,17,23,23,23,31,31,31,65.$$
To prove the claim we show that for this instance  
$\anypricevali{\frac{1}{3}}= 97>  \pessimisticb{\frac{1}{3}} = \MMSnum{3}$.

It will be convenient to index each item by a pair of different integers in $\{1,2,3,4,5,6\}$, that is, the set of items is $\items = \{\{k,j\} | 1\leq k<j \leq 6\}$. 
We can now write the items in the following six-by-six table (for convenience each value appears twice).}

%OLD: We show that the APS can be strictly larger than the pessimistic share even for the case where $b_i=\frac{1}{k}$.Consider the case of $b_i=\frac{1}{3}$ and a set of items $\items$ with values $$5,5,5,7,7,7,11,17,23,23,23,31,31,31,65.$$

%We refer to the items as $\items = \{\{k,j\} | 1\leq k<j \leq 6\}$ and their values are presented in the following table (for convenience each value appears twice).
%\{center}
\begin{table}
\centering
\begin{tabular}{|c|c|c|c|c|c|}

	\hline
	% after \\: \hline or \cline{col1-col2} \cline{col3-col4} ...
	* & 11 & 7 & 7 & 7 & 65 \\
	\hline
	11 & * & 23 & 23 & 23 & 17 \\
	\hline
	7 & 23 & * & 31 & 31 & 5 \\
	\hline
	7 & 23 & 31 & * & 31 & 5 \\
	\hline
	7 & 23 & 31 & 31 & * & 5 \\
	\hline
	65 & 17 & 5 & 5 & 5 & * \\
	\hline
\end{tabular}
\end{table}
%\{\center}

{By Definition~\ref{def:anyprice-bundles}, the APS is at least $97$: 
the sets in the support will be the six rows $T_k = \{ \{j,k\} \mid  j \neq k\}$ for $k\in [6]$, each has value of $97$. Set  $\lambda_{T_k}= \frac{1}{6}$ for each row $T_k$, and set $\lambda_{T}=0$ for every other bundle $T$. 
Then every bundle $T_k$ has value of $97$, and every item $\{j,k\}$ 
as total weight at most $1/3$ as it
appears in exactly two sets of positive weight (once in $T_j$ and once in $T_k$): $\sum_{T : \{j,k\} \in T} \lambda_T =\lambda_{T_k} +\lambda_{T_j}=  \frac{1}{6}+ \frac{1}{6}=\frac{1}{3}\leq b_i$ as needed.
The APS is not more than $97$ since it is bounded by the proportional share,  which is $97$.

In contrast, we next prove that the MMS is strictly less than $97$ (it cannot be larger than 97, as the MMS is at most the APS.) 
Assume towards a contradiction that the MMS is $97$. By definition of the MMS,  there must exist a partition of $\items$ to three bundles $P_1,P_2,P_3$, each with value of exactly $97$. To complete the proof we show that such a partition does not exist, deriving a contradiction. 

W.l.o.g., we assume that the item of value $65$  is in $P_1$.}
The only two ways to complete $P_1$ to have a value of {exactly} $97$ is by adding the items with values $5,5,5,17$ or adding the items with values $7,7,7,11$. This is since if we add an item with value $31$, then we get a value of $96$ and the is no item with value $1$.
If we add an item of value $23$ then we get value of $88$, and there is no subset of items worth $9$.
In order to add a value of $32$, one must use on of the $11,17$ valued items, since otherwise, the maximal value that can be achieved is $36$, and the second maximal value is $31$.
If the $11$ value is added, then the only way to achieve an additional value of $21$ is by adding the three $7$ valued items, and if the $17$ value is added, then the only way to achieve an additional value of $15$ is by adding the three $5$ valued items.

W.l.o.g., we assume that at least two items of value $31$ (among the remaining three) are in $P_2$.
It cannot be that the third $31$ value is also in $P_2$, since otherwise we have accumulated a value of $93$ and cannot be completed to {exactly} $97$.
It must be that at least one of the $23$ valued items is in $P_2$, since otherwise the value of $P_3$ is at least $3\cdot 23 +31=100 >97$.
It must be that at most one of the $23$ valued items is in $P_2$ since otherwise the value of $P_2$ is at least $2 \cdot 31 +2 \cdot 23 = 108>97$.
In order to complete $P_2$ to have a value of $97$, we need to add items with total value of $12$, and the only way to do so, is by adding on item of value $5$, and one item of value $7$.
But it must be that either all three items of value $5$ are in $P_1$, or all three items of value $7$  are in $P_1$, thus $P_2$ cannot be completed to have a value of $97$.
Thus there is no such partition, and the MMS is smaller than $97$ (and as it is an integer, it is at most $96$).
\end{proof}
\begin{lemma}\label{lem:pes-APS-2}
	For any additive valuation $v_i$ and any entitlement  $b_i$ it holds that \\
	$ 2\cdot \pessimisticbi \geq \anypricebi$.
\end{lemma}	
\begin{proof}
	Let $z$ be the value of the $\anypricebi$ share. Then there is a collection $S_1, \ldots, S_t$ of bundles and positive coefficients $\lambda_1, \ldots \lambda_t$ satisfying $\sum_{\{\ell \le t\}} \lambda_\ell = 1$, $v(S_\ell) \ge z$ for every bundle $S_\ell$, and $\sum_{\ell | j \in S_\ell} \lambda_\ell \le b_i$ for every item $j$. Let $k$ be the smallest integer satisfying $b_i \geq \frac{1}{k}$.
	%If $b_i=1/k$ the approximation is even better than half, see \ref{} \tec{do we want to keep it this comment and add the proof?}. Thus in the rest of the proof we assume that $\frac{1}{k-1} > b_i > \frac{1}{k}$. 
	We need to show that one can partition all items into $k$ disjoint bundles $T_1, \ldots, T_k$ such that $v(T_\ell) \ge \frac{z}{2}$ for every $T_\ell$. To distinguish between the bundles of type $S_\ell$ and those of type $T_\ell$, we refer to the latter as {\em bins}. We shall place items in bins, and {\em close} a bin once its value reaches (or exceeds) $\frac{z}{2}$.
		
	Call an item {\em large} if its value is at least $\frac{z}{2}$. We first present an argument that shows that we can assume that there are no large items.
		
	Let $h$ denote the number of large items. Each large item closes a bin by itself. If $h \ge k$ we are done, and hence we assume that $h \le k-1$. This leaves us $k - h$ bins to fill. As to the original bundles, discard (change its coefficient $\lambda_\ell$ to~0) every bundle $S_\ell$ that contains a large item. As the coefficients $\lambda_\ell$ of bundles that contain an item sum up to at most $b_i$, the sum of coefficients that remain is at least $1 - b_i h > 0$. (Note that the choice of $k$ satisfies $b_i < \frac{1}{k-1}$, and together with $h \le k-1$, this implies that $b_ih < 1$.)  W.l.o.g., assume that the sum of coefficients that remain is exactly $1 - b_i h$. If $h = k-1$ we are also done for the following reason. $k-1$ bins are closed by using the $h = k-1$ large items. As the sum of coefficients that remain is positive, at least one bundle (say $S_\ell$) has a nonzero coefficient. This bundle has no large items, implying that all its items remain. As the sum of their values is at least $z$,  these items suffice in order to close the last remaining bin. In other words, we may choose $T_k = S_\ell$. Given the above, we may assume that $h \le k - 2$.
		
	Scale every coefficient by $\frac{1}{1 - b_i h}$, getting new coefficients $\lambda'_\ell = \frac{\lambda_\ell}{1 - b_ih}$.  Now $\sum_\ell \lambda'_\ell = 1$, and for every item $j$, $\sum_{\{\ell | j \in T_\ell\}} \lambda'_\ell \le \frac{b_i}{1 - b_ih} < 1$. Denote $\frac{b_i}{1 - b_ih}$ by $b'_i$.
		Let $k'$ be the smallest integer satisfying $b' \ge \frac{1}{k'}$.
		We have that $k' = k-h$, because then $k'b'_i = \frac{(k-h)b_i}{1 - b_ih} = \frac{kb_i - hb_i}{1 - hb_i} \geq \frac{1 - hb_i}{1 - hb_i} = 1$. (A similar computation shows that $k' = k - h - 1$ does not suffice.)
		
	Renaming $k'$ by $k$ and $b'_i$ by $b_i$, we return to the starting point of the theorem, except that now every item has value below $\frac{z}{2}$.
		
	Having concluded that we can assume that no item has value above $\frac{z}{2}$, consider first the case that $1>b_i \ge \frac{1}{2}$, for which $k = 2$. Let $x$ be the highest value item. As at least one of the bundles does not contain $x$, the total value of items is at least $z + v_i(x)$ (the bundle not containing $x$ contributes at least $z$, and $x$ itself contributes $v_i(x)$). Bin $T_1$ will have value at most $\frac{z}{2} + v_i(x)$ (as its value grows by steps not larger than $v_i(x)$, and we stop adding to it the first time the value is at least $z/2$), and consequently value of at least $\frac{z}{2}$ remains for bin $T_2$, as desired.

	Consider now the case that $b_i < \frac{1}{2}$, for which $k \ge 3$. For $j \ge 1$, let $u_j$ denote the value of the $j$th most valuable item. We consider two cases.
		
		\begin{itemize}
			
			\item $u_{2k-2} < \frac{z}{4}$. For $1 \le j \le k$, put $u_j$ in bin $T_j$. Thereafter add items to the bins, closing each bin as soon as its value reaches $\frac{z}{2}$. Note that at most $k-3$  bins can close at a value above $\frac{3z}{4}$ (as only $2k-3$ items have value %OLD: above $\frac{w}{4}$),
			at least $\frac{z}{4}$, and every item has value smaller than $\frac{z}{2}$),
			and no bin can close at value above $z$ (because no single item has value % OLD: above $\frac{w}{2}$
			larger or equal to $\frac{z}{2}$). Consequently, even after all but one bin close, the total value left for the last bin is at least $(k-1)z - (k-3)z - 2\frac{3z}{4} \ge \frac{z}{2}$, as desired.
			
			\item $u_{2k-2} \ge \frac{z}{4}$.
			Then we can cover the first $k-1$ bins by partitioning the first $2k-2$ items to $k-1$ disjoint pairs,  and
				putting a pair in each of these bins.
			%Then putting two of the first $2k-2$ items is each of the first $k-1$ bins covers these bins.
			We claim that value above $\frac{z}{2}$ remains for bin $T_k$.
			% \mbc{Here is a simpler argument: "As each item has value at most $w/2$ the total value of these  $2k-2$ items is at most $w(k-1)$, so there is at least $w$ in $T_k$ as the total is at least $wk$."}
			This is because among the original bundles, at least one bundle $S_i$ contains at most one of the first $2k-2$ items. 
			(Every item is in at most a $b_i$-weighted-fraction of the bundles for $b_i < \frac{1}{k-1}$. The expected number of items among the top $2k-2$ items distributed according to $\lambda$ (i.e., bundle $S_\ell$ w.p. $\lambda_\ell$), is at most $(2k-2) \cdot b_i < 2$. Thus, there must be a bundle $S_\ell$ that contains at most one element among the top $2k-2$.)
			% \tec{Old: (If every bundle contains at least two of the first $2k-2$ items, then there must be an item that is contained in a $\frac{1}{k-1}$-weighted-fraction of the bundles. However, each item is only in a $b_i$-weighted-fraction of the bundles, and $b_i < \frac{1}{k-1}$.} 
			 %\mbc{not clear} \ufc{Is it clearer now?}\mbc{can we do this algebraically? might be easier to follow}). 
			 This bundle $S_\ell$ has value of at least $z - u_1 \ge \frac{z}{2}$ left.
		\end{itemize}
%\tec{keep the equalities? or say it might be strict inequality?}
\end{proof}

\obsApsTwo*
\begin{proof}
	{
	For an entitlement of the form $b_i=\frac{1}{n}$, the pessimistic share is the same as the MMS of $n$ players (i.e., $\pessimisticb{\frac{1}{n}}=\MMSn$).
	Then by Proposition~\ref{prop:shares-comp} it holds that $\anypricevali{\frac{1}{2}} \geq \pessimisticb{\frac{1}{2}}=\MMSTwo$, and hence the APS is at least as large as the MMS. 
	On the other hand, the MMS of two player is the maximal value of any bundle that has a value of at most $\frac{v_i(\items) }{2}$, and the APS cannot be larger as then it will be larger than the proportional share, contradicting Proposition~\ref{prop:shares-comp}.}
%	\mbc{OLD:}
%	Let $S$ be the maximal value bundle that has value at most $x=\frac{v_i(\items) }{2}$. The MMS is at least this value since the minimal value in the partition $(S,\bar{S})$ is $x$ (since the $v_i(S)\leq v_i(\bar{S})$).	By Proposition~\ref{prop:shares-comp} we know that the APS is at most the proportional share $\frac{v_i(\items)}{2}$.	Since the APS must be in the support of $v_i$, then it must be at most $x$. Combining with that $\frac{1}{2}$-APS $\geq$ MMS (by Proposition~\ref{prop:shares-comp}), we get that both shares must be equal to $x$.
\end{proof}

{\begin{remark}
	\label{ex:aps-half-mms}
	Observation~\ref{obs:APS-two-equal} uses in an essential way the fact that valuations are additive. The following example shows that for $b_i=\frac{1}{2}$,  the APS and the MMS might be different for submodular valuations.
	Consider a setting with six items  $\items=\{1,2,3,4,5,6\}$,  and let $A=\{\{1,2,3\},\{1,5,6\},\{2,4,6\},\{3,4,5\}\}$. For every agent $i$, let 
	$v_{i}(S)= 6$ if there exists $T \in A $ such that $T \subseteq S$ (i.e., $S$ contains a set in $A$), or {if $|S| \geq 4$}.
	Let $v_i(S)=2$ for any set $S$ of size $1$. Let $v_i(S)=4$ for any set $S$ of size $2$, and let $v_i(S)=5$ for any set $S$ of size $3$, such that $S \notin A$.
%	for all non empty set that does not contain any set of $A$, then $v_i(S)=2$. 
	This is a submodular valuation since the marginal of every item is in $0,1,2$, and can only decrease as more items are added.
The MMS is at most $5$ since there is no partition $(S_1,S_2)$ in which both $S_1,S_2$ have value~$6$ (either contain a set in $A$, or are of size at least $4$).
The APS is at least $6$ since if we set $\lambda_S=\frac{1}{4}$ for every $S\in A$ we get that the APS is at least $6$, by Definition~\ref{def:anyprice-bundles}. 

\paragraph{Subadditive:} If we change the value of the non-empty sets that have value less than~6 to be~3, then the function is no longer submodular, but is still subadditive. The APS of the modified function is still $6$, whereas the MMS drops to $3$. %so the gap for subadditive is at least factor of $2$.
%\tec{Do we want to mention that if we change the $3$ to $4$ then it is subadditive and we get a gap of 2?} \mbc{yes, add a sentence here about this.}

%\mbc{add the computation of the MMS and APS for this example.}
\end{remark}}

\obsTwo*
\begin{proof}
	Let $F$ be the distribution over sets (non-negative weights $\{\lambda_T\}_{T\subseteq \items}$ that total to $1$),
	that by  Definition~\ref{def:anyprice-bundles} define $\anyprice{v_1}{b_1}{\items}$, the APS of agent $1$ that has valuation $v_1$ and entitlement $b_1$.
	The expected value for the other agent, agent $2$, for the complement of a random set sampled from $F$ is $\sum_{T\subseteq \items} \lambda_T \cdot v_2(\items\setminus T) = \sum_{T\subseteq \items} \lambda_T \cdot (v_2(\items) - v_2(T))= v_2(\items) - \sum_{T\subseteq \items} \lambda_T \cdot v_2(T)\geq  v_2(\items) - b_1 \cdot v_2(\items) = b_2 \cdot v_2(\items)$, using the fact that $v_2$ is additive and as $\sum_{T\subseteq \items} \lambda_T \cdot v_2(T)\leq b_1 \cdot v_2(\items)$, since each item has total weight at most $b_1$.
	%\tee{for the random complement of the set sampled from $F$ is at least $b_2 \cdot v_2(\items)$ since every item is sampled with probability at least $1-b_1=b_2$} \mbc{explain why. this should use additivity explicitly} \tec{fixed?}, 
	As $b_2 \cdot v_2(\items)$ is the  proportional share of agent $2$, it is at least her APS according to Proposition~\ref{prop:shares-comp}.
	Thus, for at least one set $S$ in the support of $F$, agent $2$ values the set $\items\setminus S$ at least at her APS. 
	As every set in the support of $F$ gives agent $1$ her APS, 
	the allocation $A=(S,\items\setminus S)$ gives each of the two agents her AnyPrice share.	
%	\mbc{OLD:}
%	Let $S_1,\ldots,S_t$ and $\lambda_1,\ldots,\lambda_t$ be the sets and weights according to Definition~\ref{def:anyprice-bundles} with respect to valuation $v_1$ and entitlement $b_1$.	The expected value for agent $2$ for the random set $S = \items \setminus S_i\ell$ w.p. $\lambda_\ell$ is at least $b_2 \cdot v_2(\items)$ which is his proportional share, which is at least his APS according to Proposition~\ref{prop:shares-comp}. 	Thus, one of the sets $\items \setminus S_\ell$ gives agent $2$ value of at least his APS. Let $S_k$ be that set. 	The allocation $(S_k,\items\setminus S_k)$ gives both agents their APS.	
\end{proof}

\proppseudo*
\begin{proof}
Under the conditions of the proposition, the AnyPrice share $z=\anypricei$ is an integer. It suffices to design a polynomial time testing algorithm that given a candidate value for $t$, tests in time polynomial in $t$ whether $z < t$. Using such a testing algorithm, the value of $z$ can be found by binary search over values of $t$ (without need to ever test a value of $t$ larger than $2z$), invoking the testing algorithm only a polynomial number of times.

For testing a candidate value of $t$, we use the following linear program. Its variables are $p_j$, the prices for the items.

{\bf Minimize} $\sum_{j \in \items} p_j$ subject to the following constraints:

\begin{itemize}
\item $p_j \ge 0$ for every item $j \in \items$.
\item $\sum_{j \in S} p_j \ge b_i$ for every set $S \subset \items$ with $v_i(S) \ge t$.
\end{itemize}

If $z < t$ then the optimal value of the LP is strictly smaller than~1. This is because by definition of the AnyPrice share, there is a price vector $p$ with respect to which every bundle of value $t$ has price at least $b_i + \epsilon$, for some $\epsilon > 0$. Scaling all prices by $\frac{b_i}{b_i + \epsilon}$ gives a feasible solution to the LP, and its value is strictly smaller than~1.

The LP has exponentially many constraints, one for each set $S$ with value $v_i(S) \ge t$. Nevertheless, it can be solved in pseudo-polynomial time. This is because given a candidate assignment to the $p_j$ variables, one can test if there is a violated constraint by solving a knapsack problem over the items $\items$, with $b_j$ as the knapsack size, the $p_j$ serve as item sizes, and the values $v_i(j)$ are the item values. Knapsack can be solved in time pseudo-polynomial in the value of the solution, hence in time pseudo-polynomial in $t$.
\end{proof}

\section{Missing Proofs from Section~\ref{sec:allocation-game}}
\label{app:allocation-game}
%\ufc{The following paragraph was part of the proof that follows. However, the proof was simplified and the paragraph is no longer needed. Should we keep this ``obvious" statement somewhere, or should we just remove it?}
%We first observe that the TPS is monotone in $v$.
%	Formally it means that for $v_i$ and $v_i'$, if for all $j$ it holds that $v_i(j)\leq v_i'(j)$, then  
%$	\truncated{v_i}{b_i}{\items} \leq \truncated{v_i'}{b_i}{\items}$. 

Before we prove Theorem~\ref{thm:app-APS}, we first observe the following property regarding the TPS.
\begin{lemma}\label{lem:ktps}
	For every subset of items  $K \subset \items$, and for every $b_i\leq \frac{1}{|K|+1}$  and additive valuation $v_i$  over the set of items $\items$, it holds that $$
	\truncatedi \leq \truncated{\frac{b_i}{1-|K|\cdot b_i}}{v_i}{\items \setminus K}.  $$
\end{lemma}
\begin{proof}
	It suffices to prove this lemma for $K$ which is a singleton. The proof of the lemma then follows by observing that if $b_i \leq \frac{1}{|K|+1}$ then $\frac{b_i}{1-b_i} \leq \frac{1}{|K|}$ and that $\frac{\frac{b_i}{1-b_i}}{1-(|K|-1)\frac{b_i}{1-b_i}} = \frac{b_i}{1-|K|\cdot b_i}$. The lemma follows by partitioning $K$ into $|K|$ singletons, and iterating over the singletons.

	%For $K=\{j\}$, and $b_i\leq \frac{1}{2}$, by monotonicity that if we consider $v_i'$ to be the same valuation but the value of item $j$ is infinity, then $\truncated{v_i'}{b_i}{\items} \geq \truncated{v_i'}{b_i}{\items}$.
	%We next show that  $\truncated{v_i'}{b_i}{\items} \leq \truncated{v_i'}{\frac{b_i}{1-b_i}}{\items\setminus\{j\}}=\truncated{v_i}{\frac{b_i}{1-b_i}}{\items\setminus\{j\}}$, where the  equality is since the $v_i$ and $v_i'$ are have the same values for the items in $\items \setminus\{j\}$.
	%Let $z = \truncated{v_i'}{b_i}{\items}$ (note that this value is finite since $b_i \leq 1/2$ and there is only one item that is worth infinity).
	Fix $K=\{j\}$ and $b_i\leq \frac{1}{2}$.
	By definition $$z=b_i \cdot \sum_{\ell\in \items} \min(v_i(\ell),z) \le b_i \cdot z + b_i \cdot \sum_{\ell\in \items \setminus\{j\}} \min(v_i(\ell),z). $$
	By rearranging, we get that:
	$$z \le \frac{b_i}{1-b_i} \cdot \sum_{\ell\in \items \setminus\{j\}} \min(v_i(\ell),z) .$$ 
	Thus, by definition $z\leq \truncated{\frac{b_i}{1-b_i}}{v_i}{\items\setminus\{j\}}, $
	which concludes the proof. 
\end{proof}

A similar result holds for the APS (with a somewhat relaxed constraint on $b_i$).
\begin{lemma}\label{lem:kaps}
	For every subset of items  $K \subset \items$, and for every $b_i < \frac{1}{|K|}$  and additive valuations $v_i$  over the set of items $\items$, it holds that $$
	\anypricei \leq \anyprice{\frac{b_i}{1-|K|\cdot b_i}}{v_i}{\items \setminus K}.  $$
\end{lemma}

\begin{proof}
	Let $\vect{p}$ be the prices according to Definition~\ref{def:anyprice-bundles} with respect to items $\items \setminus K$ and entitlement $\frac{b_i}{1-|K|\cdot b_i}$.
	From $\vect{p}$ we derive prices $\vect{p'}$ for $\items$. For every item $j\in K$ the price is $p'_j = b_i+\epsilon$, and for every item $j\in \items \setminus K$ the price is $p'_j = p_j \cdot (1-|K|\cdot (b_i+\epsilon'))$. The ratio between $\epsilon$ and $\epsilon'$ is chosen so that the sum of prices is~1, and $\epsilon'$ is chosen to be small enough so that agent $i$ with budget $b_i$ and prices $\vect{p'}$ cannot afford any bundle that could not be afforded with budget $\frac{b_i}{1-|K|\cdot b_i}$ and prices $\vect{p}$. 
\end{proof}

%\begin{observation}
%	Adding or removing $0$ valued items does not change the TPS.
%\end{observation}

%\thmappAPS*

The following lemma presents a strategy for agent $i$ with a budget of $b_i$ that she can use in the bidding game %(Algorithm~\ref{alg:bidding}) 
to guarantee herself a value of at least $\frac{1}{2-b_i}$ fraction of her TPS. At each round of bidding she should do as follows. If there is one item that suffices by itself to obtain the goal, bid your whole budget. Else, if two items suffice, bid half your budget (and take both items if you win). If at least three items are needed, make a bid proportional to the value of the most valuable remaining item, and pick that item if you win. 

Before moving to present the formal strategy, we repeat the notations we use in the proofs (defined in the bidding game). We denote by $s_i^{(t)} = v_i(\items^{(t)})$ the sum of values of the remaining items at the beginning of round $t$, and denote by $x_i^{(t)},y_i^{(t)}$  the value of the most and second most valued items among $\items^{(t)}$. We use the notation of $s_i$ to denote $s_i^{(1)} = v_i(\items)$.
Finally, we denote the total remaining budget at the beginning of round $t$ by $B^{(t)}=\sum_k b^{(t)}_k$, where $b^{(t)}_k$ is the budget available to agent $k$ at the beginning of round $t$.

We next show that agent $i$ has a strategy that guarantees her $\frac{1}{2-b_i}$-fraction of her TPS.
\begin{lemma}
The following strategy guarantees
agent $i$ with additive valuation function $v_i$ and entitlement  $b_i$ a value of at least $\frac{1}{2-b_i}\cdot 	\truncatedi $ in the bidding game, %(Algorithm~\ref{alg:bidding})
regardless of the strategies of the other agents.
	
%	The following strategy guarantees $\frac{1}{2-b_i}$-fraction of the TPS.
	\begin{enumerate}
		\item If $x_i^{(t)} \geq  \frac{b_i^{(t)}}{B^{(t)}} \cdot s_i^{(t)}$ then bid $b_i^{(t)}$. If the bid wins select item~1 and drop out.
		\item If $x_i^{(t)} < \frac{b_i^{(t)}}{2B^{(t)}-b_i^{(t)}} \cdot s_i^{(t)}$ and $x_i^{(t)}  + y_i^{(t)} \ge \frac{b_i^{(t)}}{B^{(t)}} \cdot s_i^{(t)}$ then bid $\frac{b_i^{(t)}}{2}$. If the bid wins select items~1 and~2 and drop out. %(and drop out of the game, as the budget is exhausted).
		\item Else bid $\frac{x_i^{(t)}}{s_i^{(t)}} \cdot B^{(t)}$. If the bid wins select the most preferred item. (Observe that the bid is indeed at most $b_i^{(t)}$.)
	\end{enumerate}
\label{lem:1bi}
\end{lemma}%\paragraph{The strategy for agent $i$:} %If either $b_i^{(t)}=0$ or $x_i^{(t)}=0$, the agent drops out of the allocation game.

\begin{proof}
We divide the rounds into two stages. 
The first stage (that may be empty) holds while $\frac{x_i^{(t)}}{s_i^{(t)}} \ge \frac{b_i^{(t)}}{B^{(t)}}$.
In this stage the agent bids her whole budget. If she wins, she gets her TPS, and we are done. If she does not win, then Lemma~\ref{lem:ktps} implies that $TPS_i$ with respect to the new instance does not decrease. Moreover, the relative budget of agent $i$ (after some other agent paid for an item and the budgets are scaled back to sum to~1) increases, and so does the term $\frac{1}{2 - b_i}$. Consequently, a guarantee of $\frac{1}{2 - b_i} \cdot TPS_i$ with respect to the instance at the beginning of the second stage implies at least as good a guarantee with respect to the original instance. Hence we may assume without loss of generality that we start directly in the second stage.

In the beginning of the second stage, we have that $\frac{x_i^{(t)}}{s_i^{(t)}} \leq \frac{b_i^{(t)}}{B^{(t)}}$. We prove by induction on the number of items $m$ that as long as agent $i$ in the game, it holds that $\frac{x_i^{(t)}}{s_i^{(t)}} \leq \frac{b_i^{(t)}}{B^{(t)}}$, and that the agent receives  $\frac{b_i^{(t)}\cdot s_i^{(t)}}{2B^{(t)}-b_i^{(t)}}$ which in the beginning is at least $\frac{1}{2-b_i^{(t)}}$-fraction of the TPS.

%If the total budget $B^{(t)}$ and the total value $s_i^{(t)}$ are normalized to be~1, the strategy above simplifies to bidding $x_i^{(t)}$, unless $x_i^{(t)} < \frac{b_i^{(t)}}{2-b_i^{(t)}}$ and $x_i^{(t)} + y_i^{(t)} \ge b_i^{(t)}$, in which case one bids $\frac{b_i^{(t)}}{2}$ and selects two items if the bid wins.

%\tec{Remove this or rephrase? (We remark that there are other variations of the above strategy that will work just as well in the proof of Theorem~\ref{thm:safeStrategy1}. For example, if $x_i^{(t)} \ge \frac{b_i^{(t)}}{2-b_i^{(t)}}$ we can bid $b_i^{(t)}$ rather than $x_i^{(t)}$, and/or we can change the condition $x_i^{(t)} + y_i^{(t)} \ge b_i^{(t)}$ to $x_i^{(t)} + x_i^{(t)} \ge \frac{b}{2-b}$.)}

%We prove by induction over the number  of items $m$ the following inductive hypothesis: if initially $\frac{x_i^{(t)}}{s_i^{(t)}} \leq \frac{b_i^{(t)}}{B^{(t)}}$, then using the above strategy guarantees that the agent achieves a value of at least $\frac{b_i^{(t)}}{2B^{(t)}-b_i^{(t)}}\cdot s_i^{(t)}$.

The base case is $m=1$ (and $x_i^{(t)} > 0$). In this case $x_i^{(t)}=s_i^{(t)}$, and the assumption that $\frac{b_i^{(t)}}{B^{(t)}} \geq \frac{x_i^{(t)}}{s_i^{(t)}}=1$ implies that $b_i^{(t)}=B^{(t)}$. By rule~3 the agent will bid $b_i^{(t)}$, win the item, and obtain value $x_i^{(t)} = \frac{b_i^{(t)}}{2B^{(t)} - b_i^{(t)}}s_i^{(t)}$, as desired.

We now assume that the inductive hypothesis holds whenever the number of items is at most $m$, and we prove that it holds also when the number of items is $m+1$. The proof breaks into cases depending on which of the two rules (rules 2 and 3) is used in the current round, and on whether the agent won the bid in the current round.

\begin{itemize}	

\item Rule~3 was invoked, but the agent did not win the bid. As the agent's bid under Rule~3 is $\frac{x_i^{(t)}}{s_i^{(t)}} \cdot B^{(t)}$, whoever wins the round pays at least $\frac{x_i^{(t)}}{s_i^{(t)}} \cdot B^{(t)}$ per item. Let $k$ denote the number of items taken in the first round, and note that $k \le \frac{s_i^{(t)}(B^{(t)} - b_i^{(t)})}{x_i^{(t)} B^{(t)}}$. As each item has value at most $x_i^{(t)}$ for our agent, the value left by the next round is at least $s_i^{(t+1)} \geq s_i^{(t)} - x_i^{(t)}\cdot \frac{s_i^{(t)}(B^{(t)} - b_i^{(t)})}{x_i^{(t)} B^{(t)}} = \frac{b_i^{(t)}}{B^{(t)}} \cdot s_i^{(t)}$, so there is enough value left for the agent to potentially reach a value of $\frac{b_i^{(t)}}{2B^{(t)} - b_i^{(t)}}s_i^{(t)}$.

We now verify that the invariant $\frac{x_i^{(t)}}{s_i^{(t)}} \leq \frac{b_i^{(t)}}{B^{(t)}}$ holds after the first round.  We have $s_i^{(t+1)} \ge s_i^{(t)} - kx_i^{(t)}$, $x_i^{(t+1)} \le x_i^{(t)}$, $b_i^{(t+1)} = b_i^{(t)}$ and $B^{(t+1)} \le B^{(t)} - \frac{kB^{(t)}x_i^{(t)}}{s_i^{(t)}} = \frac{B^{(t)}(s_i^{(t)} - kx_i^{(t)})}{s_i^{(t)}}$.

$$\frac{x_i^{(t+1)}}{s_i^{(t+1)}} \le \frac{x_i^{(t)}}{s_i^{(t)} - kx_i^{(t)}} \le \frac{b_i^{(t)}s_i^{(t)}}{B^{(t)}(s_i^{(t)} - kx_i^{(t)})} \le \frac{b_i^{(t+1)}}{B^{(t+1)}}$$
where the middle inequality holds because
of the invariant $\frac{x_i^{(t)}}{s_i^{(t)}} \leq \frac{b_i^{(t)}}{B^{(t)}}$ (applied to the numerator $x_i^{(t)}$).

As the invariant holds in the second round, we can apply the inductive hypothesis and obtain that the value achieved by the agent is at least

\begin{eqnarray*}
\frac{b_i^{(t+1)}}{2B^{(t+1)}-b_i^{(t+1)}}\cdot s_i^{(t+1)}  &\ge & \frac{b_i^{(t)}}{2\frac{B^{(t)}(s_i^{(t)} - kx_i^{(t)})}{s_i^{(t)}} - b_i^{(t)}} (s_i^{(t)} - kx_i^{(t)}) \\&=& \frac{b_i^{(t)}}{2B^{(t)} - \frac{s_i^{(t)}}{s_i^{(t)} - kx_i^{(t)}}b_i^{(t)}} \cdot s_i^{(t)} \\& \ge& \frac{b_i^{(t)}}{2B^{(t)}-b_i^{(t)}}\cdot s_i^{(t)}.
\end{eqnarray*}
\item Rule~3 was invoked, and the agent won the bid. We have $s_i^{(t+1)} = s_i^{(t)} - x_i^{(t)}$, $B^{(t+1)} = B^{(t)} - \frac{x_i^{(t)}}{s_i^{(t)}} \cdot B^{(t)}$, and $b_i^{(t+1)} = b_i^{(t)} - \frac{x_i^{(t)}}{s_i^{(t)}} \cdot B^{(t)}$.  As for $x_i^{(t+1)}$, we shall have a case analysis. Suppose first that $x_i^{(t)} \le \frac{b_i^{(t)}}{2B^{(t)}} \cdot s_i^{(t)}$. For this case we use $x_i^{(t+1)} \le x_i^{(t)} \le \frac{b_i^{(t)}}{2B^{(t)}} \cdot s_i^{(t)}$. The invariant holds because:

$$\frac{x_i^{(t+1)}}{s_i^{(t+1)}} \le \frac{\frac{b_i^{(t)}}{2B^{(t)}} \cdot s_i^{(t)}}{s_i^{(t)} - x_i^{(t)}} = \frac{\frac{b_i^{(t)}}{2}}{B^{(t)} - \frac{x_i^{(t)}}{s_i^{(t)}} \cdot B^{(t)}} \le \frac{b_i^{(t)} - \frac{x_i^{(t)}}{s_i^{(t)}} \cdot B^{(t)}}{B^{(t)} - \frac{x_i^{(t)}}{s_i^{(t)}} \cdot B^{(t)}} = \frac{b_i^{(t+1)}}{B^{(t+1)}},$$
where the second inequality holds because $x_i^{(t)} \le \frac{b_i^{(t)}}{2B^{(t)}} \cdot s_i^{(t)}$ and thus $\frac{b_i^{(t)}}{2}-\frac{x_i^{(t)}}{s_i^{(t)}}\cdot B^{(t)}\geq 0$.

Suppose now that $x_i^{(t)} \ge \frac{b_i^{(t)}}{2B^{(t)}} \cdot s_i^{(t)}$. We may assume that $x_i^{(t)} < \frac{b_i^{(t)}}{2B^{(t)}-b_i^{(t)}} \cdot s_i^{(t)}$ as otherwise the agent who selects $x_i^{(t)}$ obtains a value of at least $\frac{b_i^{(t)}}{2B^{(t)}-b_i^{(t)}} \cdot s_i^{(t)}$ and we are done. Given that that $x_i^{(t)} < \frac{b_i^{(t)}}{2B^{(t)}-b_i^{(t)}} \cdot s_i^{(t)}$, we use the fact that Rule~2 was not invoked to infer that $y_i^{(t)} \le \frac{b_i^{(t)}}{B^{(t)}} \cdot s_i^{(t)} - x_i^{(t)}$.
Consequently $x_i^{(t+1)} \le y_i^{(t)} \le \frac{b_i^{(t)}}{B^{(t)}} \cdot s_i^{(t)} - x_i^{(t)}$. The invariant holds because:

$$\frac{x_i^{(t+1)}}{s_i^{(t+1)}} \le \frac{\frac{b_i^{(t)}}{B^{(t)}} \cdot s_i^{(t)} -x_i^{(t)}}{s_i^{(t)} - x_i^{(t)}} = \frac{b_i^{(t)} - \frac{x_i^{(t)}}{s_i^{(t)}} \cdot B^{(t)}}{B^{(t)} - \frac{x_i^{(t)}}{s_i^{(t)}} \cdot B^{(t)}} = \frac{b_i^{(t+1)}}{B^{(t+1)}}$$

As the inductive hypothesis holds regardless of the value of $x_i^{(t)}$, we can apply the inductive hypothesis and obtain that the value achieved by the agent is at least (the $x_i^{(t)}$ term comes from the first round):

$$x_i^{(t)} + \frac{b_i^{(t+1)}}{2B^{(t+1)}-b_i^{(t+1)}}\cdot s_i^{(t+1)} \ge x_i^{(t)} + \frac{b_i^{(t)} - \frac{x_i^{(t)}}{s_i^{(t)}} \cdot B^{(t)}}{2B^{(t)} - b_i^{(t)} - \frac{x_i^{(t)}}{s_i^{(t)}} \cdot B^{(t)}} \cdot (s_i^{(t)} - x_i^{(t)}) \ge \frac{b_i^{(t)}}{2B^{(t)}-b_i^{(t)}}\cdot s_i^{(t)}.$$

Verifying the last inequality involves tedious but straightforward algebraic manipulations, multiplying both sides by the denominators, and rearranging. After canceling matching terms the inequality becomes $B^2 - 2Bb + b^2 \ge 0$, which of course holds.

\item Rule~2 was invoked, and the agent won the bid. In this case the agent selects items~1 and~2 for a value of $\frac{b_i^{(t)}}{B^{(t)}}\cdot s_i^{(t)}$ and we are done.

\item Rule~2 was invoked, but the agent did not win the bid. As the agent's bid under Rule~2 is $\frac{b_i^{(t)}}{2}$, whoever wins the round pays at least $\frac{b_i^{(t)}}{2}$ per item. Let $k$ denote the number of items taken in the first round, and note that $k \le \frac{2(B^{(t)} - b_i^{(t)})}{b_i^{(t)}}$. As each item has value at most $x_i^{(t)}$ for our agent, and Rule~2 requires that $x_i^{(t)} < \frac{b_i^{(t)}}{2B^{(t)}-b_i^{(t)}} \cdot s_i^{(t)}$, the value left by the next round is at least $s_i^{(t+1)} \geq s_i^{(t)} - \frac{2(B^{(t)} - b_i^{(t)})}{b_i^{(t)}} \cdot \frac{b_i^{(t)}}{2B^{(t)}-b_i^{(t)}} \cdot s_i^{(t)} = \frac{b_i^{(t)}}{2B^{(t)}-b_i^{(t)}} \cdot s_i^{(t)}$, so there is enough value left for the agent to potentially reach a value of $\frac{b_i^{(t)}}{2B^{(t)} - b_i^{(t)}}s_i^{(t)}$.

We have $s_i^{(t+1)} \ge s_i^{(t)} - kx_i^{(t)} \ge s_i^{(t)} - \frac{ks_i^{(t)}b_i^{(t)}}{2B^{(t)}-b_i^{(t)}} = \frac{2B^{(t)}s_i^{(t)} - b_i^{(t)}s_i^{(t)} - ks_i^{(t)}b_i^{(t)}}{2B^{(t)}-b_i^{(t)}}$, $x_i^{(t+1)} \le x_i^{(t)} \le \frac{s_i^{(t)}b_i^{(t)}}{2B^{(t)}-b_i^{(t)}}$, $b_i^{(t+1)} = b_i^{(t)}$ and $B^{(t+1)} = B^{(t)} - \frac{kb_i^{(t)}}{2}$. The invariant $\frac{x_i^{(t)}}{s_i^{(t)}} \leq \frac{b_i^{(t)}}{B^{(t)}}$ holds in the second round because:

\begin{eqnarray*}
\frac{x_i^{(t+1)}}{s_i^{(t+1)}} &\le & \frac{s_i^{(t)}b_i^{(t)}}{2B^{(t)}-b_i^{(t)}} \cdot \frac {2B^{(t)}-b_i^{(t)}}{2B^{(t)}s_i^{(t)} - b_i^{(t)}s_i^{(t)} - ks_i^{(t)}b_i^{(t)}} \\ & =  & \frac{b_i^{(t)}}{2B^{(t)} - (k+1)b_i^{(t)}} \le \frac{b_i^{(t)}}{B^{(t)} - \frac{kb_i^{(t)}}{2}} \le \frac{b_i^{(t+1)}}{B^{(t+1)}},
\end{eqnarray*}
where the middle inequality holds because $k \le \frac{2(B^{(t)} - b_i^{(t)})}{b_i^{(t)}}$.

We can now apply the inductive hypothesis and obtain that the value achieved by the agent is at least

$$\frac{b_i^{(t+1)}}{2B^{(t+1)}-b_i^{(t+1)}}\cdot s_i^{(t+1)} \ge \frac{b_i^{(t)}}{2B^{(t)} - kb_i^{(t)} - b_i^{(t)}} \cdot \frac{2B^{(t)}s_i^{(t)} - b_i^{(t)}s_i^{(t)} - ks_i^{(t)}b_i^{(t)}}{2B^{(t)}-b_i^{(t)}} = \frac{b_i^{(t)}}{2B^{(t)}-b_i^{(t)}}\cdot s_i^{(t)},$$
as desired.
\end{itemize}
\end{proof}

%\ufc{Switch of context. Missing connecting text.}

{The next lemma which is used in Step~3 of the strategy described in Lemma~\ref{lem:35} shows that in %Algorithm~\ref{alg:bidding}
the bidding game, agent $i$ with a budget $b_i$ can guarantee a value of at least $\frac{3}{2}\cdot \anyprice{b_i/2}{v_i}{\items}$.}
%\begin{lemma}
\lemThreeFour*
\begin{proof}
Let $\bval=\anyprice{b_i/2}{v_i}{\items}$.
	By Definition~\ref{def:anyprice-bundles} there is a finite list of bundles $\{B_j\}$ and associated nonnegative weights $\lambda_j$ such that:
	
	\begin{itemize}
		\item $\sum_j \lambda_j = 1$.
		\item For every item $e$ we have $\sum_{j | e\in B_j} \lambda_j \leq \frac{b_i}{2}$.
		\item $v_i(B_j) \ge \bval$ for all $j$.
	\end{itemize}
	
	We may assume that $v_i(e) \le \bval$ for all items $e$, as if there is an item of value above $\bval$ we can reduce its value to $\bval$ without affecting the condition above, and any strategy that guarantees $\frac{3\bval}{2}$ with respect to the reduced valuation also guarantees it with respect to the original valuation.
	
%	For the situation in the beginning of round $t$ we use the following notation:
%	\begin{itemize}
%		\item $v_i$ denotes the value of the most valuable item. W.l.o.g., this will also be the item consumed in round $i$.
%		\item $b_i$ denotes the agent's budget in the beginning of round $i$.
%		\item $u_i$ denotes the total value accumulated by the agent before round $i$.
%	\end{itemize}

We denote by $u_i^{(t)}$  the total value accumulated by the agent up to the beginning of round $t$.

We present a strategy for the agent. Without loss of generality, we assume that whoever wins the item in a given round pays the amount that was bid by our agent in that round (this is the worst case for our agent). The strategy is composed  of two stages. For the {\bf first stage} the strategy has two rules.
	\begin{enumerate}
		%\item If $u_i \ge \frac{3b}{4}$, bid arbitrarily (the goal of reach value of $\frac{3b}{4}$ has been attained).
		\item If $x_i^{(t)} + u_i^{(t)} < \frac{3\bval}{2}$, bid $\frac{x_i^{(t)}}{s_i}$ (and choose one item).
		\item If $x_i^{(t)} + u_i^{(t)} \ge \frac{3\bval}{2}$, bid $b_i^{(t)}$ (and choose one item). 
%\ufc{Moved the following from outside to inside the enumerate environment.}
Here, if $\frac{x_i^{(t)}}{s_i} > b_i^{(t)}$, we say that the item is {\em underpriced}, and if $\frac{x_i^{(t)}}{s_i} < b_i^{(t)}$ we say that the item is {\em overpriced}.
	\end{enumerate}
	
	The first stage ends if one of the following two stopping conditions holds:
	\begin{enumerate}
		\item The agent accumulates a value of $\frac{3\bval}{2}$ (and then we are done).
		\item The first round $r$ for which both the following condition holds: in round $r-1$ rule~2 was applied, whereas in round $r$ rule~1 applies. In this case in round $r$, the agent plays according to the \textbf{second stage}.
	\end{enumerate}

	If the agent wins an overpriced item her accumulated value reaches $\frac{3\bval}{2}$, and we are done. Hence we may assume that the agent never wins an overpriced item. Making this assumption, it follows that if the adversary never wins underpriced items, then the total value won by the agent is at least $2\cdot \bval $. This is since  the total value of all items is at least $s_i \geq \frac{2\cdot \bval}{b_i}$, and the other agents hold a $1-b_i$ fraction of the budget. If the other agents do not win an underpriced item, the value that they win (according to $v_i$) is at most $s_i \cdot (1-b_i)$, and the value remaining for agent $i$ is at least $b_i \cdot s_i \geq 2z$. Hence it remains to analyze the situations in which the adversary does win underpriced items.
	
	The adversary can win an underpriced item only if in the given round we have that $\frac{x_i^{(t)}}{s_i} > b_i^{(t)}$. This, together with the fact that $x_i^{(t)} \le \bval$ for every $t$, implies that the agent wins at least two items before the adversary wins an underpriced item. Denoting by $w_1$ the value of the first item won by the agent and by $w_2 \le w_1$ the value of the second item won by the agent, we may assume that $w_1 + w_2 < \frac{3\bval}{2}$, as otherwise we are done.
	Hence $w_2 < \frac{3\bval}{4}$. Moreover, we may assume that $w_1 + 2w_2 > 2\cdot \bval $, as otherwise we cannot reach a round in which $u_i^{(t)} < \frac{3\bval}{2}$ and $\frac{x_i^{(t)}}{s_i} > b_i^{(t)}$. In particular, this implies that $w_2 > \frac{\bval}{2}$, and that $w_1 + w_2 > \frac{4\bval}{3}$.
	
	After the agent wins two items, we reach for the first time a round in which $b_i^{(t)} < \frac{x_i^{(t)}}{s_i}$ might hold. Thereafter, we might have a sequence of rounds in which the adversary wins underpriced items. We refer to this sequence of rounds as the {\em bad} sequence. The value of an underpriced item is at most $w_2 \le \frac{3\bval}{4}$, and the amount paid by the adversary is at least $b_i^{(t)} = b_i - \frac{w_1 + w_2}{s_i} > b_i - \frac{3\bval/2}{2/\bval} =\frac{b_i}{4}$.
	%The latter fact implies that every bundle has at most one bad item.
	
	After the end of the bad sequence, %at most one item was removed from each bundle. In two bundles, the agent won an item, and is some bundles, the adversary won an item. There may also be bundles in which no one won an item.
	if there are items of value at least $\frac{3\bval}{2} - w_1 - w_2$, the agent bids its full remaining budget ($b_i^{(t)}$) on them (by rule~2). If the agent wins any such bid, we are done. Hence we may assume that the adversary wins all such items. This marks the end of the first stage of the agent's bidding strategy, as we have reached the second stopping condition.
	
	Before continuing to the second stage, we show that not all items are exhausted in the first stage. (Later we will address the question of the value of the items that remain.)
	
	\begin{proposition}
		\label{pro:charging}
		By the end of the first stage, either the agent accumulated a value of at least $\frac{3\bval}{2}$, or there are still items remaining for the second stage.
	\end{proposition}
	\begin{proof}
		Assume for the sake of contradiction that the agent accumulated a value smaller than $\frac{3\bval}{2}$, but no items remain. The agent then has a budget of $b_i^{(r)} > \frac{b_i}{4}$ remaining. We show that the fact that no items remain implies that the total budget of the adversary and the agent is exhausted, contradicting the fact that the agent's budget is not exhausted.
		
		Recall the bundles $\{B_j\}$ implied to exist by the condition for the APS.
		
		For every item $e$, let $p(e)$ denote the price paid for the item. Charge this price to the bundles that contain $e$, where if $e \in B_j$, bundle $B_j$ is charged $\frac{2\lambda_j}{b_i}p(e)$. By the condition $\sum_{j | e\in B_j} \lambda_j \le \frac{b_i}{2}$, the bundles are not charged more than the total price of the items. Let $C_j$ denote the total charge to bundle $j$. We show that $C_j \ge \lambda_j$.
		
		For a bundle $B_j$ that contains only one item $e$, we have that $v_i(e) \ge \bval$. Such an item cannot be underpriced. Hence $p(e) \ge \frac{b_i}{2}$, and $C_j \ge  \frac{2\lambda_j}{b_i}\cdot \frac{b_2}{2} = \lambda_j$.
		
		Consider now a bundle $B_j$ that contains at least two items, and let $e_1$ and $e_2$ be items in $B_j$. Both items must be consumed in the first stage.
		In the first stage, the agent always bids $b_i^{(t)}$ which is at least $\frac{b_i}{4}$ (since all items until $w_2$ are not underpriced and thus priced at least $\frac{b_i}{4}$, and between item $w_2$ and round $r$, the agent always bids $b_i^{(r)}> \frac{b_i}{4}$). Hence $p(e_1) + p(e_2) \ge 2b_i^{(r)} > \frac{b_i}{2}$. Again, $C_j \ge  \frac{2\lambda_j}{b_i} \frac{b_i}{2} = \lambda_j$.
		
		Hence the total price of all items is at least $\sum_j C_j \ge \sum_j \lambda_j = 1$, implying that the total budget (which was~1) is exhausted.
	\end{proof}

	For the items that remain, we change the valuation function of the agent to $\hat{v}_i$, where for every item $e$ we define $\hat{v}_i(e) = \min[b_i^{(r)}\cdot s_i, \frac{b_i - 2b_i^{(r)}}{b_i^{(r)}}v_i(e)]$. Observe that $\hat{v}_i(e) \ge v_i(e)$ for every item $e$, because when the first stage ends there is no remaining item with $v_i(e) > b_i^{(r)} \cdot s_i$ (otherwise the agent will still play according to rule~2), and also $\frac{b_i - 2b_i^{(r)}}{b_i^{(r)}} > 1$ as $b_i^{(r)} < \frac{1}{3}$.
	
	We will use $\hat{x}_i^{(t)}$ to denote the corresponding maximal value with  respect to $\hat{v}_i$ at time $t$.
	
	In the {\bf second stage}, the bidder bids ``truthfully" with respect to $\hat{v_i}$.
	
	\begin{itemize}
		\item In every round $t\geq r$, if $\hat{x}_i^{(t)} \le b_i^{(t)}$ bid $\frac{\hat{x}_i^{(t)}}{s_i}$. (The analysis will show that if in a round it holds that $\hat{x}_i^{(t)}  > b_i^{(t)}$, we already have $u_i^{(t)} \ge \frac{3\bval}{2}$.)
		
	\end{itemize}
	
	%Let $B'$ be the total budget (of the adversary and the agent combined) when the second stage begins. 
	We {\bf claim} that when the second stage begins, the total value $\hat{s}_i^{(r)}$ of all items (according to $\hat{v}_i$) is at least $B^{(r)} \cdot s_i$. We postpone the proof of this claim, and first show that the claim implies our lemma.
	
	As $\max_e[\hat{v}_i(e)] \le b_i^{(r)}\cdot s_i$, Lemma~\ref{lem:safeStrategy} (for $t=1$) implies that the agent can guarantee (by following the strategy described in the lemma) a value (according to $\hat{v}_i$) of at least $\frac{b_i^{(r)}\hat{s}_i^{(r)}}{2B^{(r)}} \ge \frac{b_i^{(r)}\cdot s_i}{2}$. Scaling back to values according to $v_i$, this gives at least $\frac{s_i \cdot (b_i^{(r)})^2}{2b_i - 4b_i^{(r)}}$. Recall that the value accumulated by the agent in the first stage is $s_i \cdot(b_i - b_i^{(r)})$, 
	and hence in the second stage it suffices that the agent accumulates $s_i(b_i^{(r)} - \frac{b_i}{4})$. And indeed, $\frac{s_i \cdot (b_i^{(r)})^2}{2b_i - 4b_i^{(r)}} \ge s_i\cdot (b_i^{(r)} - \frac{b_i}{4})$. 
	(Dividing by $s_i$ then multiplying by the denominator gives $(b_i^{(r)})^2 \ge 2b_i\cdot b_i^{(r)} - \frac{b_i^2}{2} - 4(b_i^{(r)})^2 + b_i \cdot b_i^{(r)}$, and moving everything to the left hand side and completing a square gives $(\sqrt{5} b_i^{(r)} - \frac{b_i}{\sqrt{2}})^2 + (\sqrt{10} - 3)b_i\cdot b_i^{(r)} \ge 0$.)
	
	It remains to prove the {\bf claim}.  Consider the prices and the charging mechanism explained in the proof of Proposition~\ref{pro:charging}. Let $L_j$ (``leftover budget for bundle $j$") be $L_j = \max[0, \frac{b_i}{2} - \sum_{e \in B_j} p(e)]$. (Here $e$ ranges only over items consumed in the first stage.) The total budget left at the end of the first stage can be seen to satisfy $B^{(r)} \le \sum_j \frac{2\lambda_j}{b_i} L_j$.
	
	Let $V_j$ denote the value left in bundle $B_j$ at the end of the first stage according to $v_i$ , and let $\hat{V}_j \ge V_j$ denote the value after scaling $v_i$ to $\hat{v}_i$. Then the total value according to $\hat{v}_i$ at the beginning of the second stage satisfies $\hat{s}_i^{(r)} = \sum_j \frac{2\lambda_j}{b_i} \hat{V}_j$. To prove that $\hat{s}_i^{(r)} \ge B^{(r)} \cdot s_i$ we shall show that $\hat{V}_j \ge L_j\cdot s_i$ for every bundle $B_j$.
	
	If bundle $B_j$ does not contain any underpriced item, then $V_j \ge L_j \cdot s_i$, and we are done. If two or more items are consumed in $B_j$ in the first stage then $L_j = 0$ (see proof of Proposition~\ref{pro:charging}) and again $V_j \ge L_j \cdot s_i=0$. It remains to consider the cases in which exactly one item $e\in B_j$ was consumed in the first stage, and this item $e$ was underpriced. In this case, $p(e) = b_i^{(r)}$ and $L_j = \frac{b_i}{2} - b_i^{(r)}$. Moreover, $v_i(e) \le w_2$ and $V_j \ge \frac{b_i\cdot s_i}{2} - w_2 \ge \frac{b_i\cdot s_i}{2} -  s_i \cdot \frac{b_i - b_i^{(r)}}{2} = \frac{s_i \cdot b_i^{(r)}}{2}$. Consequently, we have that $\hat{V}_j \ge \min[s_i \cdot b_i^{(r)}, s_i \cdot \frac{b_i - 2b_i^{(r)}}{b_i^{(r)}}\frac{b_i^{(r)}}{2}]$. Observe that $b_i^{(r)} > L_j$, because $b_i^{(r)} > \frac{b_i}{4}$. Also,
	$\frac{b_i - 2b_i^{(r)}}{b_i^{(r)}}\frac{b_i^{(r)}}{2} = \frac{b_i}{2} - b_i^{(r)} = L_j$, so indeed we have that $\hat{V}_j \ge L_j \cdot s_i$, as desired.
\end{proof}

\subsection{A Polynomial Time $\frac{3}{5}$-APS Allocation}
\label{app:V}

The fact that the APS is NP-hard to compute in polynomial time might cause a problem for those algorithms that make explicit use of the value of the APS. In our work, there are several such algorithms, one in Theorem~\ref{thm:app-APS}, one in Lemma~\ref{lem:35} (this lemma is used by Theorem~\ref{thm:app-APS}), and one in Lemma~\ref{lem:3/4} (this lemma is used by Lemma~\ref{lem:35}). In this section we show how in these cases the need to compute the value of the APS can be avoided, while still obtaining the guarantees of the underlying algorithms.
{We assume the value of every item is a non-negative integer, and the entitlement is a rational positive number. }

These algorithms use a parameter $z$ that they initially set to be equal to the APS. Given $z$, the algorithms run in polynomial time. Moreover, their guarantees are with respect to $z$ (e.g., giving at least $\frac{3}{5} z$ Lemma~\ref{lem:35}), and not with respect to the APS.

One can ask what guarantees do these algorithms give if one initializes $z$ with a value not equal to the APS. Inspecting these algorithms, we see that if $z$ is smaller than the APS, the guarantees with respect to $z$ still hold, but they do not necessarily imply the corresponding guarantee with respect to the APS (as $\frac{3}{5} z$ is smaller than $\frac{3}{5} \cdot APS$). On the other hand, if $z$ is larger than the APS, the guarantees with respect to $z$ need not hold, but if they do hold, they do imply the desired guarantees also with respect to the APS.

In general, we can partition the possible values of $z$ into two classes. The {\em good} class contains those $z$ for which the guarantees of the algorithm do hold with respect to $z$. In particular, all values of $z$ up to and including the APS are good. The {\em bad} class contains all values of $z$ for which  the guarantees of the algorithm do not hold with respect to $z$.

Suppose that we had a polynomial time algorithm to determine whether a given value of $z$ is good or bad. Then in polynomial time, we could do a binary search on all values of $z$ up to $v_i(\items)$, and find a $z$ that is good, and satisfies also that $z+1$ is bad. This $z$ is at least as large as the APS, and using it in the respective algorithm provides the desired guarantees also with respect to the APS.

It remains to design a polynomial time algorithm that determines whether a given value of $z$ is good. For the algorithms that concern us here, there are two observations whose combination implies that this can be done.

One observation is that given {$z$}, these algorithms present an explicit strategy for the bidding game, and moreover, the guarantee holds regardless of the strategies of other agents in this game. Hence to see if the guarantee holds, it suffices to consider a {\em worst case adversary} for the bidding game. This adversary knows $v_i$ and $z$, and pools together the combined budgets of all other agents, as if there is just one other agent. (Here we use the fact that the value of the APS with entitlement $b_i$ does not depend on how the remaining entitlement is distributed among the other agents.) The adversary views the bidding game as a {zero} sum game between himself as one player and agent $i$ as the other player, and the adversary's goal is to bid and select items in a way that would minimize the final value obtained by agent $i$. As the strategy of agent $i$ is known, the optimal strategy for the adversary is a solution of a well defined optimization problem (there are no game theoretic aspects involved). 

Summarizing, the first observation shows that in order to test whether $z$ is good, it suffices to solve a well defined optimization problem. E.g., for Lemma~\ref{lem:35} this problem is to determine whether the adversary has a strategy that limits the value received by agent $i$ to below $\frac{3}{5} z$.

The second observation is that the algorithms of Lemma~\ref{lem:35} and Lemma~\ref{lem:3/4} have the following property. Once agent $i$ wins two items, we know whether the guarantees are met. For Lemma~\ref{lem:3/4} this holds because at the time we switch to the second stage of that algorithm, the agent has won two items (denoted to $w_1$ and $w_2$). As the second stage of the algorithm uses the algorithm of Lemma~\ref{lem:safeStrategy}, we can at that point check whether the conditions of Lemma~\ref{lem:safeStrategy} hold. These conditions can be checked in polynomial time, and they compare the maximum valued item to the proportional share. If the conditions hold, we are sure to succeed. If the conditions do not hold, $z$ was bad (too large). As for Lemma~\ref{lem:35}, its first steps either win one or two items and give a value of $\frac{3}{5}$, or win no item and instead call Lemma~\ref{lem:3/4}. Hence again, after the agent wins two items (in the algorithm of Lemma~\ref{lem:3/4}), we can tell whether the value of $z$ is good or bad.

Given the second observation, it is easy to solve the optimization problem of the adversary, as there are at most ${m \choose 2}$ adversarial strategies that we need to try. In each such strategy, we order the items in descending order of value according to $v_i$, and fix two items that agent $i$ is supposed to win, say items $k$ and $\ell$ in this order. In the first $k-1$ rounds the adversary wins all items and pays the bids of agent $i$, thereafter agent $i$ wins item $k$ and pays her bid, thereafter the adversary wins the items up to $\ell - 1$ and pays the bids of agent $i$, and then agent $i$ wins item $\ell$ and pays her bid. If at least one of these strategies is feasible for the adversary (the adversary can afford all its payments using its budget) and certifies that $z$ is bad, then $z$ is indeed larger than the APS. Else, $z$ is good.

Hence using the second observation we have a polynomial time algorithm that tests whether a given value of $z$ is good or bad. As explained above, this shows that the guarantee of $\frac{3}{5} \cdot APS$ can be achieved in polynomial time, even without knowing the exact value of the APS. Likewise, one can decide in the algorithm of Theorem~\ref{thm:app-APS} which of the three strategies offers the highest of the three guarantees, even without knowing the value of the APS.

%\pagebreak
\section{Missing Proofs from Section~\ref{sec:greedy-efx}}
\label{app:greedy-efx}

Here we present the definition of the Greedy-EFX algorithm. %  (defined for agent with equal entitlements).}
\begin{algorithm}
	\caption{Greedy-EFX~\citep{LiptonMMS04,BK20}}%\label{alg:greedy-efx}
	\begin{algorithmic}[1]
		\STATE Input: Set of items $\items=\{e_1,\ldots,e_m\}$ and valuation profile $\vect{v}=(v_1,\ldots,v_n)$ where $v_i(e_j)\geq v_i(e_k)$ for all $i\in \agents$ and every $j\leq k\leq m$.
		 
		\FOR{$t= 1,\ldots,m $}
		\STATE Give item $e_t$ to an agent that no one envies, breaking ties arbitrarily 
		\WHILE{There exists an envy cycle}
		\STATE Resolve {an (arbitrary)} envy cycle by rotating sets along the cycle
		\ENDWHILE
		\ENDFOR
	\end{algorithmic}
\end{algorithm}

\thmordered*
\begin{proof}[Sketch] %{\bf URI Move this sketch of proof to appendix.}
	A {\em choosing sequence} is a sequence on names of agents (repetitions are allowed). Starting with the beginning of the choosing sequence, in each round the agent whose name appears in the choosing sequence selects an item of her choice. The allocation $A$ for the ordered instance naturally induces a choosing sequence, where for every round $r$, the agent to choose in round $r$ is the one to which $A$ allocated the $r$th most valuable item in the ordered instance. Using this choosing sequence (and the pigeon-hole principle), in every round, every agent can choose an item that she values at least as much as the corresponding item that she got under $A$.
\end{proof}

\thmgreedyEFX*

\begin{proof}
%It is easy to see (if needed, see details in \cite{BK17}) that the allocation that it outputs is EFX. (This holds for the ordered instance, not necessarily for the allocation derived from it for the original instance.)
Observe that under greedy-EFX, every agent $q$ gets at least her $n$th most preferred item. Hence if her TPS is nonzero (which can only happen if there are at least $n$ items of positive value for  $q$), then the allocation gives the agent positive value.

Consider an arbitrary agent $q$, who gets bundle $B_q$, and let $v_q$ be her valuation function. We may assume without loss of generality that $v_q(B_q) > 0$, because if $v_q(B_q) = 0$ the TPS (and APS) of agent $q$ is~0. To simplify subsequent notation, we scale the valuation function of $q$ by a multiplicative factor, so that that $v_q(B_q) = 1$. Consequently, in order to prove Lemma~\ref{lem:BK17}, we need to either show that $APS_q \le \frac{4}{3}$, or to show that $TPS_q \le \frac{3n-1}{2n}$.

We shall use the following {\em properties} of greedy-EFX. (Values in these properties are taken with respect to $v_q$, and inequalities implicit in these properties need not be strict.)

\begin{enumerate}

\item Within a bundle, the items enter the bundle in order of decreasing value.

\item If $i < j$ and $B_i$ and $B_j$ each contain at least two items, then the first item in $B_i$ is more valuable than the first item in $B_j$, and the second item in $B_i$ is less valuable than the second item in $B_j$.

\item Removing the last item of any bundle, the value that remains in the bundle is at most $v_q(B_q) = 1$. (This is the EFX property.)

\end{enumerate}

Let $k_1$ denote the highest index of a bundle that contains only one item (with $k_1 = 0$ if there is no such bundle). 

\begin{claim}
\label{claim:atMost2}
If in the allocation produced by greedy-EFX every bundle contains at most two items, then $APS_q \le 1$. 
\end{claim}

\begin{proof}
Observe that every bundle up to $B_{k_1}$ contains one item, and each of the remaining bundles (if $k_1 < n$) contains two items. Fix $\epsilon < \frac{1}{4n^2}$. There are two cases to consider.

If $q \le k_1$ (in which case $q$ receives only one item, $e_q$), we upper bound $APS_q$ by considering the following price function $p$ for the items.  For $i \le k_1$ we have $p(e_i) = \frac{1}{n} + \epsilon$, with the exception that $p(e_q) = \frac{1}{n} - 2n\epsilon$. For $i > k_1$ we have $p(e_i) = \frac{1}{2n} + \epsilon$. The total price of all items is less than~1, because the price of every bundle is at most $\frac{1}{n} + 2\epsilon$, except for bundle $B_q$ that has price $\frac{1}{n} - 2n\epsilon$. Agent $q$ who has budget $\frac{1}{n}$ can afford to buy only a single item, as the sum of prices of any two items exceeds $\frac{1}{n}$. Among the items priced at most $\frac{1}{n}$, the most valuable item is $e_q$, and we assumed that $v_q(e_q) = 1$. Hence $APS_q \le 1$. 

If $q > k_1$ (in which case $q$ receives two items, $e_q$ and $e_{2n-q+1}$), we consider the following price function $p$. For $i \le k_1$ we have $p(e_i) = \frac{1}{n} + \epsilon$. For $k_1 < i < q$ we have $p(e_i) = \frac{2}{3n} + \epsilon$. For $q \le i \le 2n-q$ we have $p(e_i) = \frac{1}{2n} + \epsilon$. For $i \ge 2n-q+1$ we have $p(e_i) = \frac{1}{3n} + \epsilon$. The total price of all items is less than~1, because the price of every bundle is at most $\frac{1}{n} + 2\epsilon$, except for bundle $B_q$ that has price $\frac{5}{6n} + 2\epsilon$. Agent $q$ can afford to buy either a single item among $e_{k_1+1}, \ldots, e_{q-1}$, or two items: one among $e_q \ldots, e_m$, and the other among $e_{2n - q + 1}, \ldots, e_m$. When buying two items, the most valuable two are $e_q$ and $e_{2n-q+1}$, because the instance is ordered. These are precisely the items in bundle $B_q$. The value of every item among $e_{k_1+1}, \ldots, e_{q-1}$ is at most~1, by property~3 above. Hence $q$ can afford to buy a bundle of value at most~1, implying that $APS_q \le 1$.
\end{proof}

%In~\cite{BK17} it is shown that the greedy-EFX allocation gives every agent at least a $\frac{2n}{3n-1}$ fraction of his maximin share. In Theorem~\ref{thm:BK17} we extend this result to the anyprice share. One may attempt to adapt the proof of~\cite{BK17} that worked for the maximin share so that it applies to the seemingly similar fractional maximin share (which is equivalent to the anyprice share). However, our proof takes a different route, partly because we wish the proof of Theorem~\ref{thm:BK17} to serve as a stepping stone towards the proof of Theorem~\ref{thm:BK17a}. We work directly with the anyprice share, and exhibit item prices that certify that the anyprice share is not much larger than the value of the bundle that greedy-EFX allocates to the agent.

%Consider an arbitrary agent $q$, who gets bundle $B_q$, and let $v_q = v$ be his valuation function. As noted above, we may assume that $v(B_q) = 1$. 

It remains to handle the case that some bundle has at least three items.
Let $k$ be the largest index of a bundle that has more than two items. The first item added to $B_k$ was $e_k$, the second item was $e_{2n - k + 1}$, and $B_k$ has at least one more item. We shall refer to all items from  $e_{2n - k + 1}$ onwards as {\em small}. Note that $v_q(e_k) + v_q(e_{2n-k+1}) \le 1$, because otherwise $B_k$ would not get a third item. The fact that the instance is ordered then implies that small items have value at most $\min[\frac{1}{2}, 1 - v_q(e_k)]$.  We now perform a case analysis over the possible values of $v_q(e_k)$.

\begin{claim}
\label{claim:3/4}
If $v_q(e_k) \le \frac{3}{4}$ then $TPS_q \le \frac{3n-1}{2n}$.
\end{claim}

\begin{proof} Suppose first that every bundle (except possibly for $B_q$) has at least two items. Then the inequality $v_q(e_k) \le \frac{3}{4}$, together with the fact that small items have value at most $\frac{1}{2}$, implies that no bundle has value above $\frac{3}{2}$. (For bundles beyond $B_k$ this holds because each of their two items has value at most $\frac{3}{4}$, and for bundles up to $B_k$ this holds because their last item has value at most $\frac{1}{2}$.) Consequently, the combined value of all items is at most $1 + (n-1)\frac{3}{2} = \frac{3n-1}{2}$, and $TPS_q \le PS_q \le \frac{3n-1}{2n}$.

It remains to consider the case that there are bundles with just one item (namely, $k_1 \ge 1$). Let $K_1$ denote the set of items in these bundles (one item per bundle), but if $B_q$ happens to have one item, we exclude $e_q$ from $K_1$.
Note that $|K_1| < n$. We remove the items of $K_1$ and the $|K_1|$ agents that received these items under greedy-EFX, getting a new instance in which the TPS of agent $q$ is lower bounded by the inequality  $TPS(\frac{1}{ n - |K_1|},v_q, \items \setminus K_1) \ge TPS(\frac{1}{n},v_q,\items)$. For the new instance, the argument of the preceding paragraph implies that $TPS(\frac{1}{ n - |K_1|},v_q, \items \setminus K_1) \le  \frac{3(n-|K_1|)-1}{2(n-|K_1|)} < \frac{3n-1}{2n}$. Consequently, $TPS_q < \frac{3n-1}{2n}$.
\end{proof}

Following Claim~\ref{claim:3/4}, we may assume that $v(e_k) > \frac{3}{4}$. This case further breaks into subcases, and for each of the subcases we bound the APS by introducing an appropriate price function $p$ (an approach similar to that used for proving Claim~\ref{claim:atMost2}). In order to simply notation, we scale the budget of $q$ and all prices by a carefully chosen factor $f$, where $f$ will be slightly smaller than $n$. Using this scaling, to show that the AnyPrice share of agent $q$ is at most some value $w$, we shall associate nonnegative prices with items, satisfying the following properties:

\begin{enumerate}

\item In every bundle, the sum of prices of items is at most~1.

\item In at least one bundle, the sum of prices of items is strictly smaller than~1.

\item For every set of items of value larger than $w$, the sum of prices is at least~1.

\end{enumerate}

The first two properties imply that the total sum of prices is strictly less than~$n$. Hence for the purpose of computing the AnyPrice share, agent $q$ has a budget strictly smaller than~1. The third property implies that with such a budget, $q$ cannot afford to buy items of total value larger than $w$.

\begin{claim}
\label{claim:5/6}
If $v_q(e_k) > \frac{5}{6}$ then $APS_q < \frac{4}{3}$.
\end{claim}

\begin{proof}
Recall that $B_{k_1}$ denotes the last bin among those that have only one item, and note that $k_1 < k$ (as $B_k$ has more than two items). The price function $p$ is as follows: for $i \le k_1$ we have $p(e_i) = 1$, for $k_1 < i \le 2n-k$ we have $p(e_i) = \frac{1}{2}$, and for $i \ge 2n-k+1$ (small items) we have $p(e_i) = \frac{3}{2}v_q(e_i)$. Recall that every small item has value at most $1 - v_q(e_k)$, and under the conditions of the claim this value is strictly smaller than $\frac{1}{6}$. 

Every bundle up to~$k_1$ has price~1.  Every bundle in the range $[k_1 + 1, k]$ has one item of value at least $v_q(e_k) > \frac{5}{6}$, and small items of total value strictly less than $2\cdot \frac{1}{6} = \frac{1}{3}$, and hence its price is strictly smaller than $\frac{1}{2} + \frac{3}{2} \cdot \frac{1}{3} = 1$. Every bundle from $k+1$ onwards has two items and hence a price of at most~1. 

Agent $q$ (with a budget strictly smaller than~1) can afford at most one of the items $e_{k_1 + 1}, \ldots, e_{2n-k}$ (giving value at most~1), and has a budget smaller than $\frac{1}{2}$ left with which to purchase small items, giving additional value smaller than $\frac{2}{3}\cdot \frac{1}{2} = \frac{1}{3}$. Hence $APS_q < \frac{4}{3}$.
\end{proof}

%In the proof of the following claim, we skip some of the steps that are explained in the proof of Claim~\ref{claim:5/6}.

\begin{claim}
\label{claim:middle}
If $\frac{3}{4} < v_q(e_k) \le \frac{5}{6}$ then $APS_q < \frac{4}{3}$.
\end{claim} 

\begin{proof} The price function is as follows: for $i \le k_1$ we have $p(e_i) = 1$, for $k_1 < i \le 2n-k$ we have $p(e_i) = \max[\frac{1}{2}, v_q(e_i) - \frac{1}{3}]$, and for $i \ge 2n-k+1$ (small items) we have $p(e_i) = v_q(e_i)$. 
Every small item has value at most $1 - v_q(e_k)$, and under the conditions of the claim this value is strictly smaller than $\frac{1}{4}$.  

Every bundle up to~$k_1$ has price~1. Every bundle in the range $[k_1 + 1, k]$ contains one item $e$ of value $v_q(e) \ge v_q(e_k) > \frac{3}{4}$, and either one small item (of value strictly less than $\frac{1}{4}$), or several small items of  total value at most $2(1 - v_q(e))$. In either case,  $p(e_i) = \max[\frac{1}{2}, v_q(e_i) - \frac{1}{3}]$ implies that the price of the bundle is strictly smaller than $1$. Every bundle from $k+1$ onwards has two items, none of which has value above $\frac{5}{6}$, and hence the price of the bundle is at most~1. 

Agent $q$ can afford at most one item among $e_{k_1+1}, \ldots, e_{2n-k}$, and hence can buy at most one item at a price below its value. As $q$'s budget is less than~1, and the gain of value compared to price paid is at most $\frac{1}{3}$, she obtains value less than $\frac{4}{3}$. Hence $APS_q < \frac{4}{3}$.
\end{proof}

The four claims above complete the proof of Lemma~\ref{lem:BK17}.
\end{proof}

{Recall that in every round, the greedy-EFX allocation algorithm gives an item to an agent that has a bundle that no one envies. If there are several such agents, it may select one of them arbitrarily. Depending on the arbitration rule, we get different versions of greedy-EFX. The positive results of Lemma~\ref{lem:BK17} hold for all such versions. The negative results of the following Proposition~\ref{pro:EFXupper} are proved with respect to one such natural version. The key assumption of our version is that the first item goes to agent~1. (For convenience we also assume that if in a certain round the arbitration rule preferred an agent $i$ over an agent $j$, then in the next round in which both $i$ and $j$ are eligible candidates to receive an item, agent $j$ is preferred over agent $i$. However, this last assumption can be removed by slightly perturbing the values of items in the example. Details of how to perform these perturbations are omitted.)}

\EFXupper*

\begin{proof}
	Let $k =\lceil \frac{1}{2\epsilon}\rceil$.
	We consider setting with $n=\sum_{i=0}^k 4^i=\frac{ 4^{k+1} -1}{3}$ agents, and a set of $m=3\cdot \sum_{i=0}^k 4^i = 4^{k+1}-1$ items $\items =\big\{(a,b,c) \mid a\in\{0,1,\ldots,k\},b\in\{1,2,3\},c \in \{1,2,\ldots,4^a\}\big\}$. For agent $1$, the value of item $(a,b,c)$ is $1-\frac{a}{2k}$ if $b=1$, and $\frac{a}{2k}$ if $b\neq 1$.  
	We call items with $b=1$ large items, and items with $b\in \{2,3\}$ small items.
	For every $0 \le i \le k$, we refer to the items with $a=i$ as belonging to {\em group} $i$. 
	Hence group $i$ contains $4^i$ large items and $2\cdot 4^i$ small items.
	For every agent other than agent~1, the value of item $(a,b,c)$ is $1-2^{-k+a-1}$ if $b=1$, and $2^{-k+a-2}$  if $b\neq 1$. % for small enough $\delta$. 
	
	{We next claim that if the first item that the greedy-EFX algorithm allocates is assigned to agent $1$, then agent $1$ gets the three items of group $0$, with values $1,0$ and $0$. 
	%OLD:
	%The greedy-EFX algorithm %~\ref{alg:greedy-efx}
    %can give Agent $1$ items that are worth $1,0,0$. 
	When  
	agent $1$ receives the first item (the large item of group 0), then each other agent receives a single large item (as they value all large items by at least half, while all small items by less than half).
	After the allocation of the large items, the greedy-EFX allocates the small items in decreasing order of the groups. Each agent that received a large item of group $i$ also receives two small items of the same group.
	Thus, agent $1$ won't get another item until all items but the small items of group $0$ are allocated.}

	{While the greedy-EFX algorithm gives agent~1 value of $1$, we next show that
	%In contrast, 
	the MMS {of agent~1} is at least $\frac{3}{2}-\epsilon$. This holds since the agent can partition  the items to $n$ bundles, each with value at least $\frac{3}{2}-\epsilon$, in the following way:}
	For each $a\in \{1,2,\ldots,k\}$, the agent creates $2\cdot 4^{a-1}$ bundles by taking two large items of group $a$ and one small item of group $a-1$.
	The value of such bundle  is $2(1-\frac{a}{2k}) + \frac{a-1}{2k} = 2-\frac{a+1}{2k}\geq \frac{3}{2}-\epsilon  $.
	{The agent also creates $ \frac{2\cdot 4^k-2}{3} $ bundles, each containing three small items of group $k$, and each such bundle has a value of $\frac{3}{2}$.} 
% OLD: The agent also creates $ \frac{2\cdot 4^k-2}{3} $ bundles using for each bundle $3$ items from the  small items of group $k$, which has a value of $\frac{3}{2}$. 
	Note that $2\cdot 4^k-2$ is divisible by $3$.
	The agent also creates another bundle using  {the two remaining}  small items of group $k$, and the large item of group $0$. {That bundle has value of $2$.} %Each such bundle has a value of either  $\frac{3}{2}$ or $2$ which is at least $\frac{3}{2} -\epsilon$.
	%Since  the number of small items of group $k$ is not divisible by $3$, there must remain another small item, which together with the large item of group $0$, gives another bundle of value $\frac{3}{2} \geq \frac{3}{2}-\epsilon$. \tee{
	The number of bundles created is then $$\sum_{a=1}^k 2 \cdot 4^{a-1} + \frac{2\cdot 4^k-2}{3} +1  =\frac{2\cdot 4^k-2}{3} + \frac{2\cdot 4^k+1}{3} =  n,$$
	which 	concludes the proof.
\end{proof}

%\section{Greedy EFX for identical agents}

We next consider the outcome of the greedy-EFX algorithm for the case of agents with equal entitlements and identical valuations.
\begin{theorem}
\label{thm:BK17a}
When $n$ agents have additive {\em identical} valuations and equal entitlements, the greedy-EFX allocation gives every agent more than a $\frac{3}{4}$ fraction of her AnyPrice share.
\end{theorem}

\begin{proof}
{Lemma~\ref{lem:BK17} concerns arbitrary additive valuations. Its proof was based on analysing several cases, handled by four different claims. For every claim, the bounds proved in that lemma apply also to additive identical valuations. In three of the claims (Claim~\ref{claim:atMost2}, Claim~\ref{claim:5/6}, Claim~\ref{claim:middle}), it was established that the agent receives  more than a $\frac{3}{4}$ fraction of her AnyPrice share. 
The exception is Claim~\ref{claim:3/4}, that considers the range of values $v_q(e_k) \le \frac{3}{4}$. (Recall that $q$ is an agent that receives an allocation of value~1, index $k$ is the largest index of a bundle that contains more than two items, and $e_k$ is the $k\mbox{-}th$ most valuable item, and hence the most valuable item in bundle $B_k$.) Hence in the current proof it suffices to address the case that $v_q(e_k) \le \frac{3}{4}$. We shall address this case by further partitioning it into three cases. In two of these cases we shall not require valuation functions of different agents to be identical, and only the third of these cases makes use of the assumption of identical valuations.

\begin{claim}
\label{claim:1/3}
If $v_q(e_k) < \frac{1}{3}$ then $TPS_q < \frac{4}{3}$, and hence $APS_q < \frac{4}{3}$. This holds even if the valuation functions of agents are not identical. 
\end{claim} 

\begin{proof} 
If $v_q(e_k) < \frac{1}{3}$, then every bundle either has only a single item, or the last item to enter the bundle has value at most $v_q(e_k) < \frac{1}{3}$, and hence the bundle has value strictly less than $\frac{4}{3}$. Hence the  truncated proportional share of agent $q$ is smaller than $\frac{4}{3}$, and she gets more than a $\frac{3}{4}$ fraction of $TPS_q$.
\end{proof}

In order to upper bound the APS in the proofs of the following Claims~\ref{claim:2/3} and~\ref{claim:middleThird}, we shall use the approach used in the proof of Lemma~\ref{lem:BK17}, pricing the items so that the total price $P$ is strictly less than~$n$, and giving the agent a budget of  $\frac{P}{n}< 1$.}

\begin{claim}
\label{claim:2/3}
If $\frac{2}{3} \le v_q(e_k) \le \frac{3}{4}$ then $APS_q < \frac{4}{3}$. {This holds even if the valuation functions of agents are not identical.} 
\end{claim}

\begin{proof}
Recall that $B_{k_1}$ denotes the last bin among those that have only one item, and note that $k_1 < k$ (as $B_k$ has more than two items). The price function $p$ is as follows: for $i \le k_1$ we have $p(e_i) = 1$, for $k_1 < i \le 2n-k$ (medium items) we have $p(e_i) = \max[\frac{1}{2}, v_q(e) - \frac{1}{4}]$, and for $i \ge 2n-k+1$ (small items) we have $p(e_i) = \min[\frac{1}{4}, v_q(e)]$. Recall that every small item has value at most $1 - v_q(e_k)$, and under the conditions of the claim this value is at most $\frac{1}{3}$. 
(Remark: if $q > k$, the price of item $e_{2n - q + 1}$ is slightly decreased, as explained shortly.)

Every bundle up to~$k_1$ has price~1. Every bundle from $k+1$ up to $n$ has two items, none of which has value above $\frac{3}{4}$, and hence the price of the bundle is at most~1. For every bundle $B_i$ with $i$ in the range $[k_1 + 1, k]$, its first item has value at least $v_q(e_k) \ge \frac{2}{3}$. If $B_i$ has only one small item, then $p(B_i) \le \max[\frac{1}{2}, v_q(e_i) - \frac{1}{4}] + \frac{1}{4} \le 1$. If $B_i$ has two small items, then $p(B_i) \le \max[\frac{1}{2}, v_q(e_i) - \frac{1}{4}] + 2\min[\frac{1}{4}, 1 - v(_qe_i)] \le 1$. If $B_i$ has three or more small items then $p(B_i) \le \max[\frac{1}{2}, v_q(e_i) - \frac{1}{4}] + \frac{3}{2}(1 - v_q(e_i)) \le 1$, where the last inequality holds since $v_q(e_i) \ge \frac{2}{3}$.

If $q  \le k$ then $p(B_q) < 1$ (because $p(e_q) < v(e_q)$ whereas $v(B_q) = 1$). If $q > k$, then $v_q(B_q) = 1$ implies that for the second item in $B_q$ (which is a medium item) we have $v_q(e_{2n-q+1}) \le \frac{1}{2}$. We change its price to $\frac{1}{2} - \epsilon$ for a very small $\epsilon > 0$, so that $p(B_q) < 1$. The sum of prices is now strictly smaller than $n$, as desired.

Purchasing a medium item gives the agent a value gain of at most $\frac{1}{4}$ compared to the price paid. Purchasing a small item gives the agent a value gain of at most $\frac{1}{3} - \frac{1}{4} = \frac{1}{12}$ compared to the price paid. With a budget less than~1, the combination that gives the agent the largest gain is that of purchasing one large item and one small item (purchasing $e_{2n-q+1}$ as a second large item might be feasible, but it gives a smaller gain), for a gain of at most $\frac{1}{4} + \frac{1}{12} = \frac{1}{3}$. As the budget is less than~1, the AnyPrice share is less than $\frac{4}{3}$.
\end{proof}

It remains to deal with the case that $\frac{1}{3} \le v_q(e_k) < \frac{2}{3}$.  Here we shall use the assumption that all agents have the same valuation function $v = v_q$. This assumption implies that every item is added to a bundle that at the time has the smallest $v$ value.

\begin{claim}
\label{claim:middleThird}
If agents have identical additive valuations and $\frac{1}{3} \le v_q(e_k) < \frac{2}{3}$, then $APS_q < \frac{4}{3}$. 
\end{claim} 

\begin{proof}
In this proof we change the meaning of the terminology {\em large}, {\em medium} and {\em small}, compared to the previous cases. 
The price function is as follows. (Remark: we shall later increase some of these prices.)

\begin{itemize}

\item Items $e$ of value $v(e) > 1$ are referred to as {\em huge}, and are priced $p(e) = 1$.

\item
Items $e$ of value $\frac{2}{3} \le v(e) \le 1$ are referred to as {\em large}, and are priced $p(e) = \frac{2}{3}$.

\item Items $e$ of value  $\frac{1}{3} \le v(e) < \frac{2}{3}$ are referred to as {\em medium}, and are priced $p(e) = \frac{1}{3}$.

\item Items $e$ of value $v(e) < \frac{1}{3}$ are referred to as {\em small}, and are priced $p(e) = v(e)$.

\end{itemize}

Let us first confirm that no bundle is priced more than~1. This holds for bundles with at most two items, because $v(e_k) < \frac{2}{3}$, and no bundle can contain two items of price $\frac{2}{3}$.
It remains to consider bundles with at least three items.

At least one bundle that contains more than a single item has value above $\frac{4}{3}$, as otherwise the truncated proportional share is strictly smaller than $\frac{4}{3}$. Let $B_j$ be such a bundle. Necessarily, $B_j$ has either two or three items. Let $e_s$ denote the smallest item in $B_j$, and note that necessarily $v(e_s) = \frac{1}{3} + \delta$ for some $\delta > 0$. If $B_j$ contains three items, then necessarily $\delta \le \frac{1}{6}$.
%Likewise, if $B_j$ has two items {\em and} $v(e_k) = \frac{1}{2}$, then $\delta < \frac{1}{6}$. The value of the remaining items in $B_j$  is greater than $1 - \delta$.
Likewise, $\delta \le \frac{1}{6}$ also if $B_j$ has two items. This is because $v(e_k) < \frac{2}{3}$, implying that $j < k$, and consequently $e_s$ has smaller value than the second item in $B_k$.  That is, $v(e_s) \le v(e_{2n - k + 1}) \le \frac{1}{2}$.

Consider any bundle $B_i$ with three items (which may also be $B_j$ itself, if $B_j$ has three items). If $B_i$'s second item was added before $e_s$ (this is the case for $B_j$), then the value of this second item is above $v(e_s) > \frac{1}{3}$, meaning that every item in $B_i$ has value smaller than $\frac{2}{3}$, and hence each of its three items has price at most $\frac{1}{3}$. If $B_i$'s second item was added after $e_s$, then the first item in $B_i$ has value above  $1 - \delta$ (since this is a lower bound on the  value of  $B_j\setminus\{e_s\}$) and price $\frac{2}{3}$, and the sum of prices of the remaining two items in the bundle is at most $2\delta \le \frac{1}{3}$.

Consider now any bundle $B_i$ with four or more items. At most two items in $B_i$ have value larger than $e_s$. Regardless of whether there are two such items or only one, their sum of values is at least $1 - \delta$ (as the other items arrive after $e_s$, and $v(B_k) \ge 1 - \delta$ at the time of arrival of $e_s$), and their sum of prices is $\frac{2}{3}$. The sum of values (and hence also prices) of the remaining items in $B_i$ is at most $2\delta \le \frac{1}{3}$.

Recall that $v(B_q) = 1$. We claim that the price of $B_q$ is strictly smaller than~1. For every item $e$ we have that $p(e) \le v(e)$, with equality only if either $v(e) = \frac{2}{3}$ or $v(e) \le \frac{1}{3}$. Hence $p(B_q) \le 1$. Assume now for the sake of contradiction that $p(B_q) = 1$. Necessarily $v(e_q) > \frac{1}{3}$ (as otherwise $v(e_s) \le \frac{1}{3}$), and hence $v(e_q) = \frac{2}{3}$. This means that all other items of $B_q$ arrive after $e_s$ (as their total value is $\frac{1}{3}$ which is smaller than $v(e_s)$). But then the value in $B_q$ at the time that $e_s$ arrives (which is $\frac{2}{3}$) is smaller than the value in $B_j$ at that time (which is at least $1 - \delta \ge \frac{5}{6}$), contradicting the assumption that greedy-EFX places $e_s$ is $B_j$.

This completes the proof that the sum of prices is strictly smaller than $n$.

Given $\delta$ as defined above, we now increase the above prices, preserving the property that the price of every bundle is at most~1, and that the price of $B_q$ is strictly less than~1.

\begin{itemize}

\item
For every medium item $e$ of value $\frac{2}{3} - \delta < v(e) < \frac{2}{3}$, we increase its price to $p(e) = \frac{1}{2}$. We refer to these items as {\em special} items. Observe that if the first item in a bundle is special, then the second item has value at least $v(e_s) = \frac{1}{3} + \delta$, and hence the bundle has two items whose sum of prices is at most~1. (If the bundle is $B_q$, the second item cannot be special and the sum of prices remains strictly less than~1.) Also, there cannot be a bundle in which the first item is large and the second item is special, because $e_k$ is not large, and $e_{2n-k+1}$ is not special.

\item
The price of every small item $e$ is increased to $p(e) = \frac{v(e)}{6\delta}$ (making the price larger than the value). Note that small items belong to bundles in which the other items (or single item) have total value at least $1 - \delta$ (and total price $\frac{2}{3}$). As such, the total value of small items in the bundle is at most $2\delta$, and the total price of the bundle remains at most~1. (If $B_q$ contains small items, then their sum of values is at most $\delta$, and consequently $B_q$'s price remains strictly below~1.)

\end{itemize}

We now verify that agent $q$ cannot afford a bundle of value $\frac{4}{3}$.

The agent can afford at most one large item, paying $\frac{2}{3}$. If she does so, she has strictly less than $\frac{1}{3}$ budget left, and can buy small items of value strictly less than $\frac{1}{3}$, and the total value cannot reach $\frac{4}{3}$.

If the agent buys two medium items that are not special, she gets value at most $\frac{4}{3} - 2\delta$. The agent then has a budget strictly less than $\frac{1}{3}$ left, with which she can buy only small items. This gives additional value of strictly less than $\frac{1}{3} \cdot 6\delta = 2\delta$, and hence the total value purchased is strictly less than $\frac{4}{3}$.

If the agent buys one special item and one medium item that is not special, she gets value at most $\frac{4}{3} - \delta$. The agent then has a budget strictly less than $\frac{1}{6}$ left, with which she can buy only small items. This gives additional value of strictly less than $\frac{1}{6} \cdot 6\delta = \delta$, and hence the total value purchased is strictly less than $\frac{4}{3}$.
\end{proof}

{The proof of Theorem~\ref{thm:BK17a} follows from Claims~\ref{claim:1/3}, \ref{claim:2/3} and~\ref{claim:middleThird}, together with Claims~\ref{claim:atMost2}, \ref{claim:5/6} and~\ref{claim:middle} of Lemma~\ref{lem:BK17}.}
\end{proof}

The analysis of Greedy-EFX in the proof of Theorem~\ref{thm:BK17a} is tight up to low order terms. In fact, this holds even with respect to the maximin share.

\begin{proposition}
\label{pro:BK17a}
For every $\delta > 0$, there is an instance with identical additive valuation functions in which greedy-EFX does not provide an agent more than a $\frac{3}{4} + \delta$ of her  maximin share. (In our example, $n$ is linear in $\frac{1}{\delta}$.)
\end{proposition}

\begin{proof}
For a given value of $\delta > 0$, choose $n$ sufficiently large so that $\epsilon$ in the example below satisfies $\epsilon \le \frac{4}{5}\delta$.

There are $3n-1$ items, and $\epsilon = \frac{1}{8n-3}$. (Under this choice of $\epsilon$ that the numbers in the following example work out.) The first item has value $\frac{1}{2} + \frac{\epsilon}{2}$, thereafter item values decrease by $\epsilon$ until they reach value $\frac{1}{4} + \frac{3}{4}\epsilon$ for item $2n$. All remaining $n-1$ items have value $\frac{1}{4} + \frac{3}{4}\epsilon$. Greedy-EFX will put the first $2n$ items in different bundles (each bundle will contain the pair of items $i$ and $2n+1-i$), each of value $\frac{3}{4} + \frac{5}{4}\epsilon$. Thereafter, only $n-1$ items remain, so one bundle will have a final value of $\frac{3}{4} + \frac{5}{4}\epsilon \le \frac{3}{4} + \delta$. The maximin share is~1, putting the first two items in the same bundle, and for the remaining bundles using greedy-EFX.
\end{proof}

{Proposition \ref{pro:BK17a} shows that for identical additive valuations, 
greedy-EFX does not provide an agent more than a $\frac{3}{4}$ of her  maximin share, 
although there exists an allocation that gives every agent her maximin share (as the valuations are identical).}
%\tec{do we want to keep the next paragraph?} \ufc{Commented out next paragraph. We do not want it in this paper.}
%Proposition~\ref{pro:BK17a} suggests that to analyze the gap between balanced \tec{anyprice} share and maximin share, which corresponds to the case that all agents have the same valuation function, it may be desirable to choose the best of two algorithms. One algorithm is Greedy-EFX, and the other is some version of first fit decreasing (FF-D), in which items are sorted in decreasing value, and placed in the first available bundle. A version of FF-D can be defined by setting an upper threshold (the desirable value for the maximin share) and a lower threshold (which may give flexibility in the design of the algorithm. Run FF-D, but without placing an item in a bundle that exceeded its lower threshold, and without placing an item in a bundle that would exceed by this its upper threshold. Switch to greedy-EFX once the next item does not fit in any bundle (in the above sense).  With upper and lower thresholds of~1, FF-D gives on the example in Proposition~\ref{pro:BK17a} more than 3/4, but less than 7/8 (some bundles will have 4 items). With item values 3, 3, 2, 2, 2, 2 and two agents, FF-D achieves a ratio of $\frac{6}{7}$ (but greedy-EFX is optimal, and likewise, FF-D with a lower threshold smaller than 3).

%\input{relations}
%\appendix
%\input{old-game-allocation-game}
%\input{meeting-3-12}

%\input{prelim-TOMER-version}

\section{Beyond Goods: Chores and Mixed Manna}\label{app:chores}
In this Appendix we present an extension of our APS definition to allocation instances with \textit{chores} and with \textit{mixed manna}, and briefly discuss this extension.

Consider an agent $i$ with a valuation function $v_i$ over the set $\items$ of items {that is normalized ($v_i(\emptyset) = 0$)}. 
An item $j$ is a {\em good} with respect to $v_i$ if $v_i(S \cup \{j\}) \ge v_i(S)$ for every set $S \subseteq \items \setminus \{j\}$. 
An item $j$ is a {\em bad} (also referred to as a {\em chore}) with respect to $v_i$ if $v_i(S \cup \{j\}) \le v_i(S)$ for every set $S \subseteq \items \setminus \{j\}$.  If neither all items are required to be goods nor all items are required to be bads (with respect to $v_i$), we refer to $\items$ as being a {\em mixed manna} for agent $i$ with valuation function $v_i$. 
%Mixed manna is the most general case, and 
Clearly, settings with only goods or only bads are special cases of mixed manna.

Recall that for goods we defined valuation functions to be set functions that are normalized, namely, $v_i(\emptyset) = 0$, and monotone non-decreasing. Once chores are involved,  the requirement that valuation functions are non-decreasing is removed. The value of an item or of a (non-empty) bundle might be negative.  

Agent $i$ has entitlement $b_i$, where $0 < b_i \le 1$. Intuitively, when all items are {\em goods}, the entitlement represents what fraction of the grand bundle of goods the agent is entitled to receive. 
When items are {\em chores}, the entitlement of the agent represents what fraction of the chores {should be assigned to} the agent. {As items are indivisible, some deviations from these exact fractions might be necessary.} 
%must undertake upon herself. 

%Hence our interpretation of entitlement (when chores are involved) does map higher entitlements to receiving bundles of higher value, and should not be applied in cases in which higher entitlement is expected to correspond to receiving bundles of higher value.

We now extend the definition of the AnyPrice share to the setting of mixed manna. Recall that two definitions were presented for the APS in case of goods. We first present Definition~\ref{def:APSmixed} which is the natural extension of Definition~\ref{def:anyprice-bundles}, replacing the inequality in the last constraint by an equality. {That is, instead of requiring that for every item $j\in \items$ it holds that $\sum_{T: j\in T} \lambda_T \geq b_i$, we now require that 
$\sum_{T: j\in T} \lambda_T = b_i$.}

% \mbc{remove the def and present the change in details (point 4 holds as equality.) Explain why it is WLOG that we have equality for goods (and for chores). MB: added above, now I think we can remove this definition. }

\begin{definition}
\label{def:APSmixed} 	
	Consider a setting in which agent $i$ with valuation $v_i$ has entitlement $b_i$ to a set of indivisible items $\items$.
	The \emph{AnyPrice share (APS)} of $i$, denoted by $\anypricei$, 
	is the maximum value $z$ she can get 
	by coming up with non-negative
	weights $\{\lambda_T\}_{T\subseteq \items}$ that total to $1$ (a distribution over sets), 
	such that any set $T$ of value below $z$ has a weight of zero, and
	every item appears in sets of total weight exactly $b_i$:

	$$ \anypricei = \max_{z} z $$ \\ subject to the following set of constraints being feasible for $z$: 
	\begin{itemize}
		\item $\sum_{T\subseteq \items} \lambda_T=1$
		\item  $\lambda_T\geq 0$  for every bundle  ${T\subseteq \items}$
		\item  $\lambda_T= 0$ for every bundle ${T\subseteq \items}$ s.t. $v_i(T)<z$ % 
		\item $\sum_{T: j\in T} \lambda_T = b_i$ for every item $j\in \items$
	\end{itemize}
\end{definition}
Recall that for the case of goods the APS is at least as large as the MMS. This property also holds for  mixed manna (and thus for chores). 
This is because Definition~\ref{def:APSmixed} defines the APS as a solution to a maximization problem that is a relaxed (fractional) version of the (integral) maximization problem defining the MMS.

For mixed manna, the APS may be either positive or negative (or zero). We present some observations that are useful in determining its sign and some of its properties.

\begin{proposition}
\label{pro:positiveManna}
{Let $\items$ be mixed manna for agent $i$ with valuation function $v_i$.} If $v_i(\items) \ge 0$ then $\anypricei \ge 0$.
\end{proposition}

\begin{proof}
{The LP of Definition~\ref{def:APSmixed} is feasible for $z=0$. This can be seen by} taking bundle $\items$ with coefficient $b_i$, and the empty bundle $\emptyset$ with coefficient $1 - b_i$.
\end{proof}

Fixing a valuation function $v_i$, if there are only goods, then the APS is  a non-decreasing  function of the entitlement. Likewise, if there are only chores, then the APS is a non-increasing  function of the entitlement. However, for the case of mixed manna, the APS need not be a monotone function of the entitlement. Consider for example an instance (that is non-additive) with $\items = \{a,b\}$, where $v_i(a), v_i(b) > 0$, and $v_i(\items) < 0$. Then if $b_i < \frac{1}{2}$ the APS is zero (the LP in Definition~\ref{def:APSmixed} can be supported on the three bundles $\{a\}, \{b\}, \emptyset$), if $b_i  = \frac{1}{2}$ the APS is positive (the support is $\{a\}, \{b\}$), and if $b_i > \frac{1}{2}$ the APS is negative (the bundle $\items$ must be in the support of the LP solution). 
%Quasiconcavity (Remark~\ref{rem:quasiconcave}) implies that if the APS is negative for entitlement $b_i$, it remains negative for every entitlement $b_j > b_i$.
While for mixed manna the APS is not necessarily monotone in the entitlement, we next show that it is always a quasiconcave function of the entitlement.

\begin{remark}
\label{rem:quasiconcave}
A proof similar to that of Proposition~\ref{pro:positiveManna} shows that for a fixed valuation function $v_i$, the APS is a {\em quasiconcave} function of the entitlement $b_i$. Namely, for every $0 \le b_1 < b_2 < b_3 \le 1$ it holds that $\anyprice{b_2}{v_i}{\items} \ge \min\left[\anyprice{b_1}{v_i}{\items},\anyprice{b_3}{v_i}{\items}\right]$. This follows by considering optimal solutions,  $\{\lambda_T^1\}_T$ for $b_1$ and $\{\lambda_T^3\}_T$ for $b_3$, in Definition~\ref{def:APSmixed}. Then $\{\frac{b_1 \cdot \lambda_T^1+b_3 \cdot\lambda_T^3}{b_1+b_3}\}_T$  is a feasible solution for $b_2$, and the bundle of minimum value in its support is in the support of at least one of the two optimal solutions (for $b_1$ and $b_3$).
\end{remark}
Quasiconcavity implies that for every valuation function $v_i$ there is a threshold entitlement $\tau$, such that $\anypricebi$ is a weakly increasing function of $b_i$ for $b_i\in [0,\tau]$, and weakly decreasing for $b_i \in [\tau,1]$. In particular, if  the APS is negative for entitlement $b_i$, it remains negative for every entitlement $b_i^+ > b_i$.

\begin{proposition}
\label{pro:worstItem}  
{For any {normalized valuation $v_i$ and any} entitlement $b_i$,} if there are only chores, then $\anypricei \le \min_{j\in \items} v_i(j)$.
\end{proposition}

\begin{proof}
Let $z = \anypricei$. Then the LP of Definition~\ref{def:APSmixed} is feasible for this value of $z$. Let $j$ be the item of smallest (most negative) value under $v_i$. Then in the feasible solution of the LP, there must be some bundle $T$ that contains $j$. Hence $v_i(T) \ge z$. As there are only chores, $v_i(T) \le v_i(j)$. 
\end{proof}

%\ufe{We now turn to consider additive valuations.}

\begin{proposition}
\label{pro:APSversusProp}
{For any entitlement $b_i$,} if $v_i$ is an additive valuation (allowing mixed manna), then the AnyPrice share cannot be larger than the proportional share: $\anypricei \le b_i \cdot v_i(\items)$.
\end{proposition}

\begin{proof}
Let $z$ be a value for which the LP of Definition~\ref{def:APSmixed} is feasible. Then $z$ is the smallest value of a bundle in the support of the underlying distribution, whereas the proportional value is the average value (due to additivity).
\end{proof}

If all items are chores (even without assuming additivity), then the value of every non-empty bundle is negative.
It is convenient to view the absolute values of these negative values as positive {\em \disutilities}. Thus we replace the non-positive valuation function $v_i$ by the non-negative  \disutility\ function $c_i$, where for every $S \subseteq \items$ we have $c_i(S) =-v_i(S)$. With respect to \disutilities, the APS is positive, and the agents wish to receive bundles of small \disutility. There are allocation instances in which every allocation gives some agent a bundle of \disutility\ larger than her APS. (This follows from a similar result concerning MMS for chores and agents with equal entitlements~\citep*{ARSW17}, because with equal entitlements, the APS \disutility\ is never larger than the MMS \disutility.)  Consequently, we wish to find allocations that give every agent a bundle of \disutility\ no more than $\rho$ times her APS, with $\rho$ being as small as possible. {For $\rho = 2$, this can be achieved for additive valuations, using standard approaches.}

\begin{proposition}
\label{pro:BoBWchores}
Consider an allocation instance with $n$ agents with valuations $(v_1,v_2,\ldots,v_n)$ and 
arbitrary entitlements $(b_1,b_2,\ldots,b_n)$ for a set $\items$ of indivisible chores. If all \disutility\ functions are additive, then there is an allocation in which every agent $i$ gets a bundle of \disutility\ at most $2\cdot \anypricei$.
\end{proposition}

\begin{proof}
Consider a fractional allocation in which every agent $i$ (with entitlement $b_i$) gets a fraction $b_i$ of every item $j$. This fractional allocation gives every agent her proportional share. It is well known that every fractional allocation can be rounded to give an integral allocation in which every agent gets a bundle whose \disutility\ is the same as the \disutility\ of her original fractional bundle, up to the \disutility\ of one item~\citep*{LST90}. As her fractional \disutility\ is at most her APS \disutility (Proposition~\ref{pro:APSversusProp}, after transforming values to \disutilities) and every item has value at most the APS (Proposition~\ref{pro:worstItem}, after transforming values to \disutilities), the proposition follows.
\end{proof}

%\mbc{talk about duality of the prior definition (maybe write it) and that we now present it for the special case of chores.  The complicating issue is how and when to have negative budgets. See new paragraph below.}

For completeness, we present the price based definition of the APS for chores (in analogy to Definition~\ref{def:anyprice-prices} for goods).
{The price based definition is derived by considering the dual of the fractional MMS definition. 
We have shown the  equivalence of the two definitions for the case of goods in Proposition~\ref{prop:anyprice-eqe}.  Similar arguments (which are omitted) show the equivalence in the case of chores.}

%\mbc{define cost function $c_i=-v_i$ and use it below. }

\begin{definition}
\label{def:APSchores}
	Consider a setting in which agent $i$ with \disutility\ function $c_i$ has entitlement $b_i$ to a set of indivisible chores $\items$.
	The \emph{AnyPrice share (APS)} of agent $i$, denoted by $\anypriceic$, is the \disutility\ she can limit herself to whenever the chores in $\items$ are adversarially priced with a total price of $1$, and she picks her bundle of chores among those bundles of total price at least $b_i$: 
	$$\anypriceic = \max_{(p_1,p_2,\ldots,p_m)\in \prices}\ \ \min_{S\subseteq \items} \left\{c_i(S) \Big | \sum_{j\in S} p_j\geq b_i\right\}$$
	
%	When $\items$ and $v_i$ are clear from context we denote 	 the APS share of an agent $i$ with entitlement $b_i$ by $\anypricebi$, instead of $\anypricei$.	
\end{definition}

Observe that both in Definition~\ref{def:anyprice-prices} (for goods) and in Definition~\ref{def:APSchores} (for chores), all prices are non-negative. The difference between the definitions is that for the case of goods, the feasible bundles for an agent $i$ are those priced at most $b_i$, whereas for chores the feasible bundles are those priced at least $b_i$. 
We note that the APS in Definition~\ref{def:APSchores} measures dis-utility, and thus has a positive value which equals  to minus the value of Definition~\ref{def:APSmixed}.

For allocation instances that involve a mixture of goods and chores, there are instances in which agents have additive valuations, equal entitlements,  the MMS of every agent is strictly positive, whereas in every allocation (that allocates all items) some agent receives a bundle of value at most~0 \citep*{KMT2020}. As the APS is at least as large as the MMS, it follows that for mixed manna it is not always possible to find an allocation that gives every agent a positive fraction of her APS.

\end{document}